\documentclass[11pt]{article}
\usepackage{amsmath,amssymb,subfigure,epsfig,bbm,stmaryrd,MnSymbol,mathrsfs,hyperref,framed}
\DeclareMathOperator*\essinf{ess\,inf}

\usepackage{graphicx}
\usepackage{cite}
\usepackage[table]{xcolor}
\usepackage{booktabs}
\usepackage{multirow}
\usepackage{amsmath,amssymb,mathrsfs}
\usepackage{pstricks}
\usepackage[margin=3cm]{geometry}
\title{\textbf{High-order regularized regression in Electrical Impedance Tomography}}

\author{Nick Polydorides\footnotemark[2]\
\and Alireza Aghasi\footnotemark[3]\ \footnotemark[4] \and Eric.
L. Miller\footnotemark[3]\ \footnotemark[5]}

\begin{document}

\newtheorem{assumption}{Assumption}
\newtheorem{theorem}{Theorem}[section]
\newtheorem{lemma}[theorem]{Lemma}
\newtheorem{proposition}[theorem]{Proposition}
\newtheorem{corollary}[theorem]{Corollary}

\newenvironment{proof}[1][Proof]{\begin{trivlist}
\item[\hskip \labelsep {\bfseries #1}]}{\end{trivlist}}
\newenvironment{definition}[1][Definition]{\begin{trivlist}
\item[\hskip \labelsep {\bfseries #1}]}{\end{trivlist}}
\newenvironment{example}[1][Example]{\begin{trivlist}
\item[\hskip \labelsep {\bfseries #1}]}{\end{trivlist}}
\newenvironment{remark}[1][Remark]{\begin{trivlist}
\item[\hskip \labelsep {\bfseries #1}]}{\end{trivlist}}

\newcommand{\qed}{\nobreak \ifvmode \relax \else
      \ifdim\lastskip<1.5em \hskip-\lastskip
      \hskip1.5em plus0em minus0.5em \fi \nobreak
      \vrule height0.75em width0.5em depth0.25em\fi}

\maketitle

\renewcommand{\thefootnote}{\fnsymbol{footnote}}
\footnotetext[2]{Energy, Environment and Water Research Center,
The Cyprus Institute, Cyprus and MIT Energy Initiative, Cambridge,
MA. (\href{mailto: nickpld@mit.edu}{nickpld@mit.edu}).
Corresponding author.} \footnotetext[3]{Department of Electrical
and Computer Engineering, Tufts University, Halligan Hall,
Medford, MA.} \footnotetext[4]{Email:
\thanks{(\href{alireza.aghasi@tufts.edu}{alireza.aghasi@tufts.edu).}}}
\footnotetext[5]{Email:
\thanks{(\href{elmiller@ece.tufts.edu}{elmiller@ece.tufts.edu).}}}

\renewcommand{\thefootnote}{\arabic{footnote}}

\newcommand{\slugmaster}{%
\slugger{MMedia}{xxxx}{xx}{x}{x--x}}

\begin{abstract}
We present a novel approach for the inverse problem in electrical
impedance tomography based on regularized quadratic regression.
Our contribution introduces a new formulation for the forward
model in the form of a nonlinear integral transform, that maps
changes in the electrical properties of a domain to their
respective variations in boundary data. Using perturbation theory
the transform is approximated to yield a high-order misfit function which is then used to derive a
regularized inverse problem. In particular, we consider the
nonlinear problem to second-order accuracy, hence our
approximation method improves upon the local linearization of the
forward mapping. The inverse problem is approached using Newton's
iterative algorithm and results from simulated experiments are
presented. With a moderate increase in computational complexity,
the method yields superior results compared to those of
regularized linear regression and can be implemented to address
the nonlinear inverse problem.

\end{abstract}

\textbf{keywords:} Impedance tomography transform, quadratic
regression, Newton's method \\

\pagestyle{myheadings}

\section{Introduction}

In Electrical Impedance Tomography (EIT) voltage measurements
captured at the boundary of a conductive domain are used to
estimate the spatial distribution of its electrical properties.
The technique has numerous applications in exploration geophysics
\cite{Zhdanov}, environmental monitoring and hydrogeophysics
\cite{Binley}, \cite{Kemna}, biomedical imaging \cite{Holder},
industrial process monitoring \cite{mccann}, archaeological site
assessment \cite{papadopoulos} and non-destructive testing of
materials \cite{Santosa_Vogelius}. Owing to its many practical uses
and intriguing mathematics, EIT has seen numerous
theoretical and computational developments, e.g. the
the chapter expositions in \cite{adgablio}, \cite{Kirsch} and
\cite{Kaipio_Somer_book}. Among its fundamental challenges remain
the nonlinearity and ill-posedness of the inverse problem, which
inevitably compromise the spatial resolution of the reconstructed
images.  From the mathematical prospective, this inverse boundary
value problem, formalized by the seminal publication of
Cald\'{e}ron \cite{Calderon}, presents a number of implications on
the existence, uniqueness and numerical stability of the
solution \cite{Borcea}, \cite{adgablio}.  Although the issues of
existence and uniqueness can be eradicated under some mild
assumptions, see for example \cite{SylvUhlmann} for isotropic
conductivity fields, the instability causes the problem to be
extremely sensitive to inaccuracies and small errors in the data.
To alleviate the ill-posedness one usually resorts in implementing
some type of regularization strategy that stabilizes the solution
\cite{engl}. Based on prior information about the unknown
electrical parameters and/or the noise statistics in the
measurements, regularization schemes are applied in order to
stabilize the reconstructions. In the typical variational
framework for example, regularization methods are often expressed
as additive penalty terms augmenting the associated data misfit
function, essentially biasing the solution away from features that
are inconsistent with the available a priori information. In this
sense, \cite{johansen} examines the case of Tikhonov
regularization in the context of nonlinear system identification
emphasizing the bias-variance trade-off on the solution.

The nonlinearity inevitably increases the complexity of the problem, as the data misfit function has several local minima, and hence one is faced with the challenge of locating the solution that corresponds to the global minimum.  Aside a few notable exceptions, like the d-bar method \cite{Mueller} and
the factorization method \cite{KirschHanke}, algorithms that treat
the nonlinearity are essentially Newton-type iterative solvers, such as the
often used Gauss--Newton (GN) method, that implement local
linearization and regularization, essentially exploiting the
Fr\'{e}chet differentiability of the analytic forward operator, to
yield at each iteration a quadratic error function with respect to the unknown
parameters \cite{noser}, \cite{Oldenburg}. Starting from a
feasible guess and, in some cases, an estimate of the noise level in the data, one
applies a number of linearization--regularization cycles until a
convergence is reached in the sense of the discrepancy principle. Analysis on the convergence rates of the GN
algorithm and quasi-Newton variants for high-dimensional problems
can be found in \cite{Bakushinsky}, \cite{engl} and
\cite{Kaltenbacher}, and \cite{Haber}. These results state that
convergence is not guaranteed unless a stable Newton direction,
descent in the usual case of minimization, is computed at each
linearization point. In turn, this relies on the optimal tuning of
regularization at each iteration, indeed a delicate and
challenging task as the degree of ill-posedness may vary
significantly \cite{Lechleiter}. To rectify this problem and aid
convergence line search algorithms can be used, that scale
optimally the solution increment in the descent Newton direction
\cite{Bertsekas}, as indeed trust-region methods \cite{Conn}
although more computationally complex. An additional important
complication may arise when the typically-neglected linearization
error is significantly large, invariably when the linearization
point is `not close enough' to the true solution
\cite{Polydorides_linear}. This implies that a component of the
linearized data should not be considered in the fitting process,
since local linearization approximation is accurate in a rather
narrow trust region, and hence to cope with the lack of this
information at each iteration one seeks to recover a `small'
perturbation of the parameters. With this as background, the work in this paper focusses on the following contributions: (i) A nonlinear integral transform as a forward
model that maps arbitrarily large, bounded changes in electrical
properties to changes in boundary observations. Effectively, this
replaces the linear approximation involving the Jacobian of the
forward mapping \cite{blnpab}. The transform has a closed form and
admits a numerical approximation using the finite element method.
(ii) Exploiting the new model, a high-order misfit function is
formulated for the inverse problem in the context of regularized
regression. Numerical experiments on the resulting inverse problem
have yield solutions with small image errors and adequate spatial
resolution.

Higher-order derivatives are thus seldom used in inversion schemes
since the increase in convergence rates may not compensate
adequately for the computational effort required in computing the
derivatives, in particular when high-dimensional discrete models
with tensor parameters are concerned. Moreover, if the data misfit
residual is small then the error contribution of the higher-order
terms becomes negligibly small. The majority of inversion
algorithms, as indeed the general theory, for nonlinear inverse
problems utilize merely a first-order approximation of the
underlying model \cite{engl}. A notable exception is the
second-order method for nonlinear, highly ill-posed parameter
identification problems in some classical partial differential
equations, such as Helmholtz, diffusion and Sturm-Liouville
\cite{hettlich}. In \cite{hettlich}, the authors propose iterative
predictor-corrector schemes encompassing Tikhonov regularization
that approximate the second-order solution without solving
a quadratic equation. Using this computationally efficient
framework, they report on advantages in the final reconstructions
and significant improvements in the number of iterations required
for convergence.

As we develop our methodology we address mainly the EIT problem
with complex isotropic admittivity and the complete electrode
boundary conditions \cite{Somersalo_unique}. However, our
derivations are not constrained by isotropic or complex property
assumptions and thus can be easily shown to hold true for the
similar problems of Electrical Resistance and Capacitance
Tomography (ERT/ECT) with purely real coefficients in scalar or
tensor field material properties \cite{adgablio}, \cite{Pain}.
Moreover, we show that the form of the new forward model remains
unchanged with the governing elliptic differential equation is
addressed in the context of more generalized boundary conditions
that resemble more simplistic electrode models conventionally
encountered in the geophysical setting
\cite{Binley},\cite{Miller}.

\subsection{Notation and paper organization}

Consider a simply connected domain $B \subset \mathbb{R}^d$,
$d=2,3$ with Lipschitz smooth boundary $\partial B$ and a space
depended isotropic admittivity function
$\gamma(\mathbf{x},\omega): B \rightarrow \mathbb{C}$. At an
angular frequency $\omega \geq 0$, the admittivity can be
expressed as $$\gamma(\mathbf{x},\omega) = \sigma(\mathbf{x}) + i
\omega \epsilon(\mathbf{x}),$$ where $\infty>C_1>\sigma >c_1>0$
and $\infty>C_2 >\epsilon\geq 0$ denote the domain's electrical
conductivity and permittivity respectively for some positive
bounding constants $C_1,\,C_2,\, c$. If there are no charges or
sources in the interior of $B$ and the angular frequency of the
applied currents is small enough, then Maxwell's equations
describing the electromagnetic fields in the interior of the
domain reduce to the elliptic equation
\begin{equation}\label{pdei}
\nabla \cdot [ \gamma(\mathbf{x},\omega) \nabla u(\mathbf{x},\omega)] = 0, \quad  \mathbf{x} \in B,
\end{equation}
where $u$ denotes the scalar electric potential function.
Measuring the potential at the accessible parts of the boundary of
the domain through a finite number of sensors yields a set of
observations $\zeta$ that are likely to suffer from some type of
noise and measurement imprecision $\eta$. We will assume EIT
systems equipped with $L$ electrodes exciting the domain with a
sequence of currents $I\doteq(I^1,\ldots,I^q)$, with $I^i \doteq
(I_1,\ldots,I_L)$ all fixed at frequency $\omega$. In such a case
$\zeta$ is typically a linear combination of the electrode
potentials $U(I^i)\doteq(U_1,\ldots,U_L)$, for $i=1,\ldots,q$ at
the various current patterns. For an applied current pattern $I$
we associate an electric potential field $u$ in $\overline B$, and
an array of electrode potentials $U$ at $\partial B$. When
required by the context we shall denote their dependence on
admittivity and applied current as $u(\gamma)$ and $u(I)$, or both
as $u(\gamma,I)$; and respectively $U(\gamma)$, $U(I)$ and
$U(\gamma,I)$. The first and second partial derivatives of $u$
with respect to $\gamma$ will be denoted by $\partial_\gamma u$
and $\partial_{\gamma \gamma} u$, a notation adopted for both
continuous and discrete spatial functions, while for matrices and
vectors the differentiation is to be considered element-wise. The
position in $\overline B$ is specified by the vector $\mathbf{x}
\in \mathbb{R}^d$, while the outward unit normal vector at the
boundary is denoted $\mathbf{n}$. Matrices and vector fields are
expressed in bold capital letters while vectors and scalar fields
in small case regular. For a matrix $\bf A$, $a_j$ will denote the
$j$th row, $\mathbf{A}_{i,j}$ its $(i,j)$th element and $\bf A'$
its transpose. For a vector $v$, $v_i$ is the $i$th element
and $\bar v$ is the complex conjugate. The spaces of real and
complex numbers are given by $\mathbb{R}$ and $\mathbb{C}$, while
we use $\Re\{c\}$ to express the real component of the complex
argument $c$.

The paper is organized as follows: We begin with a brief review of
the the EIT model equations and associated preliminary concepts
and then proceed to formulate the inverse problem commenting on
existing algorithms the address the problem through local linearization.
The next section is devoted to the derivation of the nonlinear
integral transform under the complete electrode model and its
generalization to the Poisson's equation with mixed boundary
conditions. Further on we consider the high-order regularized
regression problem and approximate the nonlinear system as a
quadratic operator equation. Using the finite element we obtain a
numerical approximation and subsequently implement Newton's
algorithm to solve the problem. Finally, we present numerical
results from simulated studies that demonstrate the advantages of
the proposed methodology and we end the paper with the conclusions
section.

\section{EIT model equations and preliminaries}

The complete electrode model in electrical impedance tomography is
derived from Maxwell's time-harmonic equations at the quasi-static
limit and describes the electric potential field in the closure of
a conductive domain $B$ with known electrical properties $\gamma$
and impressed boundary excitation conditions. The model has been
extensively reviewed and analyzed in several publications,
including \cite{Somersalo_unique} where the authors prove the
existence and uniqueness of the solution, under some continuity assumptions on the interior admittivity. With reference to figure
\ref{figmain}, assuming no charges or current sources in the
interior of $B$, when a current is applied at the boundary, the
electric potential $u$ satisfies the elliptic partial differential
equation (\ref{pdei}). The applied current, inducing this field,
is expressed by the Neumann boundary conditions
\begin{eqnarray}\label{boundc}
\int_{e_\ell} \mathrm{d}s \; & \gamma(\mathbf{x},\omega) \nabla u(\mathbf{x}) \cdot \mathbf{n}  = & I_\ell, \quad \mathbf{x} \in \Gamma_{e_\ell}, \; \ell=1,\ldots,L ,\\
& \gamma (\mathbf{x},\omega) \nabla u(\mathbf{x}) \cdot \mathbf{n}  = & 0, \quad \; \mathbf{x} \in \partial B \setminus \Gamma_e,
\end{eqnarray}
where $\Gamma_e = \bigcup_{\ell=1}^L  \Gamma_{e_\ell}$. An
accurate model of the electrodes is critical when comparing
experimental measurements to synthetic model predictions. In
effect, the voltage measurement recorded at the $\ell$th electrode
with contact impedance $z_\ell$ is given by the Robin condition
\begin{equation}\label{boundv}
U_\ell = u(\mathbf{x}) + z_\ell \gamma(\mathbf{x},\omega) \nabla u(\mathbf{x}) \cdot \mathbf{n}, \quad \mathbf{x} \in \Gamma_{e_\ell}, \; \ell=1,\ldots,L ,
\end{equation}
assuming that the characteristic function of the contact impedance
is uniform on each electrode and $\Re\{z_\ell\} > 0$. The model
admits a unique solution $(u,U)$ upon enforcing the charge
conservation principle on the applied currents and a choice of
ground is made. Maintaining the conventional notation of
\cite{Somersalo_unique}, \cite{Kirsch}, and \cite{Borcea} we these
constraints imply
\begin{equation}\label{icond}
\sum_{\ell=1}^L I_\ell = 0, \quad \text{and} \quad \int_{\partial B} \mathrm{d}s\, u = 0,
\end{equation}
where the applied currents should sum up to zero and the induced
potential should have a vanishing mean on the boundary. For the
so-called forward or direct problem (\ref{pdei})-(\ref{icond}) we
adopt the following essential assumptions \cite{Kirsch},
\cite{Kaipio_Somer_book}.
\begin{assumption}\label{ass}
(a) The domain $B$ is simply connected with boundary $\partial B$ at least Lipschitz continuous.\\
(b) The electrical admittivity $\gamma \in L^\infty(\overline B)$ with $\essinf(\gamma) > c_1 >0$.\\
(c) The potential field $u \in H^1_o(B) = \{ u \in H^1(B) : \int_{\partial B} \mathrm{d}s\, u =0\}$\\
(d) The applied currents $I$ and measured voltages $\zeta$ belong in the Hilbert spaces of the $L$, and respectively $m$ dimensional complex vectors $\mathbb{C}^L$ and $\mathbb{C}^m$, where $m \geq L$.
\end{assumption}

\begin{figure}
\begin{center}
\epsfig{file=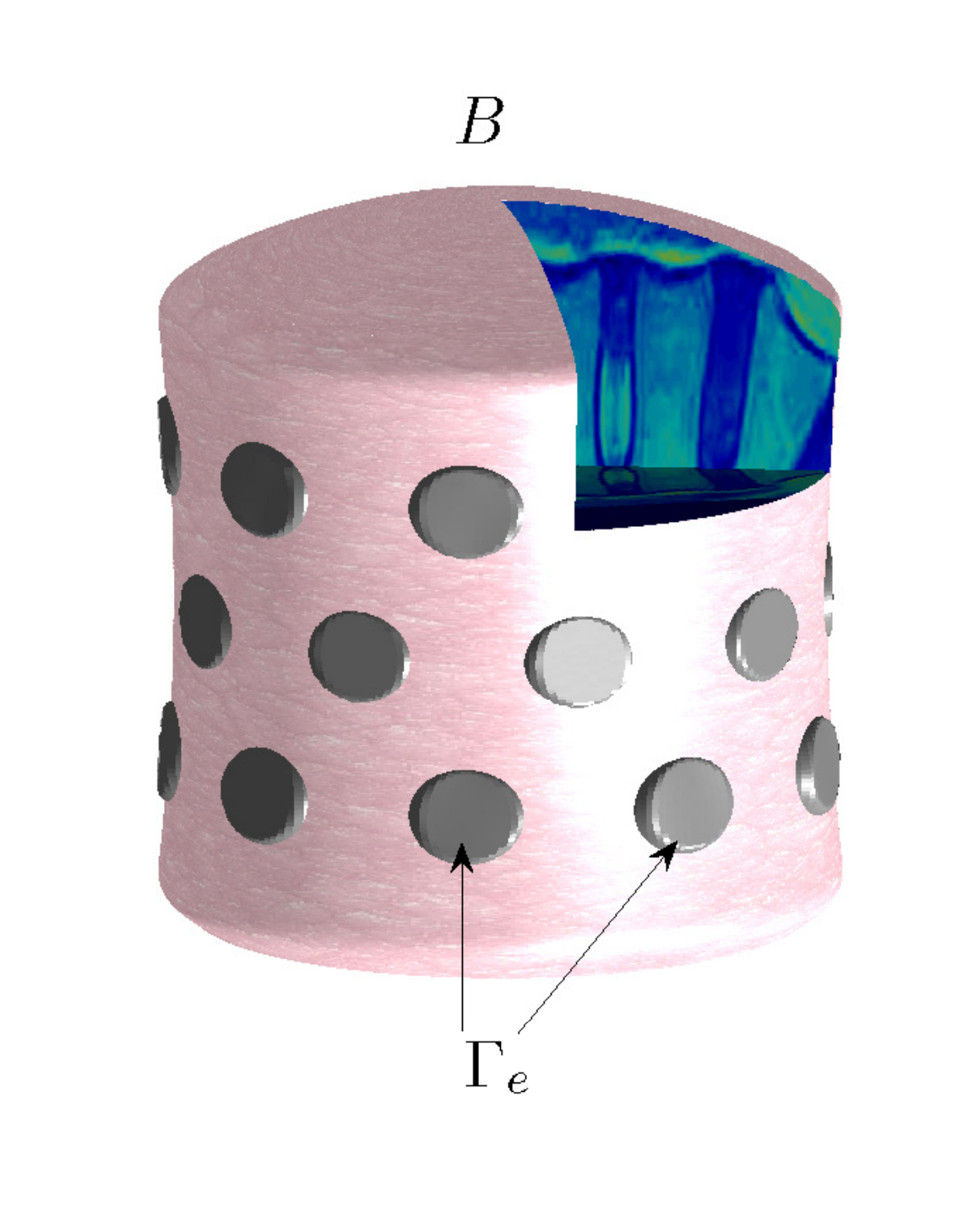,height=10cm,width=8cm}\\\vspace{-1cm}
\end{center}
\caption{The domain under consideration $B$ with $L$ round surface electrodes $\Gamma_{e_\ell}$ attached at the boundary $\Gamma_e$.}\label{figmain}
\end{figure}

We will often refer to the solution $(u,U) \in H_o^1(B) \oplus
\mathbb{C}^L$ as the \emph{direct} solution, and to the problem
(\ref{pdei})-(\ref{icond}) as the \emph{direct problem}. Pertinent to
this model is the adjoint forward problem \cite{Miller}. Consider the direct solution under a
pair drive current pattern $I^d$ with positive and negative
polarity applied at electrodes $e_p$ and $e_n$ respectively.
Moreover, let the $k$'th boundary measurement be of the form
$$
\zeta_k = U_{e_{p'}} - U_{e_{n'}}, \quad p', n' \in \{1,\ldots,L\}, \quad k=1,\ldots,m
$$
for a pair of electrodes $e_{p'}$ and $e_{n'}$. In the practical
setting of EIT or ERT, see for example the applications discussed
in \cite{Holder}, \cite{Gunther}, \cite{Kemna} and
\cite{vauhkonen}, instead of measuring the electrode potentials,
it is usual to measure the potential between adjacent electrodes.
When captured systematically, this differential type of
measurement yields $m \geq L$ linearly independent data. Based on
this measurement definition, the adjoint field solution $(v,V) \in
H_o^1(B) \oplus \mathbb{C}^L$ satisfies the equations
\begin{eqnarray}\label{adj}
\nabla \cdot [\overline{\gamma} (\mathbf{x},\omega) \nabla v (\mathbf{x})] & = & 0 \quad  \mathbf{x} \in B,\\ \label{adj1}
\overline{\gamma} (\mathbf{x},\omega)  \nabla v (\mathbf{x}) \cdot \mathbf{n} & = & 0 \quad \mathbf{x} \in \partial B \setminus \Gamma_e,\\ \label{adj2}
\int_{e_\ell} \mathrm{d}s \,\overline{\gamma} (\mathbf{x},\omega)  \nabla v (\mathbf{x}) \cdot \mathbf{n}  & = & I^m_\ell \quad \mathbf{x} \in \Gamma_{e_\ell},\\ \label{adj3}
v (\mathbf{x}) + z_\ell \overline{\gamma} (\mathbf{x},\omega)  \nabla v (\mathbf{x}) \cdot \mathbf{n} & = & V_\ell, \quad \mathbf{x} \in \Gamma_{e_\ell}, \; \ell = 1, \ldots, L
\end{eqnarray}
where $\overline{\gamma} (\mathbf{x},\omega)=\gamma
(\mathbf{x},-\omega)$ is the conjugated admittivity, and $I^{m}
\in \mathbb{C}^L$ is the adjoint current pattern whose $\ell$'th
entry equals to $ I^m_{\ell'} = \overline{I^d_{\ell}}$,  if
$\ell=e_p$ or $\ell=e_n$ and zero otherwise. The uniqueness of the
adjoint solution is subject to the constraints of the type in
(\ref{icond}).

\subsection{Green's reciprocity}
In what follows, we make reference to the \emph{reciprocity
principle}. Originally derived from Maxwell's laws of
electromagnetics, Maxwell's reciprocity principle has an analogue
for irrotational fields known as Green's reciprocity. In the
context of the impedance experiment it states, that if a current
intensity $I$ is applied at the boundary of a closed domain
between two electrodes, say $P_1\doteq(e_p, e_n)$, then the
potential measured at the boundary through another pair of
electrodes $P_2\doteq(e_{p'},e_{n'})$ will be equal to the
potential measured at $P_1$ if the same current is applied to
$P_2$. Impedance data acquisition instruments rely on this
principle to avoid making redundant, i.e. linearly dependent,
measurements. To see this consider a linear conductive medium $B$
whose admittivity $\gamma$ has a nonzero imaginary component at
the operating non-resonant frequency $\omega$. Suppose we apply a
time-harmonic electric current $I^d$ at the boundary of the domain
through electrodes $P_1$,
$$
I^d(\mathbf{x},t) = \mathbf{J}(\mathbf{x},\omega)\, e^{i \omega t}, \qquad \mathbf{x} \in \partial B,
$$
where $\mathbf{J}$ is the current density field. From Maxwell's
laws the electric and magnetic fields $\mathbf{E}$, and
$\mathbf{H}$, within the domain satisfy
$$
\nabla \times \mathbf{H}(\mathbf{x}) = \gamma(\mathbf{x}) \mathbf{E}(\mathbf{x}), \qquad \mathbf{x} \in \overline{B}.
$$
As the domain is simply connected, using $\mathbf{E}(\mathbf{x}) =
- \nabla u(\mathbf{x})$ and $\nabla \times \mathbf{H}(\mathbf{x})
= \mathbf{J}(\mathbf{x})$ reduces to Ohm's law
\begin{equation}\label{ohm}
\mathbf{J}(\mathbf{x}) =- \gamma (\mathbf{x}) \nabla u(\mathbf{x}).
\end{equation}
Let the magnitude of the applied current be equal to $I_o$, such
that $|I^d_{e_p}|= I_o$ and $I^d_{e_n} = - I^d_{e_p}$. Similarly,
allow $I^m$ a different current pattern of unit magnitude applied
through a different pair of boundary electrodes, say $P_2$,
inducing a new electric potential field. We denote the two fields
as $u(I^d)$ and $u(I^m)$ to emphasize their dependance on the
excitation currents. Taking the normal component of the vector
fields in (\ref{ohm}) for $I^d$, multiplying with $u(I^m)$ and
integrating over the boundary yields
$$
\int_{\partial B}\mathrm{d}s\; u(I^m) \mathbf{J}(I^d) \cdot \mathbf{n} = - \int_{\partial B} \mathrm{d}s\; \gamma \, u(I^m)  \nabla u(I^d) \cdot \mathbf{n}.
$$
At $\mathbf{x} \in \partial B$, let $j(\mathbf{x}) =
\mathbf{J}(\mathbf{x}) \cdot \mathbf{n}$  be the normal component
of the boundary current density field, then combining with
conditions (\ref{boundc}) and (\ref{boundv}) the left hand size of
the equation above reduces to
\begin{eqnarray*}
\int_{\partial B}\mathrm{d}s\; u(I^m) j(I^d) & = & \int_{\Gamma_e}\mathrm{d}s\; u(I^m) j(I^d)\\
                                              & = &  \sum_{\ell=1}^L I^d_\ell \int_{\Gamma_{e_\ell}}\mathrm{d}s \bigl ( U_\ell(I^m) - z_\ell j(I^m) \bigr )\\
                                              & = &  \sum_{\ell=1}^L I^d_\ell U_\ell(I^m) - \sum_{\ell=1}^L z_\ell I^d_\ell I^m_\ell\\
                                              & = & I_o \Bigl (U_{e_p}(I^m) - U_{e_n}(I^m) \Bigr )
\end{eqnarray*}
where the last simplification follows as the supports of $I^d$ and
$I^m$ are disjoint. Using the Green's first formula, the right
land side of the same equation can be developed to
$$
- \int_{\partial B} \mathrm{d}s\; \gamma\, u(I^m)  \nabla u(I^d) \cdot \mathbf{n} = - \int_B \mathrm{d}x \; \gamma \nabla u(I^m) \cdot \nabla u(I^d),
$$
and hence equating the two yields
$$
I_o \bigl (U_{e_p}(I^m) - U_{e_n}(I^m) \bigr ) = - \int_B \mathrm{d}x \; \gamma \nabla u(I^d) \cdot \nabla u(I^m).
$$
Working similarly for the adjoint field $u(I^m)$ leads to
$$
U_{e_{p'}}(I^d) - U_{e_{n'}}(I^d) = - \int_B \mathrm{d}x \; \gamma \nabla u(I^m) \cdot \nabla u(I^d),
$$
therefore for $I_o=1$ we have the standard form of Green's reciprocity theorem
\begin{equation}\label{grt}
U_{e_p}(I^m) - U_{e_n}(I^m) = U_{e_{p'}}(I^d) - U_{e_{n'}}(I^d).
\end{equation}

An alternative way to formalize this important result is via the
complete electrode admittance operator  $\mathcal{A}_{\gamma,z}:
\mathbb{C}^L \rightarrow \mathbb{C}^L$, a complex Hermitian matrix
that is the discrete equivalent to the Dirichlet-Neumann operator
encountered at the analysis of the continuum EIT model
\cite{Borcea}. For a fixed pair of $\gamma \in L^\infty(\overline
B)$ and $z \in \mathbb{R}^L$ this bounded operator maps linearly
the electrode potentials to the boundary currents inducing them,
$\mathcal{A}_{\gamma,z} U = I$. Let $\mu_d, \mu_m \in
\mathbb{R}^L$ be two vectors of zero sum
\begin{equation}\label{mus}
\mu_d(\ell) \doteq
\begin{cases} 1 & \ell = e_p\\
                        -1 & \ell=e_n \end{cases},
                                        \qquad \mu_m(\ell) \doteq \begin{cases} 1 & \ell=e_{p'}\\
                                    -1 & \ell=e_{n'}
            \end{cases},
\end{equation}
and consider the current pattern $I^d = I_o \mu_d$ where
$\sum_{\ell=1}^L I^d_\ell = 0$. By the Hermitianity of
$\mathcal{A}_{\gamma,z}$  the $k$'th measurement $\zeta_k =
U_{e_p'} - U_{e_n'} $ of the data vector $\zeta \in \mathbb{C}^m$
can now be expressed as
\begin{eqnarray*}
\zeta_k = \mu_m' U(I^d) & = & \mu_m' \, \mathcal{A}^{-1}_{\gamma,z} \,I^d\\
                     & = & I_o \mu_m' \,\mathcal{A}^{-1}_{\gamma,z} \,\mu_d\\
                    & = & I_o \mu_d' \,\mathcal{A}^{-1}_{\gamma,z} \,\mu_m\\
                    & = & \mu_d' U(I^m),
\end{eqnarray*}
thus we arrive at the principle (\ref{grt}).

\subsection{The inverse problem and its linear approximation}

The inverse problem of EIT is to reconstruct the admittivity function $\gamma \in L^\infty(\overline B)$ given the
operator $\mathcal{A}_{\gamma,z}$. Invariably, this requires
determining $\gamma$ given a finite set of linearly independent
current patters $(I^1,I^2,\ldots,I^q)$ and their respective
electrode potentials $(U^1,U^2,\ldots,U^q)$. Typically, in EIT
measurements one deals with frame(s) of (independent) data $\zeta$
that arise as linear combinations of the $U$ vectors. To address
this ill-posed problem some prior information on the data noise
$\eta$ and the (spatial) properties of $\gamma$ are needed.
To approach this problem one usually considers the nonlinear
operator equation
\begin{equation}\label{nfwd}
\zeta = \mathcal{E}(\gamma) + \eta,
\end{equation}
where $\mathcal{E}: L^\infty(\overline B) \rightarrow \mathbb{C}^m$.
A solution to this problem can be obtained by considering the regularized regression problem
\begin{equation}\label{nlp}
\gamma^* = \arg \min_\gamma \bigl \{ \bigl \|\zeta - \mathcal{E}(\gamma)  \bigr \|^2 + \mathcal{G}(\gamma) \bigr \},
\end{equation}
where $\mathcal{G}: L^\infty(\overline B) \rightarrow \mathbb{R}$
is a regularization functional. The choice of $\mathcal{G}$
depends on the a priori knowledge on $\gamma$, and it usually
takes the form of a smoothness enforcing term \cite{eidors}, an
$L1$ norm allowing for sparse solutions \cite{L1citation} or a
total variation norm that preserves large discontinuities in the
electrical properties \cite{Borsic_tv, Vogel}. As the forward
operator was proved to be analytic \cite{Borcea}, then subject to
the differentiability of $\mathcal{G}$, problem (\ref{nlp})
becomes suitable for gradient optimization methods \cite{engl}.
Linearizing $\mathcal{E}$ locally within a sphere
$S_{\gamma_p,\kappa} = \{\gamma : \|\gamma - \gamma_p\|^2 \leq
\kappa^2\}$, centered at an a priori guess-estimate $\gamma_p \in
L^\infty(\overline B)$, yields the Taylor series expansion
\begin{equation}\label{tay1}
\mathcal{E}(\gamma|S_{\gamma_p,\kappa}) = \mathcal{E}(\gamma_p) + \partial_\gamma \mathcal{E}(\gamma_p)(\gamma - \gamma_p) +   \mathcal{O} (\|\gamma-\gamma_p\|^2),
\end{equation}
where $\partial_\gamma \mathcal{E}: L^\infty(\overline B)
\rightarrow \mathbb{C}^m$ is the Fr\'{e}chet derivative of the
forward operator and $\kappa \geq 0$ can be thought to be the
Taylor series convergence radius. Truncating the series to
first-order accuracy yields the linearized approximation of
(\ref{nfwd})
\begin{equation}\label{lin}
\zeta \simeq \mathcal{E}(\gamma_p) + \partial_\gamma \mathcal{E}(\gamma_p) (\gamma - \gamma_p) + \eta, \qquad \gamma \in S_{\gamma_p,\kappa},
\end{equation}
which upon inserting into problem (\ref{nlp}) leads to the
regularized least-squares problem-- that coincides with the first
iteration of the regularized GN algorithm, for the optimal
admittivity perturbation
\begin{equation}\label{rlsq}
\delta \gamma_p^* = \arg\min_{\|\delta \gamma\|^2 \leq \kappa} \Bigl \{\bigl \|\delta \zeta - \partial_\gamma \mathcal{E}(\gamma_p) \delta \gamma \bigr \|^2 + \mathcal{G}(\delta \gamma) \Bigr\}, \quad \delta \zeta = \zeta  - \mathcal{E}(\gamma_p).
\end{equation}
Given the invertibility of the Hessian $\bigl [
\partial_\gamma \mathcal{E}(\gamma_p)' \partial_\gamma
\mathcal{E}(\gamma_p)   +  \partial_{\gamma
\gamma}\mathcal{G}(\gamma_p) \bigr ]$, implementing a GN algorithm
for $p=0,1,2\ldots$ yields a sequence of solutions
$\{\gamma_0,\gamma_1, \gamma_2, \ldots \}$ that converges to a
point in the neighborhood of $\gamma^*$, subject to the level of
noise in the data.  Analysis and numerical results on the
implementation of GN for the problem (\ref{rlsq}) can be found in
many publications and textbooks on EIT, see for example
\cite{vauhkonen}, \cite{engl}, \cite{Lechleiter} and
\cite{Kaltenbacher}.  We emphasize that this popular approach, as
well as its variants of Levenberg-Marquardt \cite{engl} and
quasi-Newton schemes \cite{Haber}, rely fundamentally on the local
linearization of the forward operator $\mathcal{E}$, and thus
yield a linear regression problem. Moreover, when $\mathcal{G}$ is
quadratic, the resulting cost-objective function to be minimized
is quadratic and thus Newton-type methods provide for speedy
analytically expressed solutions. The Noser algorithm proposed in
\cite{noser} is a typical example of this approach, where the
solution is computed after a single regularized GN iteration. Here
we propose an alternative approach that leads to high-order
regression problems. In this study we address explicitly the
quadratic case. The starting point toward this direction is the
nonlinear integral admittivity transform that we derive next.

\section{Nonlinear integral transform}

\subsection{Perturbation in power}
To derive the nonlinear transform that maps changes in admittivity
to those they cause on the observed boundary data we follow an
approach of power perturbation. The method, which is due to
Lionheart, has been developed in \cite{eidors} and
\cite{Polydorides_linear} to treat the real conductivity problem.
Here we extend it to the complex admittivity case incorporating
also the nonlinear terms arising in the perturbation analysis.
With minimal loss of generality we restrict ourselves to the case
of real contact impedance. If $\gamma$ and $u$ are smooth enough,
then applying the divergence theorem to (\ref{pdei}) for a test
function $\psi \in H_o^1(B)$ we have
\begin{equation}\label{w1}
0 = \int_B \mathrm{d}x\, \psi \nabla  \cdot \gamma \nabla u = - \int_B \mathrm{d}x \; \gamma \nabla u \cdot \nabla \psi + \int_{\partial B} \mathrm{d} s \; \psi \, \gamma \nabla u \cdot \mathbf{n}.
\end{equation}
If $\psi$ is set to satisfy the boundary conditions on the applied
currents (\ref{boundc}), the above becomes
$$
 \int_B \mathrm{d}x \; \gamma \nabla u \cdot \nabla \psi = \sum_{\ell=1}^L \int_{\Gamma_{e_\ell}} \mathrm{d} s \; (\psi - \Psi_\ell) \, \gamma \nabla u \cdot \mathbf{n} + \sum_{\ell=1}^L I_\ell \Psi_\ell ,
$$
where $\Psi \in \mathbb{C}^L$ is a test vector for the electrode
potentials. Plugging in the boundary condition on the measurements
(\ref{boundv}) yields the weak form of the forward problem
\cite{Kaipio_Somer_book}
\begin{equation}\label{weakv}
\int_B \mathrm{d}x \; \gamma \nabla u \cdot \nabla \psi + \sum_{\ell=1}^L \frac{1}{z_\ell} \int_{\Gamma_{e_\ell}} \mathrm{d} s \; (\psi - \Psi_\ell)(u - U_\ell) = \sum_{\ell=1}^L I_\ell \Psi_\ell,
\end{equation}
for all $(\psi,\Psi) \in H^1_o(B) \oplus \mathbb{C}^L$. Existence
and uniqueness of the weak (variational) solution $(u,U) \in
H^1_o(B) \oplus \mathbb{C}^L$ has been proved in
\cite{Somersalo_unique}. If $z_\ell >0$, then substituting
$\psi=\overline u$, $\Psi = \overline U$ into the weak form yields
the power conservation law
\begin{equation}\label{pcl}
\int_B \mathrm{d}x\; \gamma |\nabla u |^2 + \sum_{\ell=1}^L z_\ell \int_{\Gamma_{e_\ell}} \mathrm{d} s \; |\gamma \nabla u \cdot \mathbf{n}|^2 =\sum_{\ell=1}^L I_\ell \overline{U_\ell},
\end{equation}
which states that the power driven into the domain is either
stored as electric potential or dissipated at the contact
impedances of the electrodes. Consider now a complex perturbation
$\gamma \rightarrow \gamma + \delta \gamma$, causing $u
\rightarrow u + \delta u$ in the interior, and $U_\ell \rightarrow
U_\ell + \delta U_\ell$, $j \rightarrow j + \delta j$ at the
boundary. Recall that the normal component of the current density
field at the boundary is $j = \gamma \nabla u \cdot \mathbf{n}$,
under the new state of the model the volume integral in
(\ref{pcl}) becomes
\begin{eqnarray*}
\int_B \mathrm{d}x \; (\gamma + \delta \gamma) |\nabla(u + \delta u)|^2 & = & \int_B \mathrm{d}x\, \gamma|\nabla u|^2 + \int_B \mathrm{d}x\, \gamma \nabla u \cdot \nabla \overline{\delta u} \\
& & + \int_B \mathrm{d}x\, \gamma \nabla \delta u \cdot \nabla \overline{u}  + \int_B \mathrm{d}x\, \gamma|\nabla \delta u|^2 \\
& &  + \int_B \mathrm{d}x\, \delta \gamma |\nabla (u + \delta u)|^2.
\end{eqnarray*}
Notice that $\nabla \delta u \cdot \nabla \overline{u} =
\overline{\nabla u \cdot \nabla \overline{\delta u}}$ hence the
second and third integrals on the right sum up to $2 \int_B
\mathrm{d}x\, \gamma \; \Re\{\nabla u \cdot \nabla
\overline{\delta u}\}$. For $I$ and $z_\ell$ fixed, the
surface term in (\ref{pcl}) gbecomes
$$
\sum_{\ell=1}^L z_\ell \int_{\Gamma_{e_\ell}} \mathrm{d} s \; |j + \delta j|^2 = \sum_{\ell=1}^L z_\ell \int_{\Gamma_{e_\ell}}\mathrm{d} s \;  \bigl ( |j|^2 + j \, \bar{\delta j} + \delta j \, \bar{j} + |\delta j|^2 \bigr ),
$$
hence putting together the power conservation law for the new
state of the model and subtracting (\ref{pcl}) gives
\begin{eqnarray*}
\sum_{\ell=1}^L I_\ell \overline{\delta U_\ell} & = & \int_B \mathrm{d}x \; \gamma | \nabla \delta u|^2 + \int_\Omega \mathrm{d}x\, \gamma \nabla u \cdot \nabla \overline{\delta u} + \int_\Omega \mathrm{d}x\, \delta \gamma |\nabla (u + \delta u)|^2\\
& + & \int_B \mathrm{d}x\, \gamma \nabla \delta u \cdot \nabla \overline{u} +  \sum_{\ell=1}^L z_\ell \int_{\Gamma_{e_\ell}}\mathrm{d} s\, \delta j \, \overline{j} + \sum_{\ell=1}^L z_\ell \int_{\Gamma_{e_\ell}}\mathrm{d} s\, j\,\overline{\delta j} \\
& + &\sum_{\ell=1}^L z_\ell \int_{\Gamma_{e_\ell}}\mathrm{d} s\, |\delta j|^2.
\end{eqnarray*}
From the weak form (\ref{w1}) with $\psi=\overline{\delta u}$, the second integral above simplifies as
\begin{eqnarray*}
\int_B \mathrm{d}x\, \gamma \nabla u \cdot \nabla \overline{\delta u} & = & \int_{\partial B}\mathrm{d} s\; \overline{\delta u} \; \gamma \nabla u \cdot \mathbf{n} \\
                                                                                                                                & = & \int_{\Gamma_e}\mathrm{d} s\; \overline{\delta u} \, j\\
                                                                                                                                & = & \sum_{\ell=1}^L \int_{\Gamma_{e_\ell}}\mathrm{d} s\; (\overline{\delta U_\ell} - z_\ell \overline{\delta j}) \, j\\
                                                                                                                                & = & \sum_{\ell=1}^L I_\ell \overline{\delta U_\ell} - \sum_{\ell=1}^L z_\ell \int_{\Gamma_{e_\ell}}\mathrm{d} s\; \overline{\delta j} j,
\end{eqnarray*}
thus substituting back into the previous equation gives the perturbed power conservation law
\begin{eqnarray}\label{ppcl}
 \int_B \mathrm{d}x \; \gamma | \nabla \delta u|^2 + \int_B \mathrm{d}x \; & & \delta \gamma | \nabla (u + \delta u)|^2 + \int_B \mathrm{d}x \; \gamma \nabla \delta u \cdot \nabla \overline{u}\\ \nonumber
& & + \sum_{\ell=1}^L z_\ell \int_{\Gamma_{e_\ell}} \mathrm{d} s |\delta j|^2 + \sum_{\ell=1}^L z_\ell \int_{\Gamma_{e_\ell}} \mathrm{d} s\; \delta j \, \overline{j} = 0
\end{eqnarray}
In $B$, subtracting $\nabla \cdot \gamma \nabla u = 0$ from
$\nabla \cdot (\gamma + \delta \gamma)\nabla (u + \delta u) = 0$
gives the elliptic equation
\begin{equation}\label{iden}
\nabla \cdot [\gamma\nabla \delta u + \delta \gamma \nabla (u + \delta u)] = 0 \quad \text{in} \; B,
\end{equation}
and then applying (\ref{w1}) for $\psi= \overline{\delta u}$ yields
\begin{eqnarray*}
 \int_B \mathrm{d}x \; \gamma | \nabla \delta u|^2 & + & \int_B \mathrm{d}x \; \delta \gamma \nabla (u + \delta u) \cdot \nabla \overline{\delta u} \\
& = & \int_{\Gamma_{e_\ell}}\mathrm{d} s\;  \gamma \, \overline{\delta u}  \nabla \delta u \cdot \mathbf{n} + \int_{\Gamma_{e_\ell}}\mathrm{d} s\;  \delta \gamma \, \overline{\delta u}\nabla (u + \delta u) \bigr ) \cdot \mathbf{n}\\
& = & \int_{\Gamma_e}\mathrm{d} s\; \overline{\delta u} \, \delta j
\end{eqnarray*}
where the second equality holds true by the definition of the
perturbed normal component of boundary current density $j + \delta
j = ( \gamma + \delta \gamma )\nabla ( u + \delta u ) \cdot
\mathbf{n}$. Substituting back to (\ref{ppcl}) yields
\begin{eqnarray}\label{appc} \nonumber
& & \int_B \mathrm{d}x \; \delta \gamma | \nabla u|^2 + \int_B \mathrm{d}x \; (\gamma + \delta \gamma) \nabla \delta u \cdot \nabla \overline{u} + \sum_{\ell=1}^L z_\ell \int_{\Gamma_{e_\ell}} \mathrm{d} s |\delta j|^2\\
& & +  \sum_{\ell=1}^L z_\ell \int_{\Gamma_{e_\ell}} \mathrm{d} s \; \delta j \, \overline{j} + \int_{\Gamma_e} \mathrm{d} s \;  \overline{\delta u} \, \delta j = 0,
\end{eqnarray}
while applying the perturbations to the electrode potential
boundary condition (\ref{boundv}) gives $\overline{\delta u} =
\overline{\delta U_\ell} - z_\ell \overline{\delta j}$, and
therefore the last integral term becomes
\begin{eqnarray*}
\int_{\Gamma_e}\mathrm{d} s\, \overline{\delta u}\, \delta j & = & \sum_{\ell=1}^L \int_{\Gamma_{e_\ell}}\mathrm{d} s\, (\overline{\delta U_\ell} - z_\ell \overline{\delta j}) \delta j\\
                                        & = & \sum_{\ell=1}^L \overline{\delta U_\ell}  \int_{\Gamma_{e_\ell}}\mathrm{d} s\, \delta j - \sum_{\ell=1}^L z_\ell \int_{\Gamma_{e_\ell}}\mathrm{d} s\, |\delta j|^2\\
                                        & = &  - \sum_{\ell=1}^L z_\ell \int_{\Gamma_{e_\ell}}\mathrm{d} s\, |\delta j|^2,
\end{eqnarray*}
where the last equation is due to the following lemma.
\begin{lemma}\label{dj}
The perturbations in electrical admittivity $\gamma \rightarrow
\gamma + \delta \gamma$, and induced electric potential in the
interior of the domain $u \rightarrow u + \delta u$ give rise to a
perturbation in the boundary current density with vanishing
integral
$$
\int_{\partial B}\mathrm{d} s\, \delta j(\mathbf{x}) = 0, \qquad \mathbf{x} \in \partial B.
$$
\end{lemma}
\begin{proof}
From the Neumann boundary condition (\ref{boundc}) the current
applied at the $\ell$'th electrode satisfies
$$
I_\ell = \int_{\Gamma_{e_\ell}} \mathrm{d} s\, \gamma \nabla u \cdot \mathbf{n} = \int_{\Gamma_{e_\ell}} \mathrm{d} s\, j.
$$
Keeping $I_\ell$ fixed before and after effecting the perturbations gives
$$
I_\ell = \int_{\Gamma_{e_\ell}} \mathrm{d} s\, (\gamma + \delta \gamma) \nabla (u + \delta u) \cdot \mathbf{n} = \int_{\Gamma_{e_\ell}} \mathrm{d} s\, (j + \delta j).
$$
Splitting the last integral, equating the right hand sides of the
two equations above, and recalling from (\ref{boundc}), that
$(j(\mathbf{x})+\delta j(\mathbf{x})) = 0$ for $\mathbf{x} \in
\partial  B \setminus \Gamma_e$ yields the result.
\end{proof}

Effectively equation (\ref{appc}) reduces further to
\begin{equation}
 \int_B \mathrm{d}x \; \delta \gamma | \nabla u|^2 + \int_B \mathrm{d}x \; (\gamma + \delta \gamma) \nabla \delta u \cdot \nabla \overline{u}  +  \sum_{\ell=1}^L z_\ell \int_{\Gamma_{e_\ell}} \mathrm{d} s \; \delta j \, \overline{j} = 0,
\end{equation}
and using once again the perturbed Robin condition the last integral simplifies further to
\begin{eqnarray*}
\sum_{\ell=1}^L z_\ell \int_{\Gamma_{e_\ell}} \mathrm{d} s \; \delta j \, \overline{j} & = & \sum_{\ell=1}^L \delta U_\ell \int_{\Gamma_{e_\ell}} \mathrm{d} s \; \overline{j} - \int_{\Gamma_{e_\ell}} \mathrm{d} s\; \delta u \cdot \overline{j}\\
                                                          & = & \sum_{\ell=1}^L \overline{I_\ell} \delta U_\ell - \int_{\Gamma_{e_\ell}}\mathrm{d} s\; \delta u \cdot \overline{j}\\
                                                          & = &     \sum_{\ell=1}^L \overline{I_\ell} \delta U_\ell - \int_{\Gamma_{e_\ell}}\mathrm{d} s\; \overline{\gamma}\, \delta u \cdot \nabla \overline{u} \cdot \mathbf{n}
\end{eqnarray*}
Now, consider the adjoint field problem (\ref{adj1})-(\ref{adj3})
subject to a current $I^m = \overline{I^d}$. Then by the
properties of the complete electrode admittance operator
$\mathcal{A}_{\gamma,z}$ it is easy to show that the adjoint
solution $v(\overline{\gamma},I^m)$ coincides with
$\overline{u}(\gamma,I^d)$. Applying the divergence theorem to the
adjoint field equation (\ref{adj})  gives
$$
\int_B \mathrm{d}x\; \overline{\gamma} \nabla \overline{u} \cdot \nabla \delta u = \int_{\Gamma_e}\mathrm{d} s\; \delta u \overline{\gamma} \nabla \overline{u} \cdot \mathbf{n} = \int_{\Gamma_e}\mathrm{d} s\; \delta u \, \overline{j}.
$$
From the above the perturbed power conservation law finalizes to
\begin{equation}\label{fppcl}
\sum_{\ell=1}^L \overline{I_\ell} \delta U_\ell = - \int_B \mathrm{d}x \, \delta \gamma |\nabla u|^2 - \int_B \mathrm{d}x \, \delta \gamma \nabla \delta u \cdot \nabla \overline{u} - \int_B \mathrm{d}x\, (\gamma - \overline{\gamma}) \nabla \delta u \cdot \nabla \overline{u}.
\end{equation}
Notice that for the purely real conductivity case, i.e. the cases
of electrical resistance tomography where $\omega=0$, the third
term on the right hand side vanishes and the above collapses to
the formula provided in \cite{Polydorides_linear}.

\begin{lemma}
If the applied currents are purely real,  the perturbed power conservation law (\ref{fppcl}) simplifies to
\begin{equation}\label{sppcl}
\sum_{\ell=1}^L I_\ell \delta U_\ell = - \int_\Omega \mathrm{d}x \, \delta \gamma \nabla u \cdot \nabla u - \int_\Omega \mathrm{d}x \, \delta \gamma \nabla \delta u \cdot \nabla u.
\end{equation}
\end{lemma}
\begin{proof}
Consider applying the diverge theorem to (\ref{iden}) for a test
function $\psi=\bar{u}$ and to the adjoint pde (\ref{adj}) for
$\psi=\delta u$. Then upon subtracting the later from the former
yields,
\begin{eqnarray*}
\sum_{\ell=1}^L \bar{I_\ell} \delta U_\ell & = & - \int_B \mathrm{d}x \, \delta \gamma |\nabla u|^2 - \int_B \mathrm{d}x \, \delta \gamma \nabla \delta u \cdot \nabla \overline{u} - \int_B \mathrm{d}x\, (\gamma - \overline{\gamma}) \nabla \delta u \cdot \nabla \overline{u}\\
& = & \int_{\Gamma_e}\mathrm{d} s\; \overline{\gamma} \delta u \nabla \overline{u} \cdot \mathbf{n} - \int_{\Gamma_e}\mathrm{d} s\; \overline{u} \bigl (\gamma \nabla \delta u + \delta \gamma \nabla (u+\delta u) \bigr  ) \cdot \mathbf{n}\\
& = & \int_{\Gamma_e}\mathrm{d} s\;  \delta u \, \overline{j} - \int_{\Gamma_e}\mathrm{d} s\; \overline{u} \,  \delta j\\
& = & \int_{\Gamma_e}\mathrm{d} s\; (\delta U_\ell - z_\ell \delta j) \overline{j} - \int_{\Gamma_e}\mathrm{d} s\; (\overline{U_\ell} - z_\ell \overline{j}) \delta j ,
\end{eqnarray*}
where the last equation is due to lemma (\ref{dj}). Similarly,
from the diverge theorem to (\ref{iden}) with $f=u$ and to
(\ref{pdei}) with $\psi=\delta u$ one obtains
\begin{eqnarray*}
& & - \int_B \mathrm{d}x \, \delta \gamma \nabla u \cdot \nabla u - \int_B \mathrm{d}x \, \delta \gamma \nabla \delta u \cdot \nabla u \\
& & = \int_{\Gamma_e}\mathrm{d} s\; \gamma \delta u \nabla u \cdot \mathbf{n} - \int_{\Gamma_e}\mathrm{d} s\; u \bigl (\gamma \nabla \delta u + \delta \gamma \nabla (u+\delta u) \bigr  ) \cdot \mathbf{n}\\
&  & = \int_{\Gamma_e}\mathrm{d} s\;  \delta u \, \overline{j} - \int_{\Gamma_e}\mathrm{d} s\; \overline{u} \,  \delta j\\
&  & = \int_{\Gamma_e}\mathrm{d} s\; (\delta U_\ell - z_\ell \delta j) j - \int_{\Gamma_e}\mathrm{d} s\; (U_\ell - z_\ell j) \delta j = \sum_{\ell=1}^L I_\ell \delta U_\ell.
\end{eqnarray*}
From the above the result follows in the case where
$\overline{I_\ell} = I_\ell$, i.e. the imaginary component of the
currents is zero.
\end{proof}

For simplicity we assume the case of real excitation currents. For
a current pattern $I$,  let $\gamma_p, \; \gamma \in
L^\infty(\overline B)$, the states of the model before and after
the admittivity perturbation so that the change on the potential
of the $\ell$'th electrode is
$$
\delta U_\ell (I) = U_\ell(\gamma,I) - U_\ell(\gamma_p,I),
$$
and evaluate equation (\ref{fppcl}) for some pair drive current
patterns that satisfy the constraint (\ref{icond}). Let $\mu_d,
\mu_m \in \mathbb{R}^L$ as in (\ref{mus}) some discrete patterns
of zero sum, and define the currents
$$
I^d = a \mu_d, \qquad I^m = \mu_m, \qquad I^c = I^d + I^m.
$$
Suppose the currents are applied to the model with known
admittivity $\gamma_p$, and then to that of the unknown $\gamma$,
giving rise to $U(\gamma_p,I^t) = \mathcal{A}^{-1}_{\gamma_p,z}
I^t$, and $U(\gamma,I^t) = \mathcal{A}^{-1}_{\gamma,z} I^t$, from
which we compute the difference as
$$
\delta U(I^t) = U(\gamma,I^t) - U(\gamma_p,I^t),
$$  for $t=\{d,m,c\}$. Based on the linearity of the admittance operator we deduce that
$$
\delta U(I^c) = \mathcal{A}^{-1}_{\gamma,z} (I^d + I^m) - \mathcal{A}^{-1}_{\gamma_p,z} (I^d + I^m),
$$
$\delta U(I^d) = \mathcal{A}^{-1}_{\gamma,z} I^d - \mathcal{A}^{-1}_{\gamma_p,z} I^d$, and $\delta U(I^m) = \mathcal{A}^{-1}_{\gamma,z} I^m - \mathcal{A}^{-1}_{\gamma_p,z} I^m$. Evaluating the left hand side of (\ref{fppcl}) for the three current patterns yields
$$
\sum_{\ell=1}^L I^c_\ell \delta U^c_\ell - \sum_{\ell=1}^L I^d_\ell \delta U^d_\ell - \sum_{\ell=1}^L I^m_\ell \delta U^m_\ell = I_o \bigl ( \delta U^m_{e_p} - \delta U^m_{e_n} \bigr ) + \bigl ( \delta U^d_{e_{p'}} - \delta U^d_{e_{n'}} \bigr ).
$$
It is worth noticing that only $\delta U^d$ are realistically
measurable, since data acquisition occurs only under the direct
patterns and borrowing the reciprocity result (\ref{grt}) for
$I_o=1$ gives
\begin{equation}
\sum_{\ell=1}^L I^c_\ell \delta U^c_\ell - \sum_{\ell=1}^L I^d_\ell \delta U^d_\ell - \sum_{\ell=1}^L I^m_\ell \delta U^m_\ell = 2 \bigl ( \delta U^d_{e_{p'}} - \delta U^d_{e_{n'}} \bigr ).
\end{equation}
Expanding the corresponding right hand sides from (\ref{sppcl}) yields
$$
\sum_{\ell=1}^L I^c_\ell \delta U^c_\ell - \sum_{\ell=1}^L I^d_\ell \delta U^d_\ell - \sum_{\ell=1}^L I^m_\ell \delta U^m_\ell =  -2 \int_B \mathrm{d}x\; \delta \gamma \; \nabla u(I^d) \cdot \nabla u(I^m) - 2 \int_B \mathrm{d}x\; \delta \gamma \; \nabla \delta u(I^d) \cdot \nabla u(I^m),
$$
where we have used $u(\gamma, I^c) = u(\gamma, I^d) +
u(\gamma,I^m)$ for the interior fields. Let the $k$'th measurement
be $\zeta_k = \mu_m U$ and note that $u(\gamma_p,I^m) =
\overline{v}$, for $v$ the adjoint fields solution of (\ref{adj}).
In effect, substituting and simplifying yields
\begin{equation}\label{jacobian}
\delta \zeta_k = -\int_B \mathrm{d}x \; \delta \gamma \;\nabla u(\gamma_p,I^d) \cdot \nabla \overline{v}(\gamma_p,I^m) - \int_B \mathrm{d}x\; \delta \gamma \; \nabla \delta u(I^d) \cdot \nabla \overline{v}(\gamma_p,I^m).
\end{equation}
We are now ready to tabulate our main result in the form of the
following theorem.

\begin{theorem}\label{thm} (The forward EIT transform)
Consider the complete electrode model of (\ref{pdei}) -
(\ref{icond}) on a simply connected domain $B$, and suppose assumptions \ref{ass} hold. Suppose further that the applied currents
are purely real and that boundary measurements $\zeta \in
\mathbb{C}^m$ are observed. If $u$ is the direct solution of this
problem and $v$ the pertinent adjoint vector satisfying
(\ref{adj}), then for any prior admittivity guess $\gamma_p \in
L^\infty(\overline B)$ with direct solution
$\mathcal{E}(\gamma_p)$, the data change $\delta \zeta_k$ the
$k$th element of the residual $\delta \zeta = \zeta -
\mathcal{E}(\gamma_p)$ satisfies
\begin{equation}\label{nlt}
\delta \zeta_k = - \int_B \mathrm{d}x \; \delta \gamma \nabla u(\gamma) \cdot \nabla \overline{v}(\gamma_p),
\end{equation}
where  $\delta \gamma = \gamma - \gamma_p$ is the residual vector between the target solution and the initial-prior guess.
\end{theorem}
\begin{proof}
The result follows immediately by substituting $\delta u =
u(\gamma) - u(\gamma_p)$ for all direct currents $I^d$ to the
integral equation (\ref{jacobian}), and holds true for all
admissible bounded perturbations $\delta \gamma$. This completes
the proof.
\end{proof}

We would like to note that, in the Appendix we provide an alternative derivation of (\ref{nlt}) suggested to us by an anonymous reviewer based on a weak formulation of the problem.

\subsection{Generalization to Poisson's equation with mixed boundary conditions}

Although the complete electrode model is now widely used for EIT, our new
model formulation in (\ref{nlt}) as well as the image
reconstruction method to be described next are easily amenable to
treat more simplistic electrode models. In particular, we now show
that the above result holds true for a more general setting of
impedance imaging involving the Poisson equation with Dirichlet
and Neumann boundary conditions and point electrodes
\cite{Miller}, \cite{Kemna}. In geo-electrical application one
usually encounters the model
\begin{equation}\label{pdeg}
\nabla \cdot [ \gamma(\mathbf{x},\omega) \nabla u (\mathbf{x},\omega)]=f(\mathbf{x}), \quad \mathbf{x} \in B,
\end{equation}
with boundary conditions of the form
\begin{equation}\label{pdeg_bc}
\alpha(\mathbf{x})\gamma(\mathbf{x},\omega)\nabla u(\mathbf{x},\omega) \cdot \mathbf{n} + \beta(\mathbf{x}) u(\mathbf{x},\omega) = 0, \quad \mathbf{x} \in \partial B.
\end{equation}
where $\alpha$ and $\beta$ are functions defined on $\partial B$
and are not simultaneously zero to thoroughly impose the boundary
conditions. To consider problems with different types of boundary
conditions on different regions of $\partial B$, the functions
$\alpha$ and $\beta$ are allowed to be discontinuous. Figure
\ref{figgeo} shows a common geophysical problem associated with
the model in (\ref{pdeg})--(\ref{pdeg_bc}). In this problem
$\Gamma_n$ is the interface between the earth and air where a zero
current condition ($\beta=0$) holds. In the remaining boundary
$\Gamma_m=\partial B \setminus \Gamma_n$, the values $\alpha$ and
$\beta$ are appropriately chosen to model an infinite half-space
\cite{pollock2008temporal}. When the sources of current are far
from $\Gamma_m$, a zero potential condition ($\alpha=0$) may be
used as an approximation to the infinite half-space
\cite{tripp1984two}.

\begin{figure}
\begin{center}
\epsfig{file=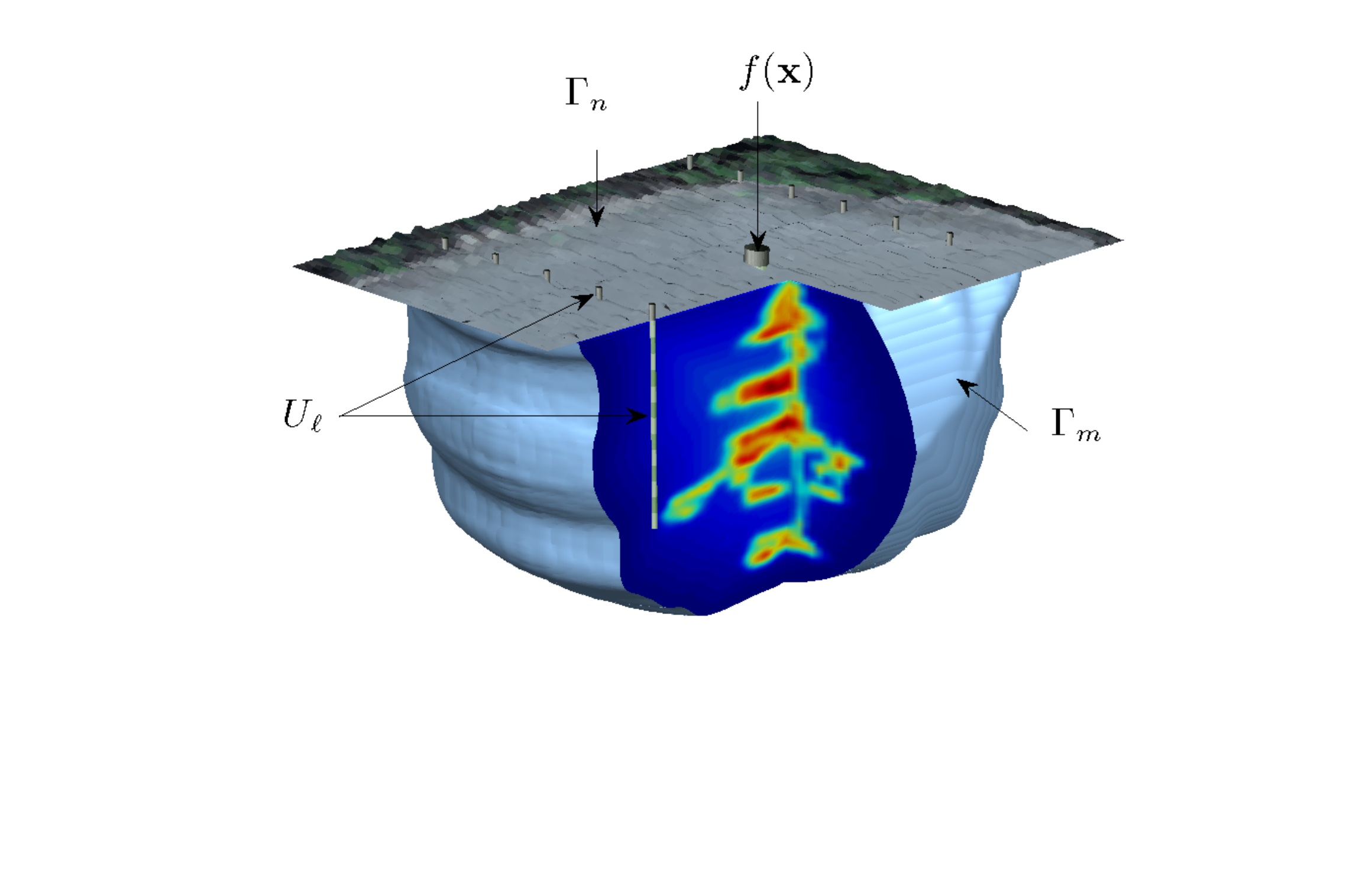,height=11cm,width=16.5cm}\\\vspace{-2cm} %8,12
\end{center}
\caption{Geophysical application problem setting. The current sources $f(x)$  are applied though the borehole electrodes yielding electrode potentials $U_\ell$. $\Gamma_m$ is the model termination boundary and $\Gamma_n$ is the upper surface of the model domain $B$.}\label{figgeo}
\end{figure}

The electric potential measurements are collected through
point-wise electrodes, contact impedances of which are effectively
zero. The measurement points are $\mathbf{x}_\ell$ for $\ell =
1,2,\ldots,L$ and the measured potential at every point is
\begin{equation}\label{eq:3}
U_\ell=\int_B  \mathrm{d}x \; u(\mathbf{x})\delta(\mathbf{x}-\mathbf{x}_\ell),
\end{equation}
where $\delta(.)$ denotes the Dirac delta function. Consider a
perturbation $\gamma\rightarrow \gamma + \delta \gamma$ in the
additivity causing the potential perturbation $ u\rightarrow
u+\delta u$. Introducing these into (\ref{pdeg})--(\ref{pdeg_bc})
gives
\begin{eqnarray}
\nabla \cdot \big((\gamma+\delta \gamma) \nabla ( u+\delta u)\big) &=& f, \hspace{1.45 cm} \mbox{on}\;B,\label{eq:4}\\
\alpha(\gamma+\delta \gamma)\nabla( u+\delta u)\cdot \mathbf{n}+\beta( u+\delta u)
&=& 0, \hspace{1.45cm} \mbox{on}\;\partial B.\label{eq:5}
\end{eqnarray}
Expanding (\ref{eq:3}) and (\ref{eq:4}) and using (\ref{pdeg})--(\ref{pdeg_bc}) to simplify the resulting terms yields
\begin{eqnarray}
\nabla \cdot (\delta \gamma \nabla  u)+\nabla \cdot (\gamma \nabla \delta u)+\nabla \cdot (\delta \gamma \nabla \delta u)&=&0, \hspace{1.45 cm} \mbox{on}\;B,\label{eq:6}\\
\alpha(\delta \gamma \nabla  u \cdot \mathbf{n} + \gamma \nabla \delta u\cdot \mathbf{n}+ \delta \gamma \nabla \delta u \cdot
\mathbf{n})+\beta\delta u&=&0, \hspace{1.45cm}
\mbox{on}\;\partial B.\label{eq:7}
\end{eqnarray}
Based on (\ref{eq:3}) a perturbation in the measurement at $\mathbf{x}_\ell$ can be written as a volume integral
\begin{equation}
  \label{eq:8}
  \delta U_\ell =  \int_B \mathrm{d}x \; \delta u(\mathbf{x})\delta(\mathbf{x}-\mathbf{x}_\ell)
\end{equation}
To proceed with finding a closed form for the measurement perturbation $\delta U_\ell$, it is useful to define $v$, as the
solution to the adjoint system
\begin{eqnarray}
\nabla \cdot (\overline \gamma \nabla  v_\ell)&=&\delta(\mathbf{x}-\mathbf{x}_\ell),
\hspace{1.45cm}\mathbf{x}\in B,\label{eq:9}
\\\alpha\overline \gamma\nabla v_\ell\cdot \mathbf{n}+\beta v_\ell &=& 0,
\hspace{2.65cm} \mathbf{x}\in\; \partial B,\label{eq:10}
\end{eqnarray}
from which it is easily inferred that $\overline{v_\ell}$ satisfies
\begin{eqnarray}
\nabla \cdot (\gamma \nabla \overline{v_\ell})&=&\delta(\mathbf{x}-\mathbf{x}_\ell),
\hspace{1.45cm}\mathbf{x}\in B,\label{eq:11}
\\\alpha\gamma\nabla \overline{v_\ell}\cdot \mathbf{n}+\beta \overline{v_\ell} &=&
0, \hspace{2.65cm} \mathbf{x}\in\; \partial B.\label{eq:12}
\end{eqnarray}
Using (\ref{eq:8}) and (\ref{eq:11}) we conclude that the
perturbation to the residuals can be written in terms of the
adjoint field as
\begin{equation}
  \label{eq:13}
  \delta U_\ell =  \int_B \mathrm{d}x \; \delta u(\mathbf{x}) \;\nabla\cdot(\gamma \nabla \overline{v_\ell}).
\end{equation}
The remaining derivation requires extensive use of the following
identity derived from Green's theorem \cite{marsden2003vector} for
vector function $\mathbf{\Psi}$ and scalar function $\psi$
\begin{equation}
  \label{eq:14}
  \int_B \mathrm{d}x\, \mathbf{\Psi}\cdot\nabla \psi  + \int_B \mathrm{d}x\, \psi \nabla \cdot \mathbf{\Psi}  =
  \int_{\partial B} \mathrm{d}s\, \psi \mathbf{\Psi} \cdot\mathbf{n}.
\end{equation}
We begin by taking $\psi = \delta u$ and $\mathbf{\Psi} = \gamma \nabla \overline{v_\ell}$ in (\ref{eq:13}) to obtain
\begin{equation}
  \label{eq:15}
  \delta U_\ell =
    -\int_B \mathrm{d}x\, \gamma \nabla \overline{v_\ell} \cdot \nabla\delta u
  + \int_{\partial B}  \mathrm{d}s\, \gamma \delta u\nabla\overline{v_\ell}\cdot \mathbf{n}.
\end{equation}
Next using $\psi =  \overline{v_\ell}$ and $\mathbf{\Psi} =
\gamma\nabla\delta u$ in the first term on the right hand side of
(\ref{eq:15}), we have
\begin{align}
  \label{eq:16}
    \delta U_\ell =
     \int_B \mathrm{d}x\, \overline{v_\ell} \nabla\cdot(\gamma\nabla\delta u)
    - \int_{\partial B}  \mathrm{d}s\, \gamma\;\overline{v_\ell}\nabla\delta u \cdot \mathbf{n}
    + \int_{\partial B}  \mathrm{d}s\, \gamma \delta u\nabla\overline{v_\ell} \cdot \mathbf{n}.
\end{align}
From (\ref{eq:6}), $\nabla \cdot (\gamma \nabla \delta u)=-\nabla \cdot (\delta \gamma \nabla  u)-\nabla \cdot (\delta \gamma \nabla \delta u)$ which we use in the first term on the right hand side of (\ref{eq:16}) to arrive at
\begin{align}
  \label{eq:17}
  \delta U_\ell =
     -\int_B \mathrm{d}x\, \overline{v_\ell} \nabla\cdot\big(\delta \gamma\nabla (u+ \delta u)\big)
    - \int_{\partial B} \mathrm{d}s\, \gamma\;\overline{v_\ell}\nabla\delta u \cdot \mathbf{n}
    + \int_{\partial B} \mathrm{d}s\, \gamma \delta u\nabla\overline{v_\ell}\cdot \mathbf{n}.
\end{align}
Appealing once more to (\ref{eq:14}) with $\psi =
\overline{v_\ell}$ and $\mathbf{\Psi} = \delta \gamma\nabla
(u+\delta u)$ in the first term of (\ref{eq:17}) gives
\begin{align}
\nonumber
  \delta U_\ell &=
     \int_B \mathrm{d}x\, \delta \gamma \nabla \overline{v_\ell}\cdot\nabla u+\int_B \mathrm{d}x\, \delta \gamma \nabla \overline{v_\ell}\cdot\nabla
     \delta u\\
    &- \int_{\partial B} \mathrm{d}s\, \big(\gamma\overline{v_\ell}\nabla\delta u\cdot\mathbf{n}+\delta \gamma\overline{v_\ell}\nabla  u \cdot\mathbf{n}+\delta \gamma\overline{v_\ell}\nabla  \delta u \cdot\mathbf{n}-\gamma \delta u \nabla \overline{v_\ell} \cdot\mathbf{n} \big).\label{eq:18}
\end{align}
We now show that the surface integral term in (\ref{eq:18}) is zero. For this purpose we multiply both sides of (\ref{eq:12}) by
$\delta u$ to arrive at
\begin{equation}
  \label{eq:19}
  \alpha\gamma\delta u\nabla \overline{v_\ell} \cdot \mathbf{n}+\beta\delta u \overline{v_\ell} = 0
\end{equation}
Using (\ref{eq:7}) to replace the term $\beta\delta u$ in (\ref{eq:19})
results in
\begin{equation}
  \label{eq:20}
  -\alpha \big( \gamma\overline{v_\ell}\nabla\delta u\cdot\mathbf{n}+ \delta \gamma\overline{v_\ell}\nabla  u \cdot\mathbf{n}+\delta \gamma\overline{v_\ell}\nabla  \delta u \cdot\mathbf{n}-\gamma\delta u\nabla \overline{v_\ell} \cdot\mathbf{n} \big)=0, \qquad
  \text{on}\;\partial B.
\end{equation}
The parenthesized expression in (\ref{eq:20}) is the same as the
surface integrand in (\ref{eq:18}). We partition the boundary
$\partial B$ into $\Gamma_\alpha$ where $\alpha\neq 0$ and
$\partial B \setminus\Gamma_\alpha$ where $\alpha=0$. Clearly
(\ref{eq:20}) results the inside bracket expression to vanish on
$\Gamma_\alpha$. On the remaining surface $\partial B
\setminus\Gamma_\alpha$ that $\alpha=0$, we certainly have
$\beta\neq 0$ since $\alpha$ and $\beta$ may not be simultaneously
zero and using this fact in (\ref{eq:7}) and (\ref{eq:12}) would
result in $\delta u=0$ and $\overline{v_\ell}=0$ which again make
the inside bracket term zero. Therefore the surface integral in
(\ref{eq:18}) vanishes both on $\Gamma_\alpha$ and $\partial B
\setminus\Gamma_\alpha$ and therefore
\begin{equation}\label{eq:21}
  \delta U_\ell =
     \int_B \mathrm{d}x\, \delta \gamma \nabla \overline{v_\ell}\cdot\nabla u+\int_B \mathrm{d}x\, \delta \gamma \nabla \overline{v_\ell}\cdot\nabla
     \delta u,
\end{equation}
and thus by substituting for $\delta u$ in the second term we arrive at the result of the theorem \ref{thm}.

\section{High-order regularized regression}\label{results}

Within the $d-$ dimensional sphere $S_{\gamma_p,\kappa}$, the
electric potential field in the interior of the domain admits a
Taylor expansion
$$
u(\gamma) = u(\gamma_p) + \partial_\gamma u(\gamma_p) \delta \gamma + \mathcal{O}(\|\delta \gamma^2\|)
$$
hence to first-order accuracy this can be approximated by
\begin{equation}\label{linap}
u(\gamma) \simeq \hat{u}(\gamma) = u(\gamma_p) + \partial_\gamma u(\gamma_p) \delta \gamma.
\end{equation}
Introducing the right hand side of (\ref{linap}) in the integral equation (\ref{nlt}) gives
\begin{eqnarray}\label{exnlt}
\delta \zeta_k & \approx & - \int_B \mathrm{d}x \; \delta \gamma \nabla  \bigl (  u(\gamma_p) + \partial_\gamma u(\gamma_p) \delta \gamma \bigr )\cdot \nabla \overline{v}(\gamma_p),\\\nonumber
                       & = & - \int_B \mathrm{d}x \; \delta \gamma \nabla  u(\gamma_p) \cdot \nabla \overline{v}(\gamma_p) - \int_B \mathrm{d}x \; \delta \gamma \; \nabla (\partial_\gamma u(\gamma_p) \delta \gamma) \cdot \nabla \overline{v}(\gamma_p),
\end{eqnarray}
where the first, linear term, involves the definition of the
Fr\'{e}chet derivative of the forward mapping as in (\ref{lin})
\cite{Miller}, \cite{Oldenburg}, and the second nonlinear term the
differential operator $\partial_\gamma u(\gamma_p)$ that provides
a measure on local sensitivity of the potential in the interior of
the domain to perturbations in electrical properties. From
(\ref{linap}), (\ref{exnlt}), it is trivial to deduce that the
linear approximation of the forward operator $\mathcal{E}$ as in
(\ref{tay1}), as proposed by Calder\'{o}n in \cite{Calderon},
effectively imposes a zeroth-order Taylor approximation on the
electric potential $\hat u(\gamma) \simeq u(\gamma_p)$. In turn
this enforces $\partial_\gamma u$ and higher-order derivatives to
vanish everywhere in $\bar B$, thus elliminating the nonlinear
terms in (\ref{jacobian}) and (\ref{exnlt}). Let the linear
operator $\partial_\gamma \mathcal{E} = \mathcal{J}:
L^\infty(\overline B) \rightarrow \mathbb{C}^m$, and nonlinear,
quadratic in $\delta \gamma$, $\mathcal{K}: L^\infty(\overline B)
\rightarrow \mathbb{C}^m$ defined by
\begin{eqnarray}\label{linop}
\mathcal{J}\, \delta \gamma & \doteq & - \int_B \mathrm{d}x \; \delta \gamma \; \nabla  u(\gamma_p) \cdot \nabla \overline{v}(\gamma_p),\\\label{nonlinop}
\mathcal{K} \,\delta \gamma & \doteq & - \int_B \mathrm{d}x \; \delta \gamma \; \nabla \partial_\gamma u(\gamma_p) \delta \gamma \cdot \nabla \overline{v}(\gamma_p)
\end{eqnarray}
then the inverse problem can be formulated in the context of
regularized regression based on the nonlinear operator equation
\begin{equation}\label{contip}
\delta \zeta = \mathcal{J} \delta \gamma + \mathcal{K} \delta \gamma + \eta.
\end{equation}

%%%change k to N in number of elements in the model
\subsection{Numerical approximation}

Usually the EIT problem is approached with a numerical
approximation method like finite elements, where the governing
equations are discretized on a finite dimensional model of the
domain, say $B_h(n,N)$ comprising $n$ nodes connected in $N$
elements \cite{eidors}. For simplicity in the notation we assume
linear Lagrangian finite elements and consider element-wise linear
and constant basis functions for the support of the electric
potential $u$ and conductivity $\gamma$ respectively,
\begin{equation}\label{approx}
u(\mathbf{x},\omega) = \sum_{i=1}^n u_i \phi_i, \quad \phi_i:B_h \rightarrow \mathbb{R}, \qquad \gamma(\mathbf{x},\omega) = \sum_{i=1}^N  \gamma_i \chi_i, \quad \chi_i:B_h \rightarrow \mathbb{R}
\end{equation}
where $\{\phi_i\}_{i=1}^n$ and  $\{\chi_i\}_{i=1}^N$ the
respective bases in $B_h$. Following the discretization of the domain into a finite number of elements, the basis functions $\{\phi_1,\ldots,\phi_n\}$ in the expansion of the potential are assumed to belong in a finite, $n$-dimensional subspace of $H_o^1(B)$. For clarity in the notation, we keep $u$ and $\gamma$ as the vectors of coefficients relevant to the respective functions as from now on we deal exclusively the
numerical approximation of the problem. On the discrete domain the
weak form of the operator equation (\ref{contip}) is approximated
by
\begin{equation}\label{quadmodel}
{\delta \zeta}_k = j_k' \delta \gamma + {\delta \gamma}' \mathbf{K}^{k} \delta \gamma + \eta_k, \qquad k=1,\ldots,m
\end{equation}
where $\zeta_k \in \mathbb{C}$ is the $k$th measurement, $j_k$ the
$k$th row of the Jacobian matrix $\mathbf{J}$ that is the discrete
form of $\partial_\gamma \mathcal{E}(\gamma_p)$, $\mathbf{K}^{k}
\in \mathbb{C}^{N \times N}$ is the $k$th coefficients (Hessian)
matrix derived from $\mathcal{K}$ in (\ref{nonlinop}), $\eta_k$
the noise in the $k$th measurement and $\delta \gamma \in
\mathbb{C}^N$ the required perturbation in the admittivity
coefficients. Let the additive noise be uncorrelated zero-mean
Gaussian with diagonal covariance matrix $\mathbf{C}_\eta$, with positive diagonal element $c_k$ then the data
misfit function
\begin{equation}\label{quadmis}
Q(\delta \gamma) = \sum_{k=1}^m c_k^{-1} \bigl ( {\delta \zeta}_k - j_k'\delta \gamma - \delta \gamma' \mathbf{K}^k \delta \gamma \bigr )^2
\end{equation}
can be used to define the regularized \emph{quadratic regression} problem
\begin{equation}\label{qrp}
{\delta \gamma}^* = \arg\min_{\delta \gamma \in \mathbb{C}^N} \xi(\delta \gamma), \qquad \xi(\delta \gamma) \doteq \frac{1}{2} \bigl \{ Q(\delta \gamma) + \alpha \mathcal{G}(\delta \gamma) \bigr \}
\end{equation}
with $\mathcal{G}: \mathbb{C}^N \rightarrow \mathbb{R}$ a convex
differentiable regularization term. On the other hand, choosing to
neglect the matrices $\mathbf{K}^k$ yields the conventional misfit
function
\begin{equation}\label{linmis}
\Lambda(\delta \gamma) = \sum_{k=1}^m c_k^{-1} \bigl ( {\delta \zeta}_k - j_k'\delta \gamma \bigr )^2,
\end{equation}
often used in the context of regularized linear regression
formulations. As shown in \cite{adgablio} the Jacobian matrix can
be computed directly from (\ref{linop}) and (\ref{approx}) using
numerical integration as
\begin{equation}\label{jac}
\mathbf{J}_{k,j} =  - \int_{B_j} \mathrm{d}x \; \chi_j \; \sum_{l \,\in \,\mathrm{supp}(B_j)} u_l \nabla  \phi_l  \sum_{l \,\in \, \mathrm{supp}(B_j)} \overline{v}_l \nabla  \phi_l , \quad k=1,\ldots,m, \; j=1,\ldots,N
\end{equation}
with $v$ the coefficients of the adjoint field solution
corresponding to the $k$th measurement, and $\mathrm{supp}(B_j)$
the support of the $j$th element. To derive the respective element
of $\mathbf{K}^k$ we follow an approach similar to that of Kaipio
et al. in \cite{Kaipio_mc} that is based on the Galerkin
formulation of the problem. For this we choose
$\{\phi_1,\ldots,\phi_{n}\}$ as a test basis for the potentials
and by substituting into the variational form of the model we
arrive at
$$
\sum_{i=1}^{n} \sum_{j=1}^{n} \Bigl (\int_B\mathrm{d}x \, \gamma \nabla \phi_i \cdot \nabla \phi_j +  \sum_{\ell=1}^L z_\ell \int_{\Gamma_{e_\ell}} \mathrm{d}s \, \phi_i \, \phi_j \Bigr ) u_i  - \sum_{\ell=1}^L  z_\ell \int_{\Gamma_{e_\ell}} \mathrm{d} s\,  \phi_i \;  U_\ell = 0.
$$
Imposing the Neumann conditions for the applied boundary currents yields the additional equations
$$
I_\ell = - z_\ell \sum_{i=1}^{n}  \Bigl ( \int_{\Gamma_{e_\ell}} \mathrm{d}s\, \phi_i \Bigr ) u_i + z_\ell \, |\Gamma_{e_\ell}| \, U_\ell , \quad \ell=1,\ldots,L,
$$
with $|\Gamma_{e_\ell} |$ the area of the $\ell$th electrode. In matrix
form the electric potential expansion coefficients $u \in
\mathbb{C}^{n}$ and the electrode potentials $U \in \mathbb{C}^L$
can be computed by solving the $(n + L) \times (n+ L)$ matrix
equation
\begin{equation}\label{matrix}
\left[\begin{array}{c c}\mathbf{A}_{11} & \mathbf{A}_{12} \\  \mathbf{A}_{12}' & \mathbf{A}_{22}\end{array}\right]\left[\begin{array}{c} u \\U\end{array}\right] = \left[\begin{array}{c}0 \\I\end{array}\right],
\end{equation}
where
\begin{eqnarray*}
{\mathbf{A}_{11}}_{\,i,j} & = & \int_B\mathrm{d}x\, \gamma \nabla \phi_i \cdot \nabla \phi_j +  \sum_{\ell=1}^L z_\ell  \int_{\Gamma_{e_\ell}} \mathrm{d}s \, \phi_i \, \phi_j, \quad i,j=1,\ldots,n\\
{\mathbf{A}_{12}}_{\,i,\ell} & = &  - z_\ell \int_{\Gamma_{e_\ell}} \mathrm{d}s\,  \phi_i ,  \quad i=1,\ldots,n, \quad \ell=1,\ldots,L,\\
{\mathbf{A}_{22}}_{\,\ell,\ell} & = &  z_\ell |\Gamma_{e_\ell}|, \quad \ell=1,\ldots,L.
\end{eqnarray*}
For a conductivity $\gamma$ and applied current $I$, let
$\begin{bmatrix} u\\ U  \end{bmatrix} = \mathbf{A}^{-1}(\gamma)
\begin{bmatrix}0 \\I\end{bmatrix}$ the solution of (\ref{matrix}).
Using the matrix differentiation formula, the partial derivatives
with respect to the $q$th admittivity element are
\begin{eqnarray*}
\partial_{\gamma_q} \begin{Bmatrix} u\\ U  \end{Bmatrix}  =  \frac{\partial}{\partial_{\gamma_q}} \begin{Bmatrix} \mathbf{A}^{-1}(\gamma) \begin{bmatrix}0 \\I\end{bmatrix} \end{Bmatrix}& = & -\mathbf{A}^{-1}(\gamma) \, \partial_{\gamma_q} \bigl \{ \mathbf{A}(\gamma) \bigr \} \, \mathbf{A}^{-1}(\gamma) I\\
                                                           & = & -\mathbf{A}^{-1}(\gamma) \, \partial_{\gamma_q} \bigl \{ \mathbf{A} (\gamma) \bigr \} \, \begin{bmatrix} u\\ U  \end{bmatrix},
\end{eqnarray*}
where
$$
\partial_{\gamma_q} \bigl \{ \mathbf{A} (\gamma) \bigr \} = \partial_{\gamma_q} \bigl \{ \mathbf{A}_{11} (\gamma) \bigr \} = \int_{B_q} \mathrm{d}x\, \nabla \phi_i \cdot \nabla \phi_j , \quad q=1,\ldots,N,
$$
as only the block $\mathbf{A}_{11}$ depends on admittivity. Separating the above as
$$
\partial_{\gamma_q} \begin{Bmatrix} u\\ U  \end{Bmatrix}  = \left[\begin{array}{c|c} \partial_{\gamma_q}u & \partial_{\gamma_q} U\end{array}\right]'
$$
and evaluating the upper part for all elements in the model yields the required matrix in vector concatenation form
\begin{equation}
\partial_{\gamma} u(\gamma) = \bigl [ \begin{array}{c|c|c|c} \partial_{\gamma_1} u(\gamma) & \partial_{\gamma_2} u(\gamma) & \ldots  & \partial_{\gamma_N} u(\gamma) \end{array} \bigr ],
\end{equation}
while $\partial_{\gamma_q} U$ are the elements of the Jacobian
matrix $\mathbf{J}$. Effectively the element of $\mathbf{K}^k$
matrix is given by
$$
\mathbf{K}^{k}_{r,j} =  - \int_{B_j} \mathrm{d}x \; \psi_j \; \sum_{l \in \mathrm{supp}(B_r)}  \nabla \phi_l \; \partial_{\gamma_r} u_l  \sum_{l \, \in \, \mathrm{supp}(B_j)} \overline{v}_l \nabla  \phi_l, \quad k=1,\ldots,m, \;\, r,j=1,\ldots,N
$$
with $v$ the adjoint field corresponding to the $k$th measurement
and $\partial_{\gamma_r} u$ the derivative of the $k$th direct
field with respect to $\gamma_r$.

\subsection{Newton's minimization method}

We propose solving the regularized problem (\ref{qrp}) using
Gauss-Newton's minimization method \cite{engl}. At a feasible
point $\delta \gamma_p$ the minimization cost function $\xi$ is
approximated by a second-order Taylor series \cite{hettlich}
\begin{equation}
\hat \xi(\delta \gamma) = \xi(\delta \gamma_p) + \partial_{\delta \gamma} \xi(\delta \gamma_p)(\delta \gamma - \delta \gamma_p)
+  \frac{1}{2}(\delta \gamma - \delta \gamma_p)'\partial_{\delta \gamma \delta \gamma}\xi(\delta \gamma_p) (\delta \gamma - \delta \gamma_p),
\end{equation}
where applying first-order optimality conditions $ \partial_{\delta \gamma} \hat \xi(\delta \gamma) = 0$ yields the linear system
$$
\partial_{\delta \gamma} \hat \xi(\delta \gamma_p) = - \partial_{\delta \gamma \delta \gamma}\hat \xi(\delta \gamma_p) (\delta \gamma - \delta \gamma_p).
$$
From (\ref{quadmis}), let the $k$th residual function be
$$
r_k(\delta \gamma) = c_k^{-1/2} \Bigl( {\delta \zeta}_k - \sum_{j=1}^N \mathbf{J}_{k,j} \delta \gamma_j - \sum_{j=1}^N \delta \gamma_j \sum_{l=1}^N \mathbf{K}^k_{j,l} \delta \gamma_l \Bigr ),
$$
such that $Q(\delta \gamma) = \|r(\delta \gamma)\|^2$, then the
cost gradient $\partial_{\delta \gamma} \hat \xi(\gamma_p)$ and
Hessian $\partial_{\delta \gamma \delta \gamma}\hat \xi(\delta
\gamma_p) $ are expressed as
\begin{eqnarray}\label{grad}\nonumber
\partial_{\delta \gamma} \hat \xi(\delta \gamma_p) &=& \partial_{\delta \gamma} r(\delta \gamma_p)' r(\delta \gamma_p)  + \alpha \mathbf{C}_{\gamma}^{-1} \delta \gamma_p,\\ \label{hess}\nonumber
\partial_{\delta \gamma \delta \gamma}\hat \xi(\delta \gamma_p) &=& \partial_{\delta \gamma} r(\delta \gamma_p)' \partial_{\delta \gamma} r(\delta \gamma_p) + \alpha \mathbf{C}_\gamma^{-1}
\end{eqnarray}
for $r(\delta \gamma) = \begin{bmatrix} r_1(\delta \gamma), &
\ldots, & r_m(\delta \gamma)\end{bmatrix}'$, and assuming a
Tikhonov-type regularization function $\mathcal{G}(\delta \gamma)
= \alpha {\delta \gamma}' \mathbf{C}_\gamma^{-1} \delta \gamma$,
with $\mathbf{C}_\gamma^{-1}$ positive semidefinite and $\alpha$ a
positive regularization parameter. The Jacobian of the residual
$\partial_{\delta \gamma} r(\partial \gamma_p) \in \mathbb{C}^{m
\times N}$ is then formed using the vectors
$$
\partial_{\delta \gamma_l} r_k(\delta \gamma) = - c_k^{-1/2} \mathbf{J}_{k,l} - c_k^{-1/2} \sum_{j=1}^N \Bigl ( \mathbf{K}^k_{l,j} + \mathbf{K}^k_{j,l}\Bigr ) \delta \gamma_j, \quad l=1,\ldots,N,
$$
evaluated at $\delta \gamma_p$ like
$$
\partial_{\delta \gamma} r(\delta \gamma_p) =  \begin{bmatrix} \partial_{\delta \gamma} r_1(\delta \gamma_p) & | & \partial_{\delta \gamma} r_2(\delta \gamma_p) & | & \ldots & | & \partial_{\delta \gamma} r_m(\delta \gamma_p) \end{bmatrix}'.
$$
If $\partial_{\delta \gamma \delta \gamma}\hat \xi(\delta
\gamma_p)$ is full rank and positive definite the solution can be
computed iteratively using Newton's algorithm
\begin{equation}\label{mainit}
\delta \gamma_{p+1} = \delta \gamma_p - \partial_{\delta \gamma \delta \gamma}^{-1}\hat \xi(\delta \gamma_p) \, \partial_{\delta \gamma} \hat \xi(\delta \gamma_p), \qquad p=0,1,2,\ldots
\end{equation}
Using standard arguments from the convergence analysis of Newton's
method on convex minimization it is easy to show convergence as in
\cite{engl}, \cite{Lechleiter}
\begin{equation}\label{convan}
\hat \xi(\delta \gamma_p) > \hat \xi(\delta \gamma_{p+1}) \geq \|\eta\|, \quad \|\delta \gamma^* - \delta \gamma_p\| \geq \|\delta \gamma^* - \delta \gamma_{p+1}\|, \quad p=0,1,\ldots,
\end{equation}
however a convergence in the sense of the discrepancy principle is
more appropriate as the data are likely to contain noise
\cite{Kaipio_Somer_book}.
\begin{corollary}
Initializing the quadratic regression iteration (\ref{mainit})
with $\delta \gamma_0 = 0$ yields a first iteration that coincides
with the linear regularized regression estimator
\begin{equation}\label{tik2}
\delta \gamma_1 = \bigl ( \mathbf{J}' \mathbf{C}_\eta^{-1} \mathbf{J} + \alpha \mathbf{C}_\gamma^{-1} \bigr )^{-1} \mathbf{J}' \mathbf{C}_\eta^{-1} \delta \zeta
\end{equation}
\end{corollary}
\begin{proof}
The proof is by substitution of the residual and its Jacobian at
$\delta \gamma_0=0$ into the expressions for the gradient and
Hessian of the cost function. In particular for $r(\delta
\gamma_0) =  \mathbf{C}_\eta^{-1/2} \delta \zeta$ and
$\partial_{\delta \gamma} r(\delta \gamma_0) =
\mathbf{C}_\eta^{-1/2} \mathbf{J}$, iteration (\ref{mainit})
yields the result.
\end{proof}

Combining the convergence remarks of (\ref{convan}) with the
corollary above, we assert that for $p>1$ the quadratic regression
iterations should converge in a solution whose error does not
exceed that of the linear regression problem (\ref{rlsq}). Suppose
now that at a certain iteration $p$ the value of the residual
$r(\delta \gamma_p)$ converges to the level of noise $\|\eta\|$.
Then according to the discrepancy principle one updates the
admittivity estimate as $\gamma_{p+1} = \gamma_p + \delta
\gamma_p$ and thereafter the definitions of $\mathbf{J}$ and
$\mathbf{K}^k$, and then proceeds to the next iteration.
Effectively, the resulting scheme can be expressed as a
Newton-type algorithm. \vspace*{5pt}
\begin{center}
\line(1,0){430}
\end{center}
\begin{enumerate}\label{algo}
\vspace*{5pt} \item Given data $\zeta \in \mathbb{C}^m$ with noise
level $\|\eta\|$ and a finite domain $B_h$ with unknown
admittivity $\gamma^* \in \mathbb{C}^N$ \item Set $q=0$, choose
initial admittivity distribution $\gamma_0$, \item For
$q=1,2,\ldots \quad$ (Exterior iterations) \item Compute data
$\delta \zeta = \zeta - \mathcal{E}(\gamma_{q-1})$, and matrices
$\mathbf{J} \in \mathbb{C}^{m \times N}$, $\mathbf{K}^k \in
\mathbb{C}^{N \times N}$, for $\gamma_{q-1}$, and $k=1,\ldots,m$,
\begin{enumerate}
\item Set $p=0$, $\delta \gamma_p=0$, \item For $p=1,2,\ldots
\quad$ (Interior iterations) \item Compute update $$\delta
\gamma_{p} =  \delta \gamma_{p-1} - \tau_p \partial_{\delta \gamma
\delta \gamma}^{-1}\hat \xi(\delta \gamma_{p-1}) \,
\partial_{\delta \gamma} \hat \xi(\delta \gamma_{p-1}), \; \tau_p
>0,$$ \item End $p$ iterations \item Compute update $$\gamma_{q} =
\gamma_{q-1} + \tau_q \delta \gamma_p, \;  \tau_q>0,$$
\end{enumerate}
\item End $q$ iterations
\end{enumerate}\vspace*{-5pt}
\begin{center}
\line(1,0){430}
\end{center}
\vspace*{5pt} In performing the outer iterations, a complication
will likely arise in that a certain update admittivity change
$\delta \gamma_p$ may cause the real and/or imaginary components
of $\gamma_{q+1}$ to become zero or negative. This of course
violates a physical restriction on the electrical properties of
the media, and the solution cannot be admitted. For this reason
the problem of (\ref{qrp}) should be posed as a linearly
constrained problem
$$
 {\delta \gamma}^* = \arg\min_{\gamma_q > \delta \gamma} \xi(\delta \gamma),
$$
at each $\gamma_q$. A convenient heuristic to prevent this
complication is by adjusting the step sizes $\tau_q, \tau_p$ until
the above inequality is satisfied \cite{eidors}, \cite{vauhkonen}.
Note also, that the above methodology makes no explicit
assumptions on the type of the regularization functional
$\mathcal{G}(\gamma)$, aside its differentiability, thus we
anticipate it can be also be implemented in conjunction with total
variation and $\ell_1$-type regularization \cite{Borsic_tv} as
well as the level sets method \cite{Dorn}.

\section{Numerical results}\label{numerics}

To test the performance of the proposed algorithm we perform some
numerical simulations using two-dimensional models, although the
extension to three dimensions follows in a trivial way. In this
context we consider a rectangular domain $B  =
[-16,16]\times[0,-32] \subset \mathbb{R}^2$, with $L=30$ point
electrodes attached at its boundary in a borehole and surface
arrangement as shown in figures \ref{fig4} and \ref{fig5}. As a
first test case the domain is assumed to have an unknown target
conductivity $\gamma^*$ whose real and imaginary components are
functions with respective bounds $1.46 \leq \sigma^* \leq 5.60$
and $0.74 \leq \omega \epsilon^* \leq 3.90$. To compute the
measurements we consider 15 pair drive current patterns $I^d$,
$d=1,\ldots,L/2$,  yielding a vector of $m=390$ linearly
independent voltage measurements $\zeta \in \mathbb{C}^m$. The
forward problem is approximated using the finite element method
outlined in the previous section, and to the measurements we add a
Gaussian noise signal of zero mean and positive definite
covariance matrix $\mathbf{C}_\eta = 10^{-5} \max |\zeta|\,
\mathbf{I}$, where $\mathbf{I}$ is the identify matrix. For the
forward problem we use a finite dimensional model $B_f$ comprising
$n=1701$ nodes connected in $N=3144$ linear triangular elements.
All other computations are performed on a coarser grid $B_i$ with
$n=564$ nodes and $N=1038$ elements. The two finite models are
nested, hence for any function $\gamma$ approximated on $B_i$ with
expansion coefficients $\gamma_i$ there exists a projection
$\gamma_f = \Pi \gamma_i$, mapping it onto $B_f$. To reconstruct
the synthetic data we assume an initial homogeneous admittivity
model $\gamma_0 = 3.90 + 2.40 i$ which coincides with the mean
value of  $\gamma^*$, a methodology adopted from \cite{helsinki}.

At the initial admittivity guess $\gamma_0$ we approximate the
potential $u(\gamma^*)$ using the zeroth-order and first-order
Taylor series $u(\gamma_0)$ and $u(\gamma_0) + \partial_\gamma
u(\delta \gamma_0)(\gamma - \gamma_0)$ respectively. The
normalized approximation errors are illustrated at the top of
figure \ref{fig1} next to those of the error in the induced
potential gradient as this is involved in the computation of the
$\mathbf{K}^{k}$ matrices for $k=1,\ldots,390$. The results show
that the linear approximation sustains a smaller error in both
quantities and at all applied current patterns. In the same figure
we also plot the measurement perturbations $\delta \zeta = \zeta -
\mathcal{E}(\gamma_0)$ versus the linear and the quadratic
predictions to demonstrate that the proposed quadratic regression
will fit the noisy measurements at a smaller error. In particular,
the quadratic and linear misfit cost functions in (\ref{quadmis})
and (\ref{linmis}) are evaluated at $Q(\delta \gamma_0) = 0.06$
and $\Lambda(\delta \gamma_0) = 0.13$, where $\delta \gamma_0 =
\gamma^* - \gamma_0$. Notice the impact of the second-order term,
that brings the norm of the data misfit to about half of that of
the linear case.

\begin{figure}
\begin{center}
\epsfig{file=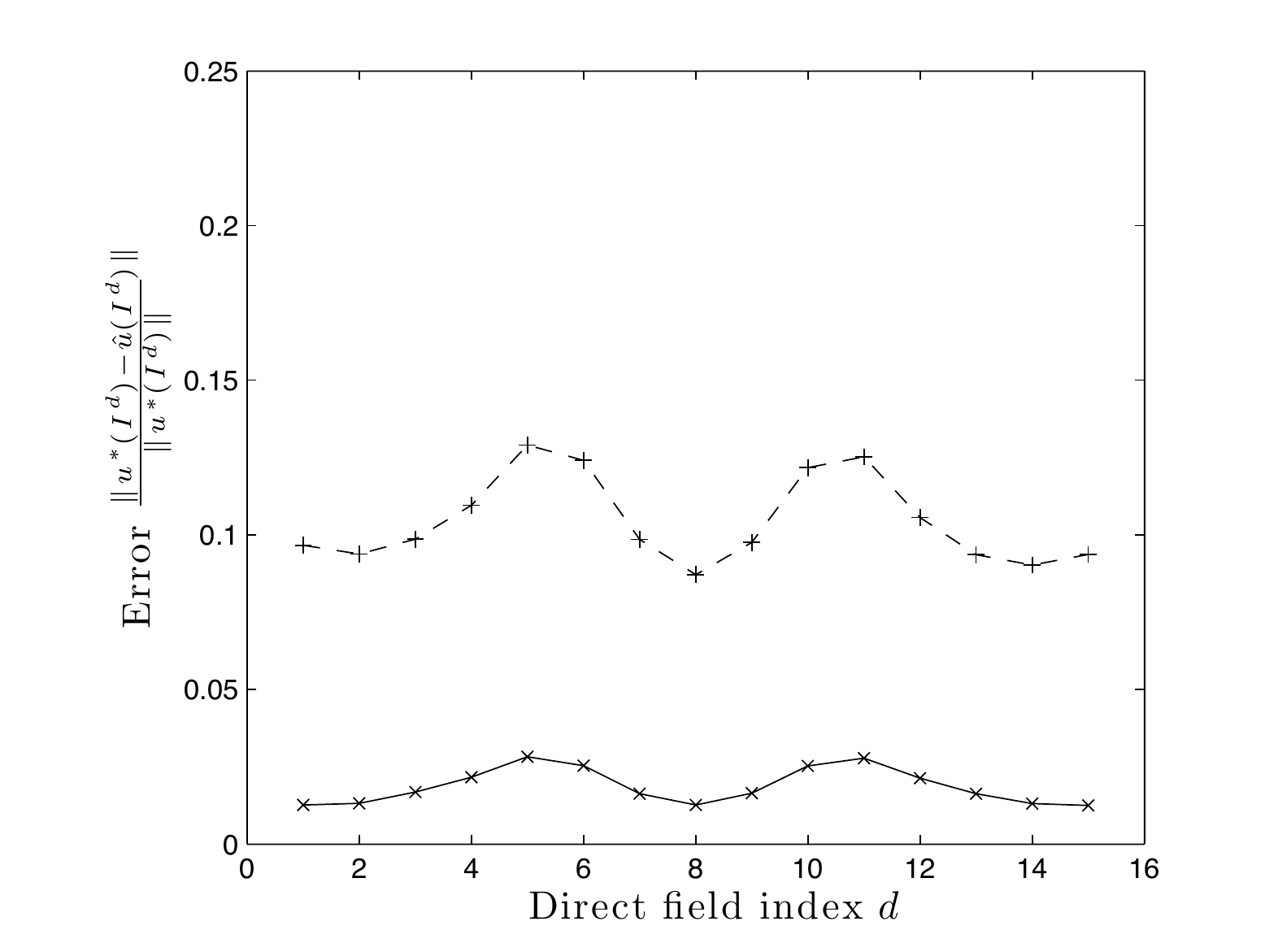,height=5cm,width=7cm} \epsfig{file=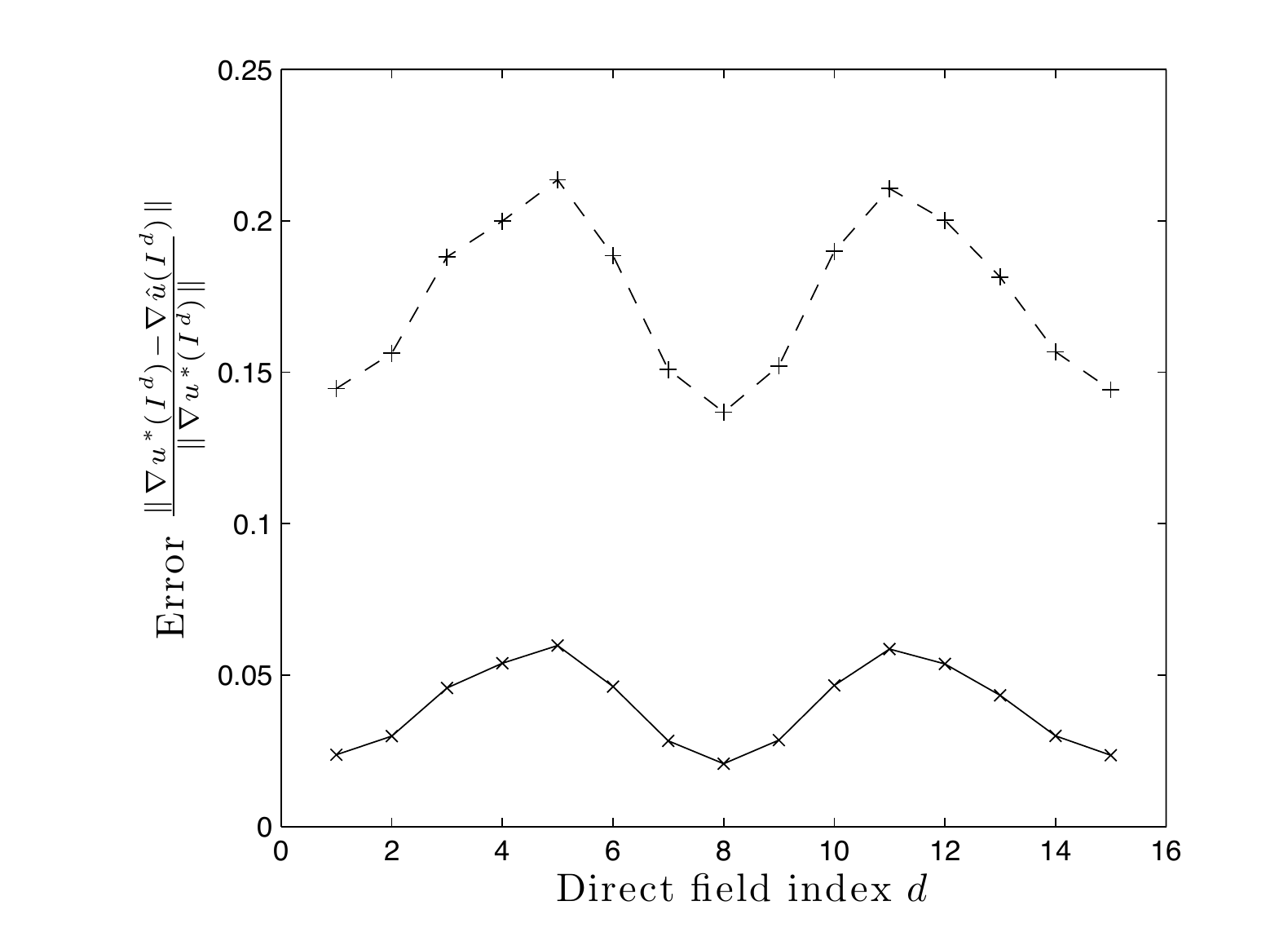,height=5cm,width=7cm}\\
\epsfig{file=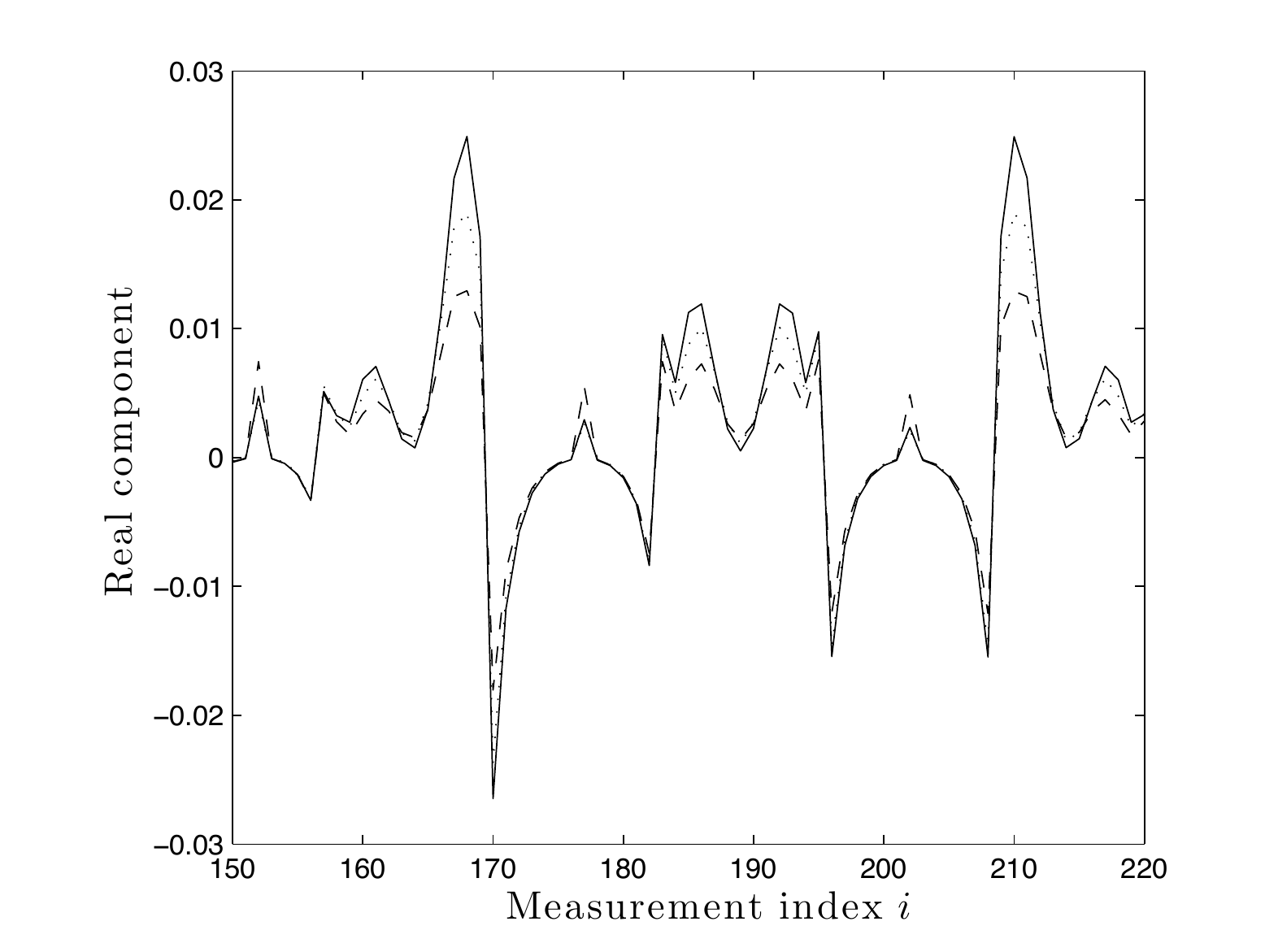,height=5cm,width=7cm} \epsfig{file=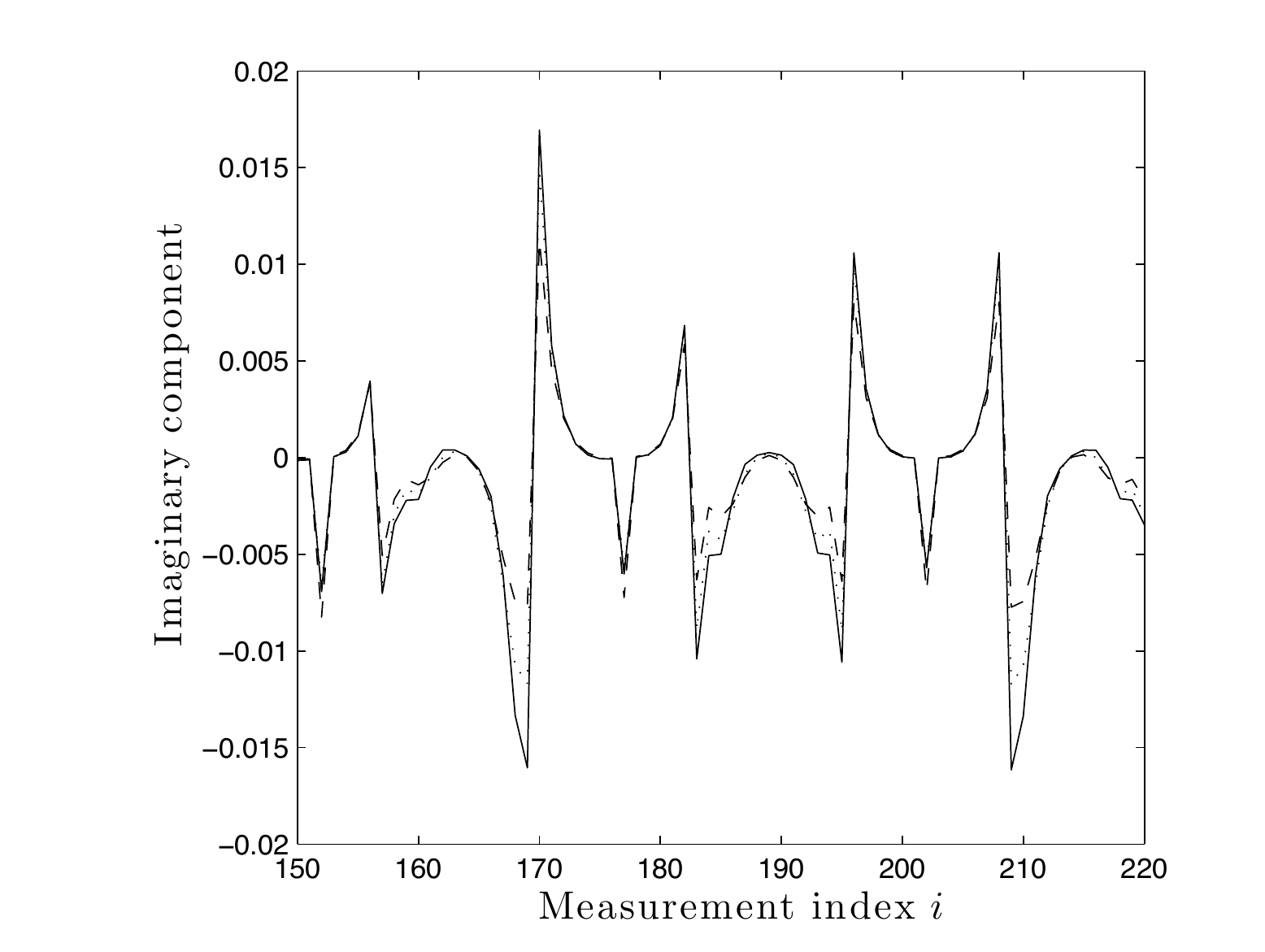,height=5cm,width=7cm}
\caption{At the top row, the normalized errors in the electric potential field approximation and its gradient, assuming zeroth-order (dashed line with $+$ markers) and first-order (solid line with $\times$ markers) Taylor series approximations of $u(\gamma^*,I^d)$ direct fields. In both cases the errors with the linear approximation are lower. Second row, the quality of the linear and quadratic approximations in predicting the nonlinear change in the boundary data $\delta \zeta$. The solid line denotes $\delta \zeta_i$, the dashed $j_i'\delta \gamma$ and the dotted $j_i'\delta \gamma + \sum_{i=1}^m \delta \gamma' \mathbf{K}^i \delta \gamma$, over the interval $i=150,\ldots,220$. The corresponding data misfit norms are  $0.057$ for the quadratic approximation $Q(\delta \gamma^*)$ and $0.123$ for the linear $\Lambda(\delta \gamma^*)$, assuming no additive noise. With the prescribed additive noise these values change to $0.062$  and  $0.126$ respectively.}
\label{fig1}
\end{center}
\end{figure}

\begin{figure}
\begin{center}
\epsfig{file=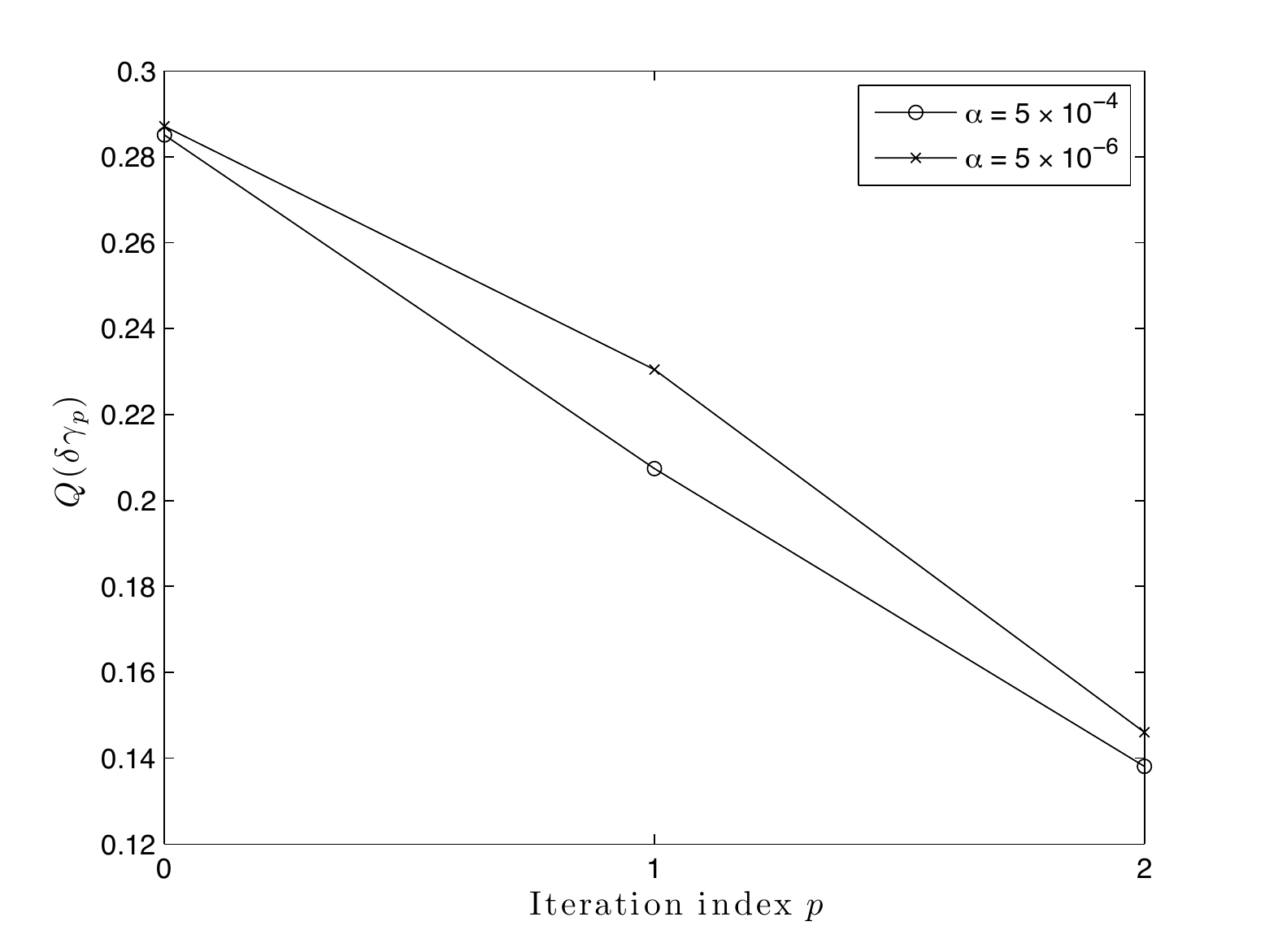,height=5cm,width=7cm} \epsfig{file=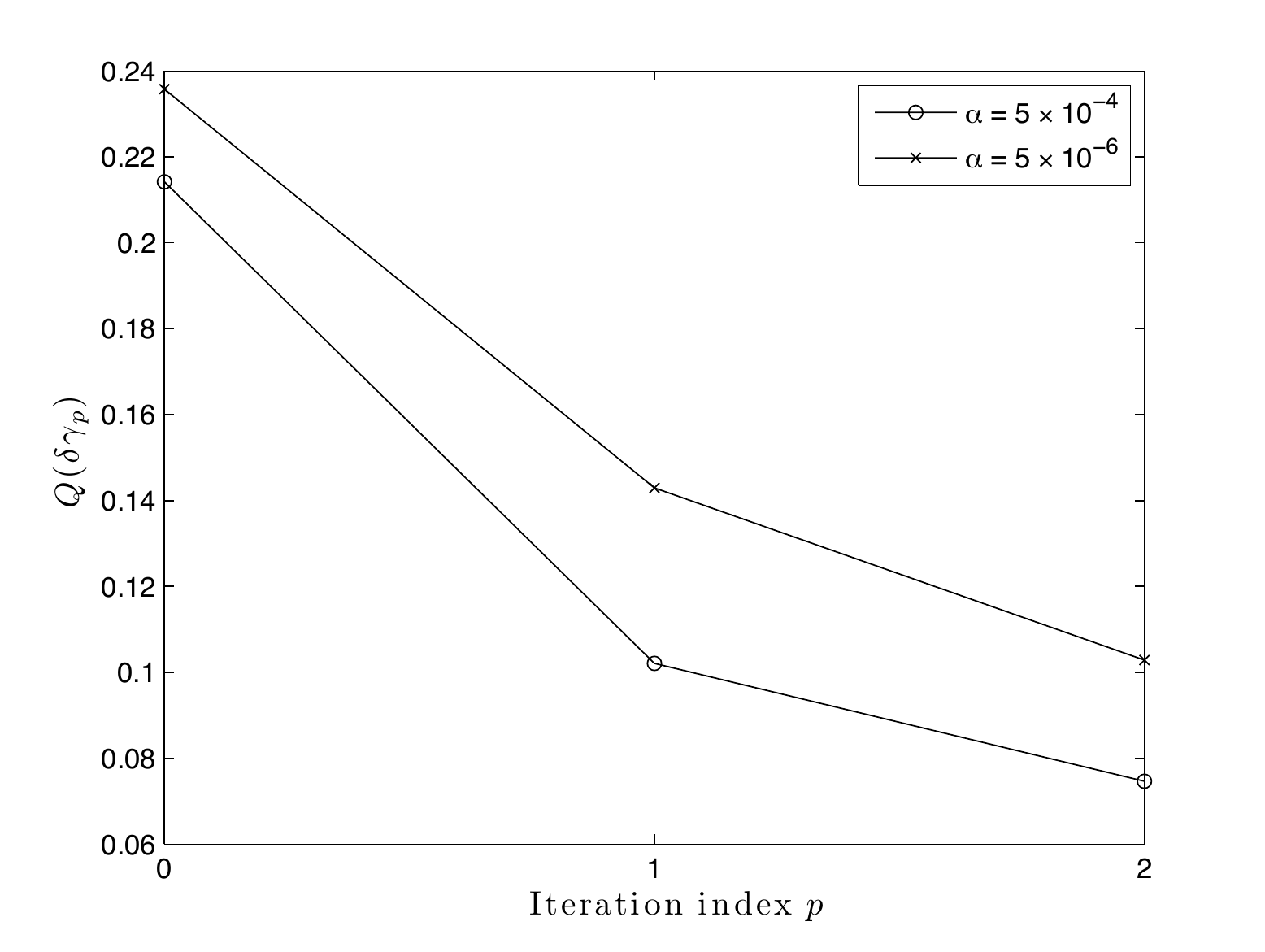,height=5cm,width=7cm}
\caption{Indicative convergence of the proposed method, in terms of minimizing the quadratic misfit error $Q(\delta \gamma_p)$  for two different values of the regularization parameter $\alpha$. Left the results during the first exterior iteration $q=1$, and right the corresponding values for $q=2$. In these results, $\delta \gamma_0=0$, $\delta \gamma_1$ coincides with the Tikhonov solution, and $\delta \gamma_2$ is the regularized quadratic regression solution. Notice that the quadratic regression solution has lower data misfit errors in both GN iterations. Between the first and second exterior iteration the admittivity increment was scaled to preserve positivity, hence the apparent discontinuity in the error reduction.}
\label{fig2}
\end{center}
\end{figure}

\begin{figure}
\begin{center}
\epsfig{file=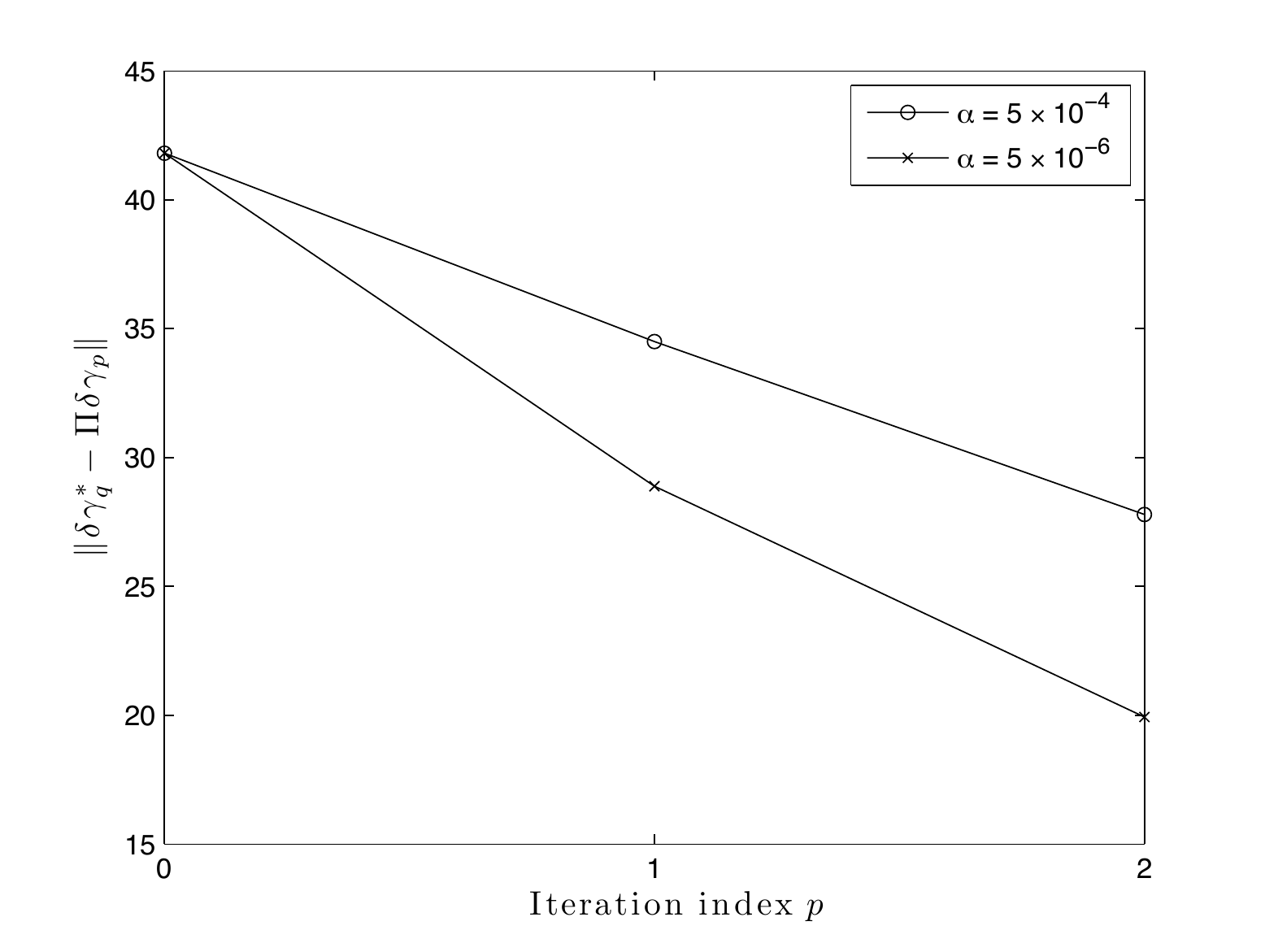,height=5cm,width=7cm} \epsfig{file=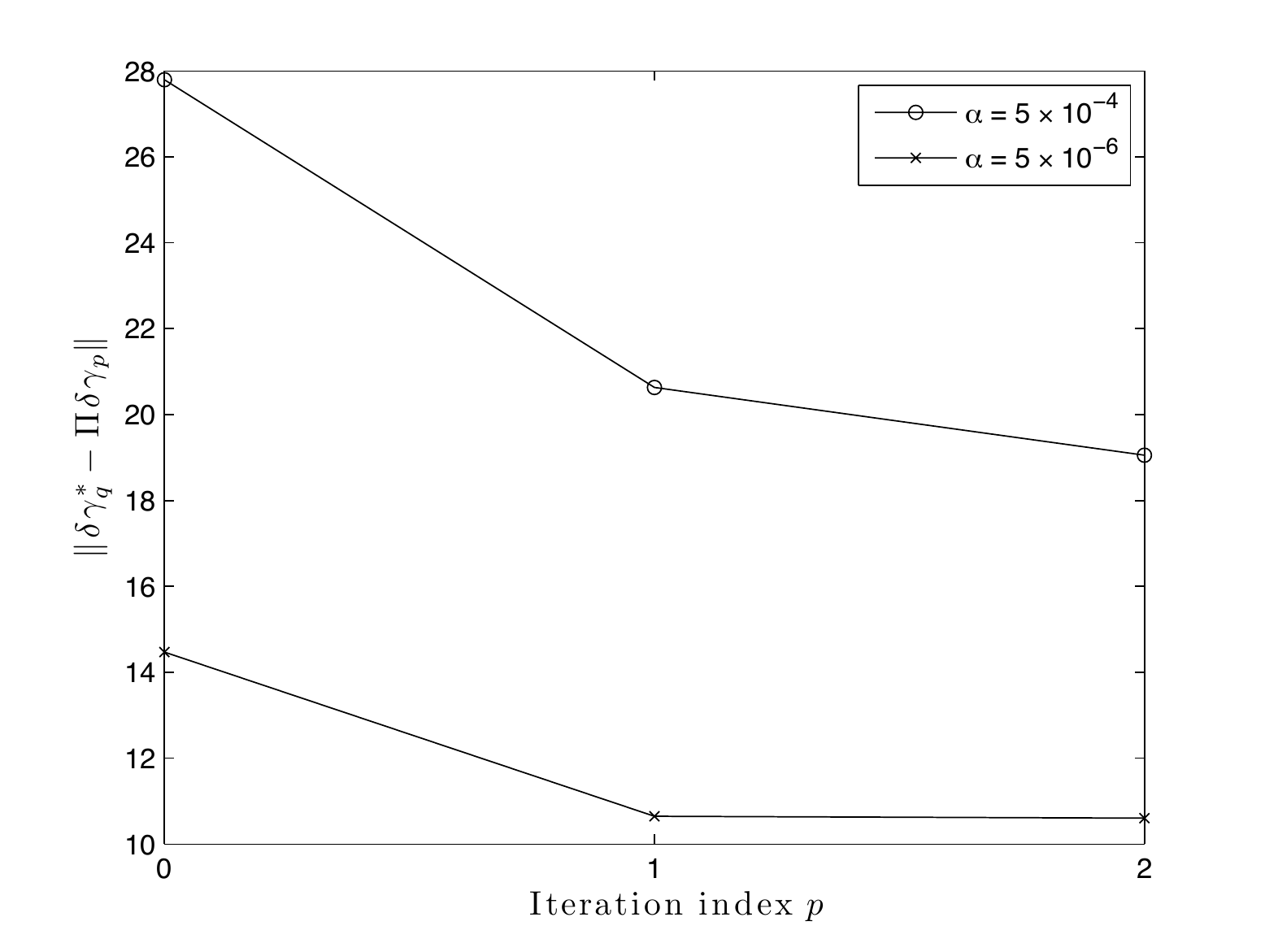,height=5cm,width=7cm}
\caption{Indicative convergence of the proposed method, in terms of minimizing the the image error $\|\delta \gamma_q^* - \Pi \delta \gamma_p\|$  for two different values of the regularization parameter $\alpha$. Left the results during the first exterior iteration $q=1$, and right the corresponding values for $q=2$. In these figures $\delta \gamma_0 = 0$, $\delta \gamma_1$ coincides with the Tikhonov solution, and $\delta \gamma_2$ is the regularized quadratic regression solution. Notice that the quadratic regression solution maintains lower image errors at each external iteration. Between the first and second exterior iteration the admittivity increment was scaled to preserve positivity, hence the apparent discontinuity in the error reduction.}
\label{fig2b}
\end{center}
\end{figure}

\begin{figure}
\centering \epsfig{file=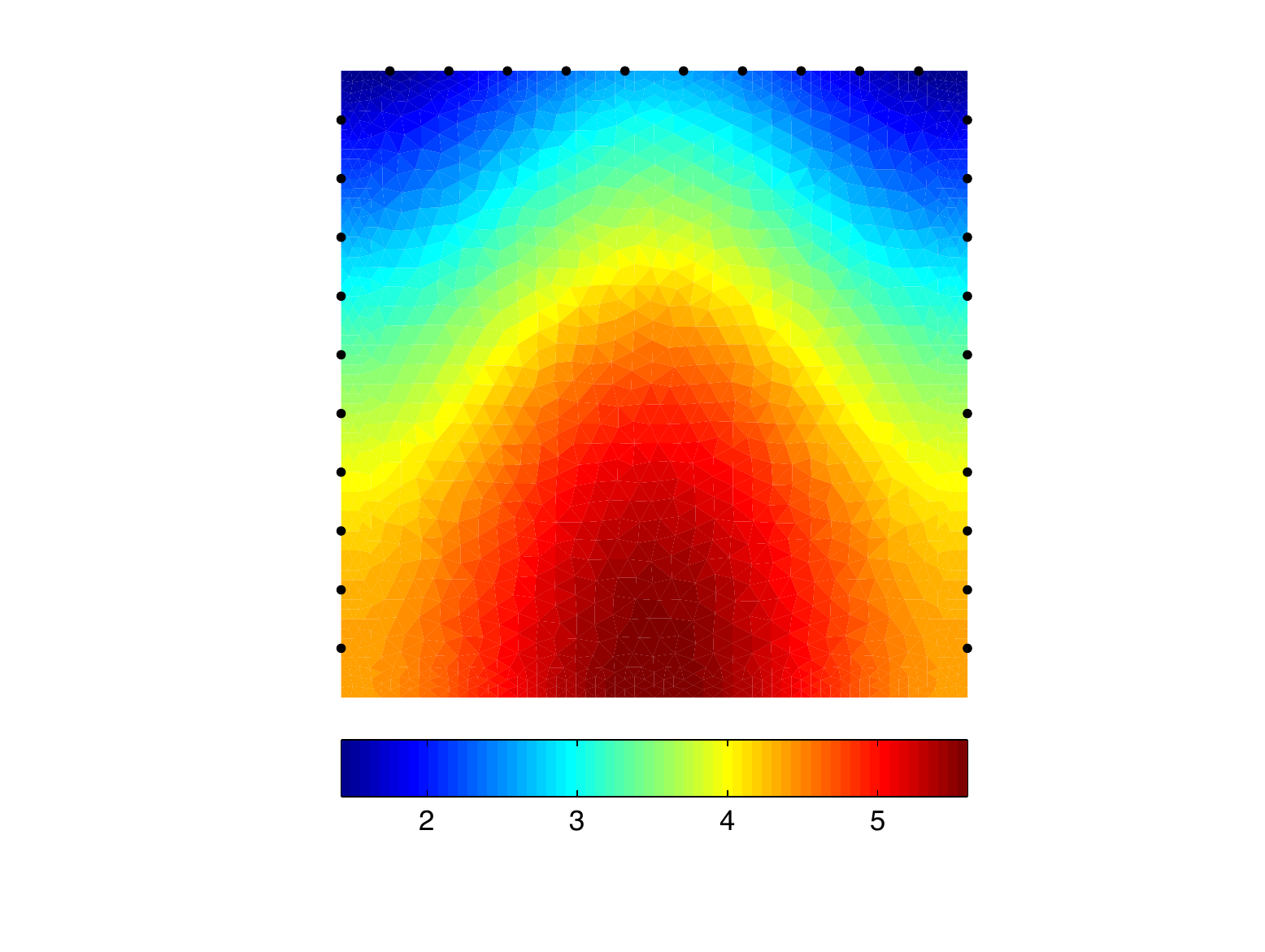,height=4cm,width=5.33cm}\\
\epsfig{file=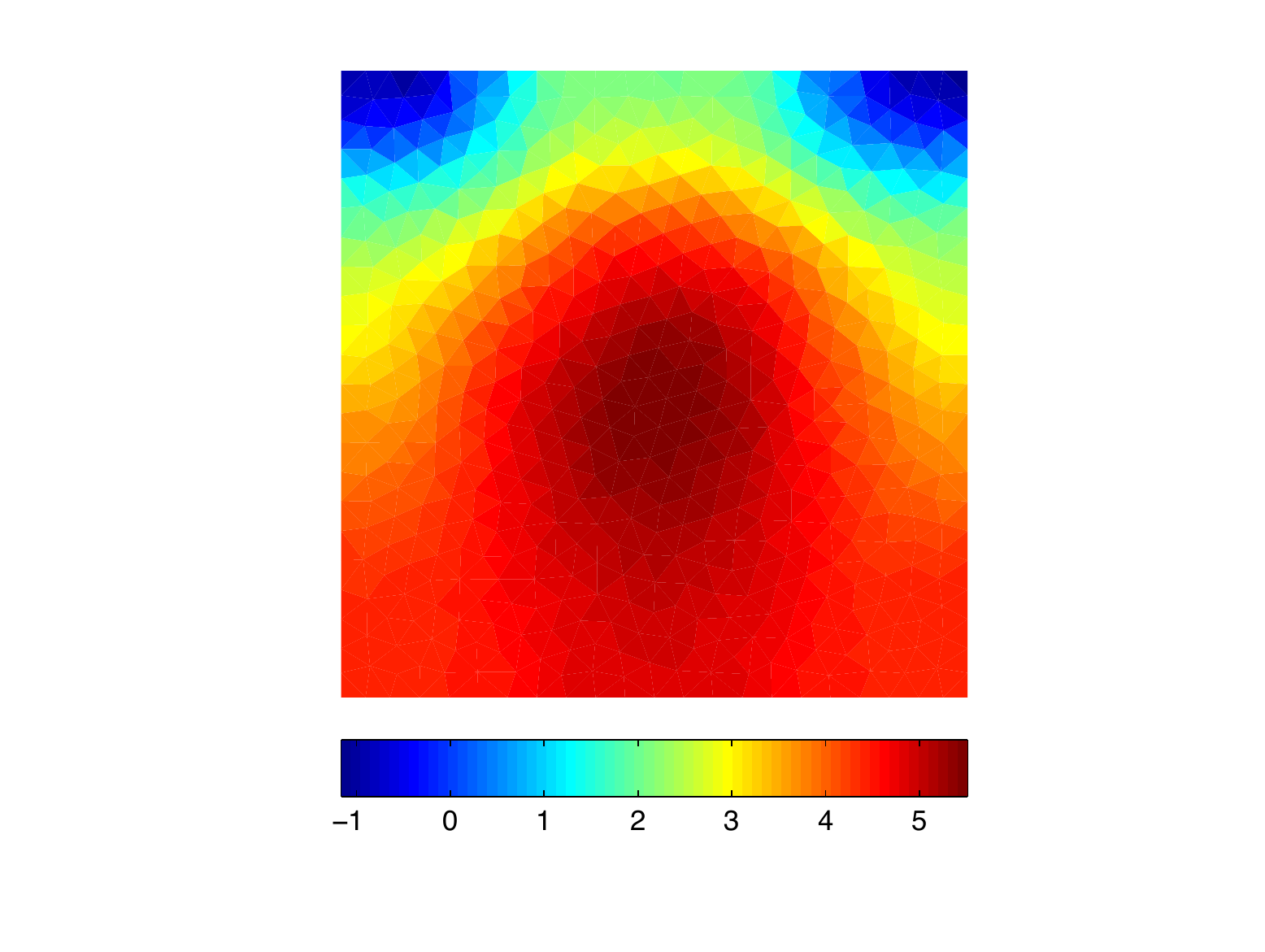,height=4cm,width=5.33cm} \hspace*{-20pt}\epsfig{file=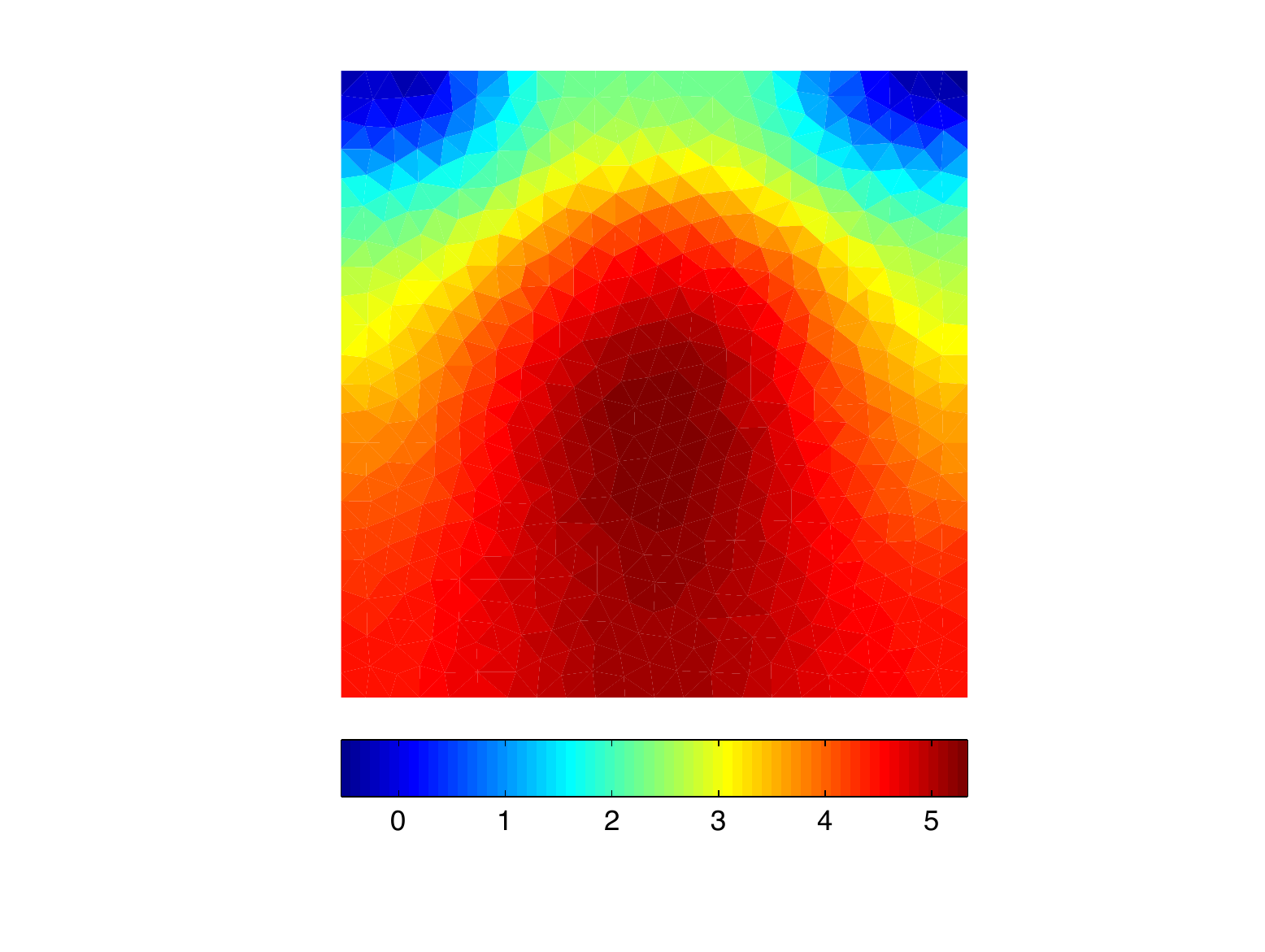,height=4cm,width=5.33cm} \hspace*{-20pt} \epsfig{file=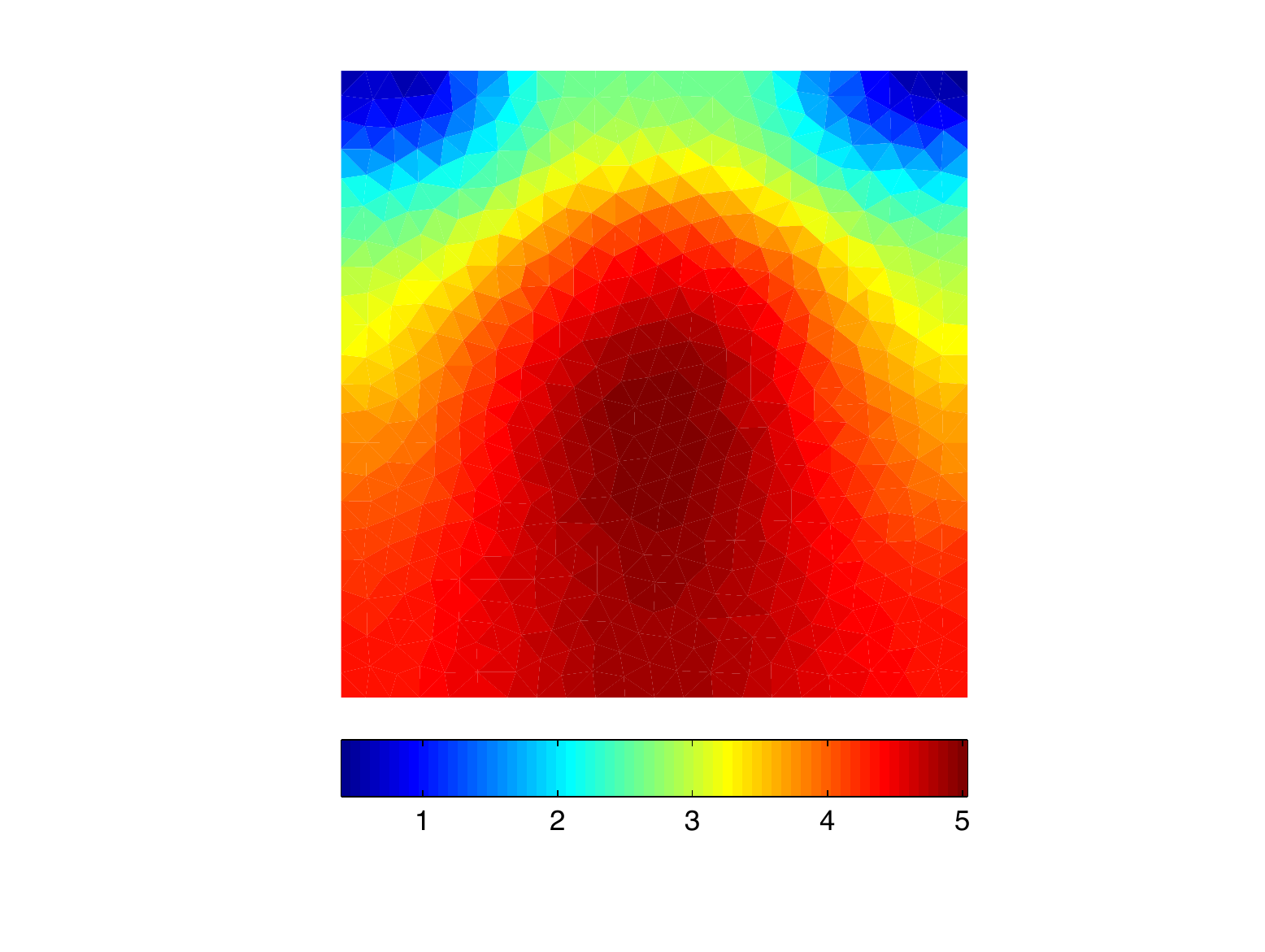,height=4cm,width=5.33cm}\\
\epsfig{file=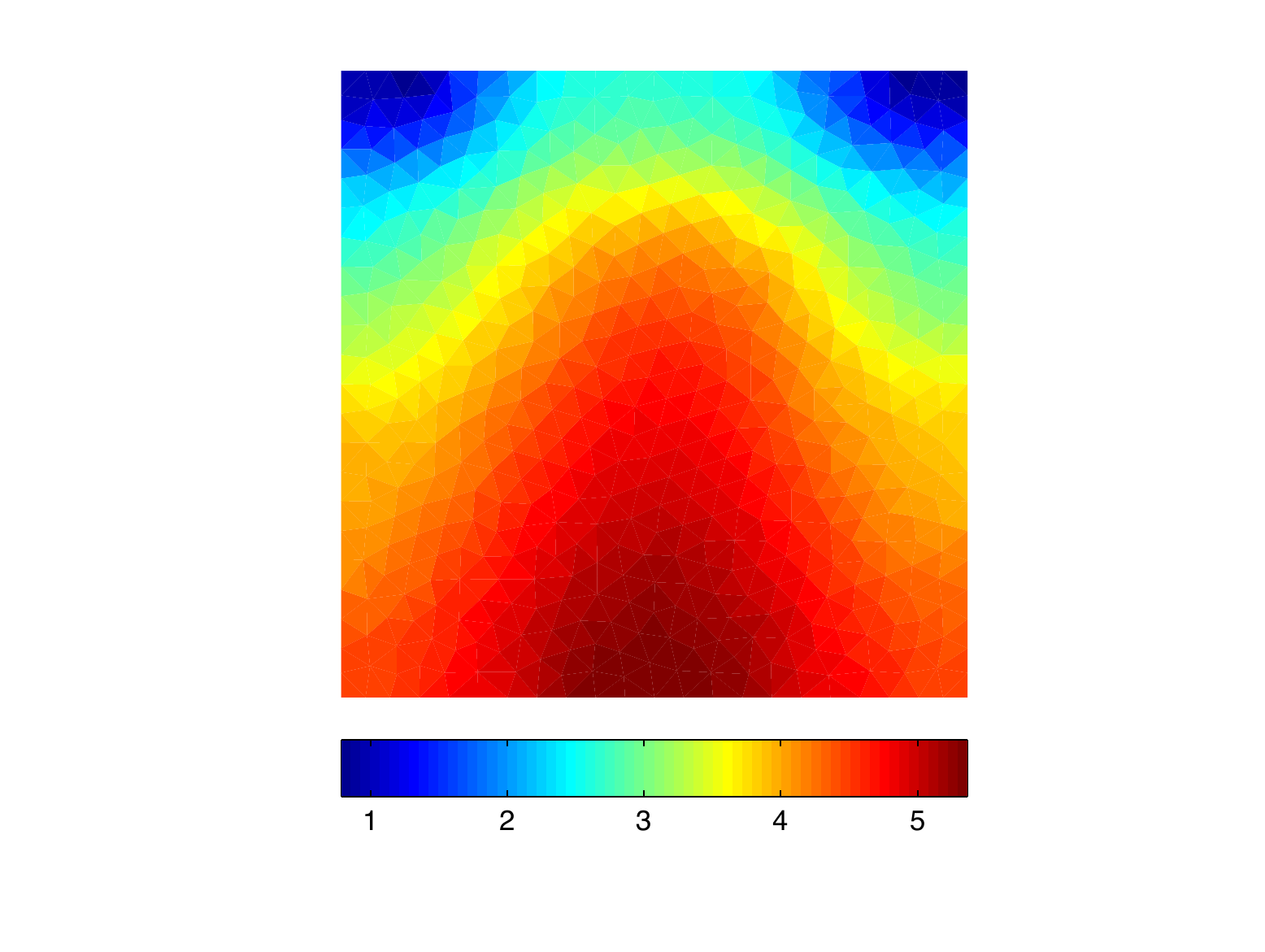,height=4cm,width=5.33cm} \hspace*{-20pt} \epsfig{file=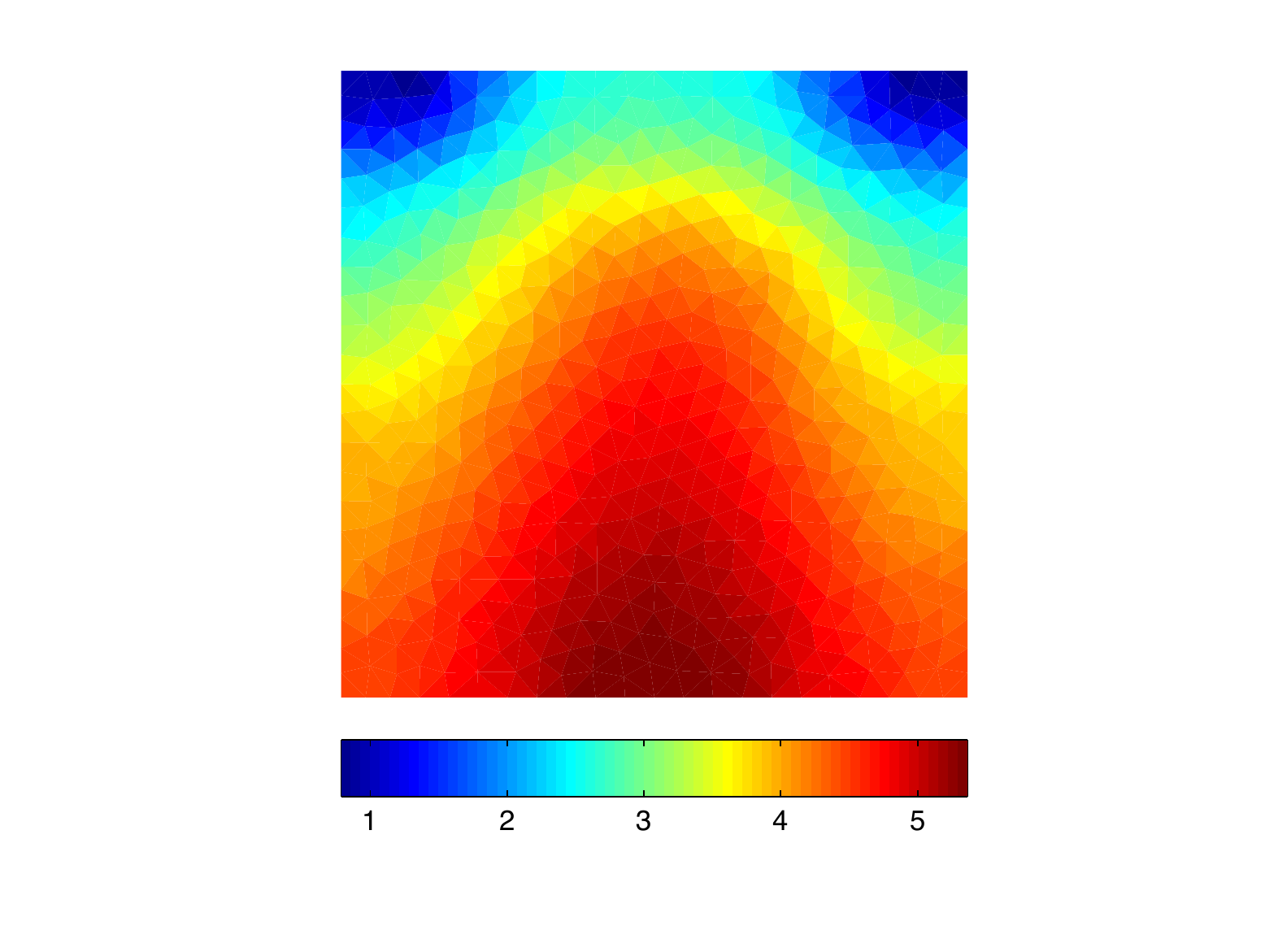,height=4cm,width=5.33cm}  \hspace*{-20pt} \epsfig{file=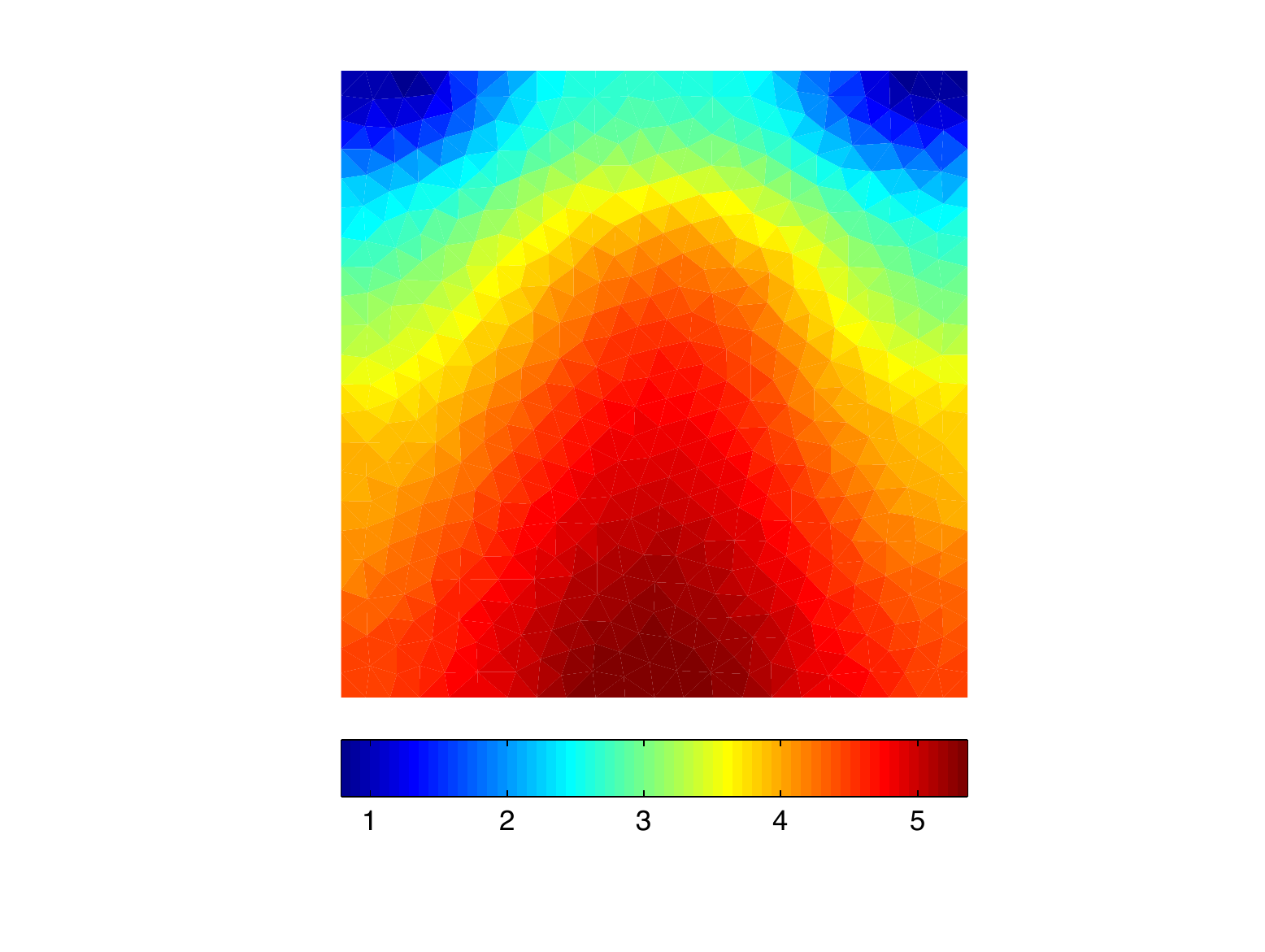,height=4cm,width=5.33cm}
\caption{At the top, the simulated target conductivity $\sigma^*$ profile on $B_f$, as used in the first test example and the arrangement of the electrodes. In the second row, from left to right, the respective images resulted from first exterior iteration using $\alpha = 5 \times 10^{-6}$, namely the real components of $\gamma_0 + \tau \delta \gamma_1$, $\gamma_0 + \tau \delta \gamma_2$, and $\gamma_1$ on $B_i$. Similarly at the bottom row,  the respective images from the second exterior GN iteration,  $\gamma_1 + \tau \delta \gamma_1$, $\gamma_1 + \tau \delta \gamma_2$, and $\gamma_2$, using the same value of $\alpha$.}
\label{fig4}
\end{figure}

\begin{figure}
\centering \epsfig{file=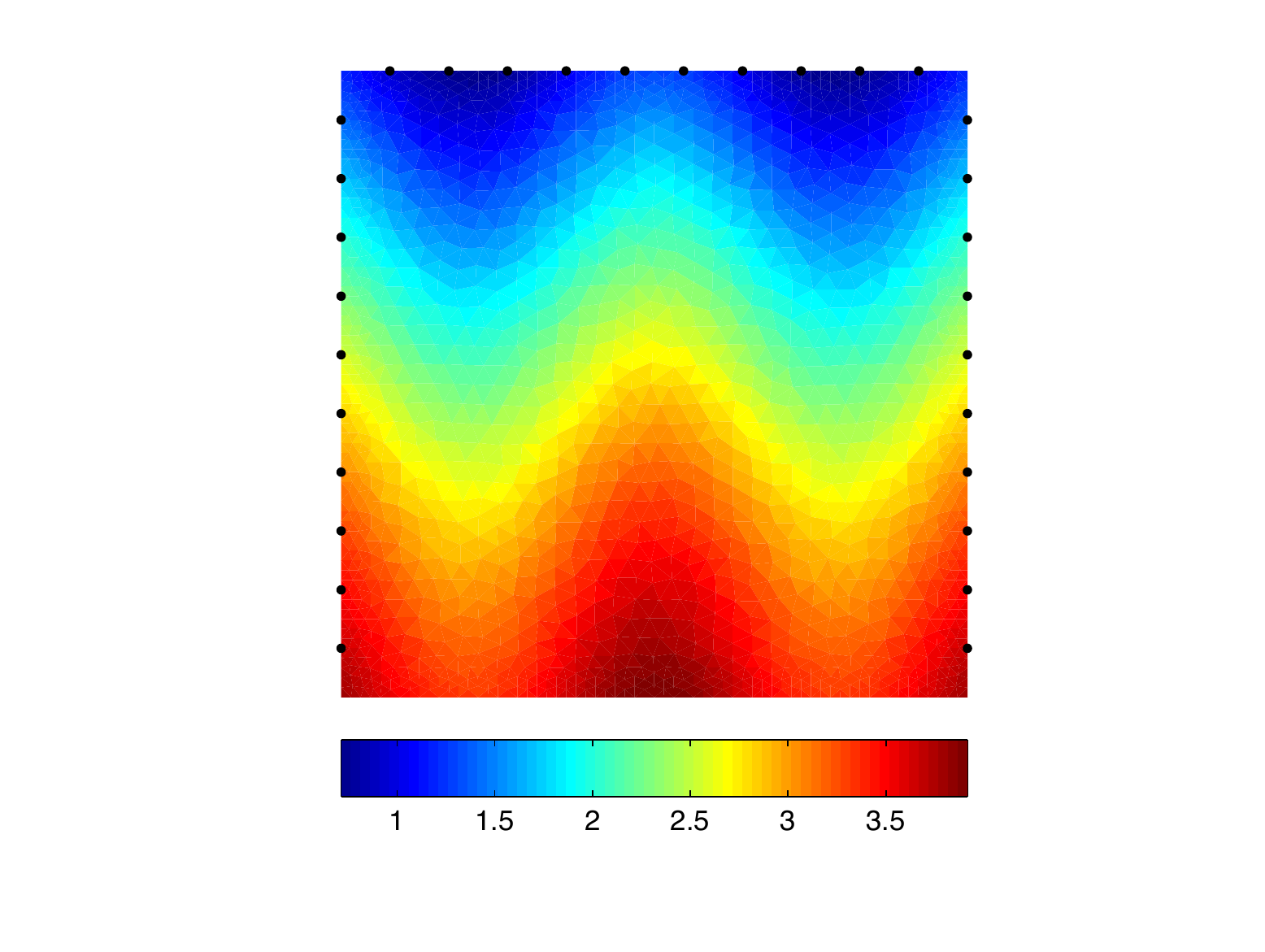,height=4cm,width=5.33cm}\\
\epsfig{file=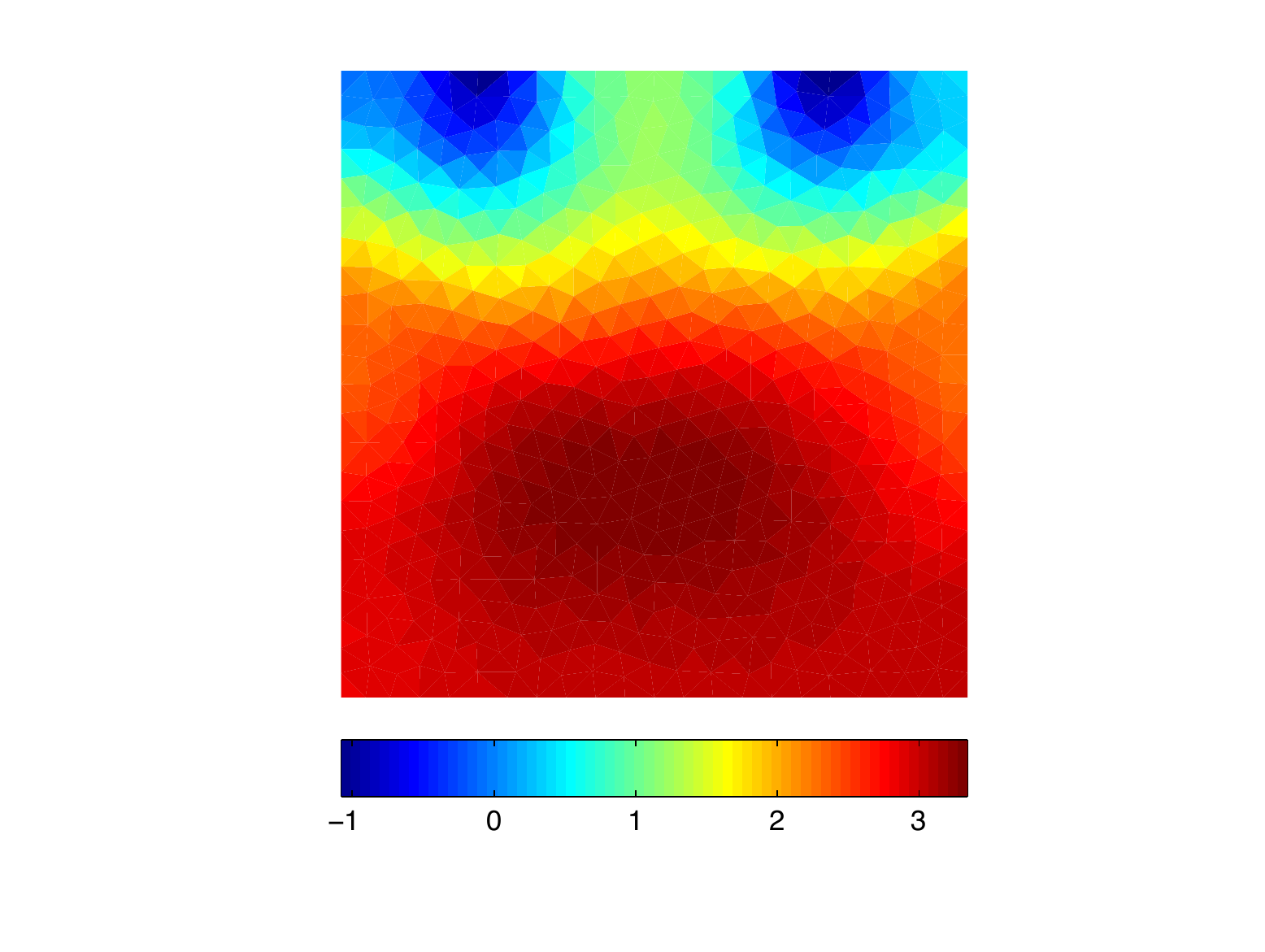,height=4cm,width=5.33cm} \hspace*{-20pt}\epsfig{file=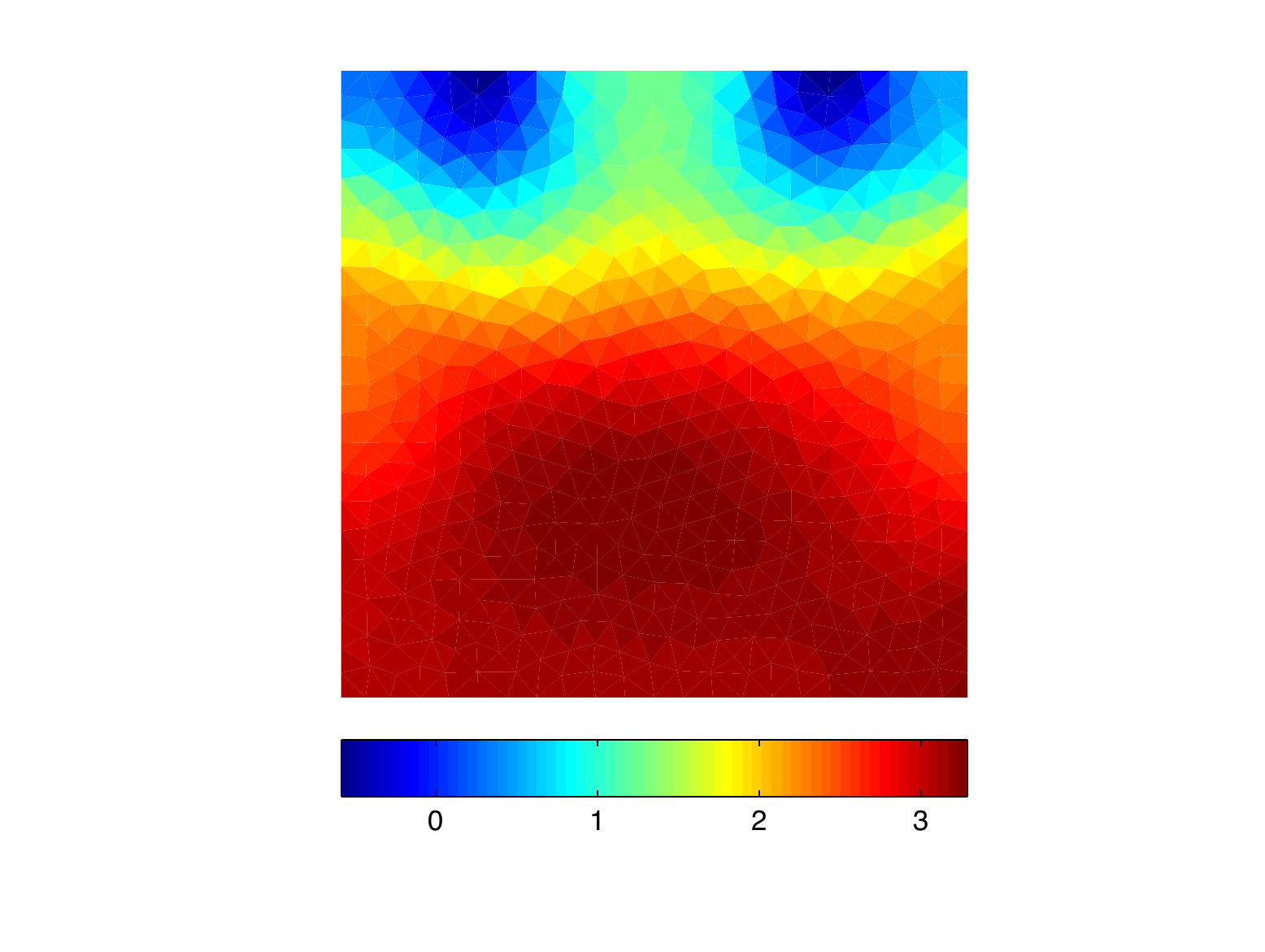,height=4cm,width=5.33cm} \hspace*{-20pt} \epsfig{file=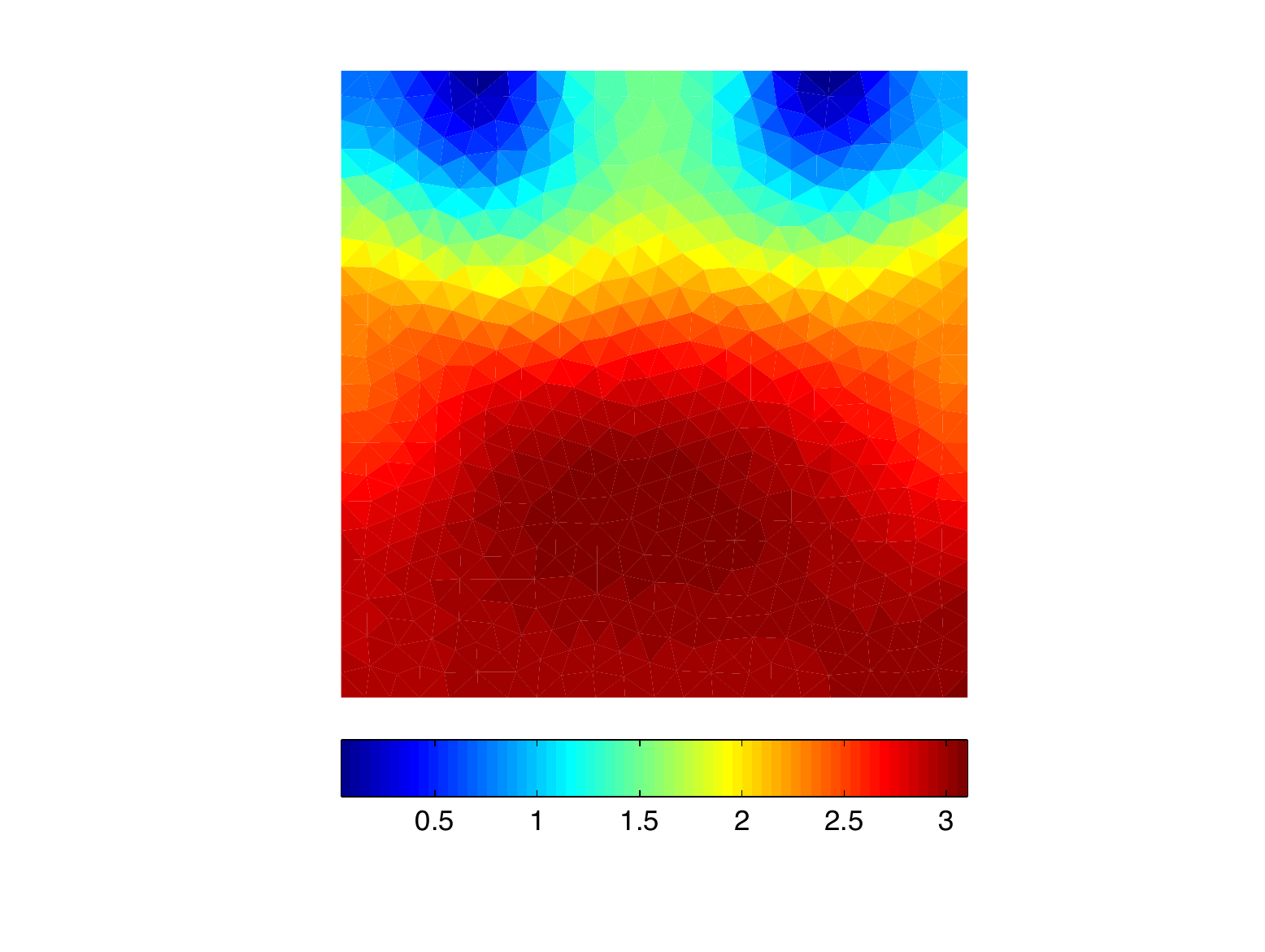,height=4cm,width=5.33cm}\\
\epsfig{file=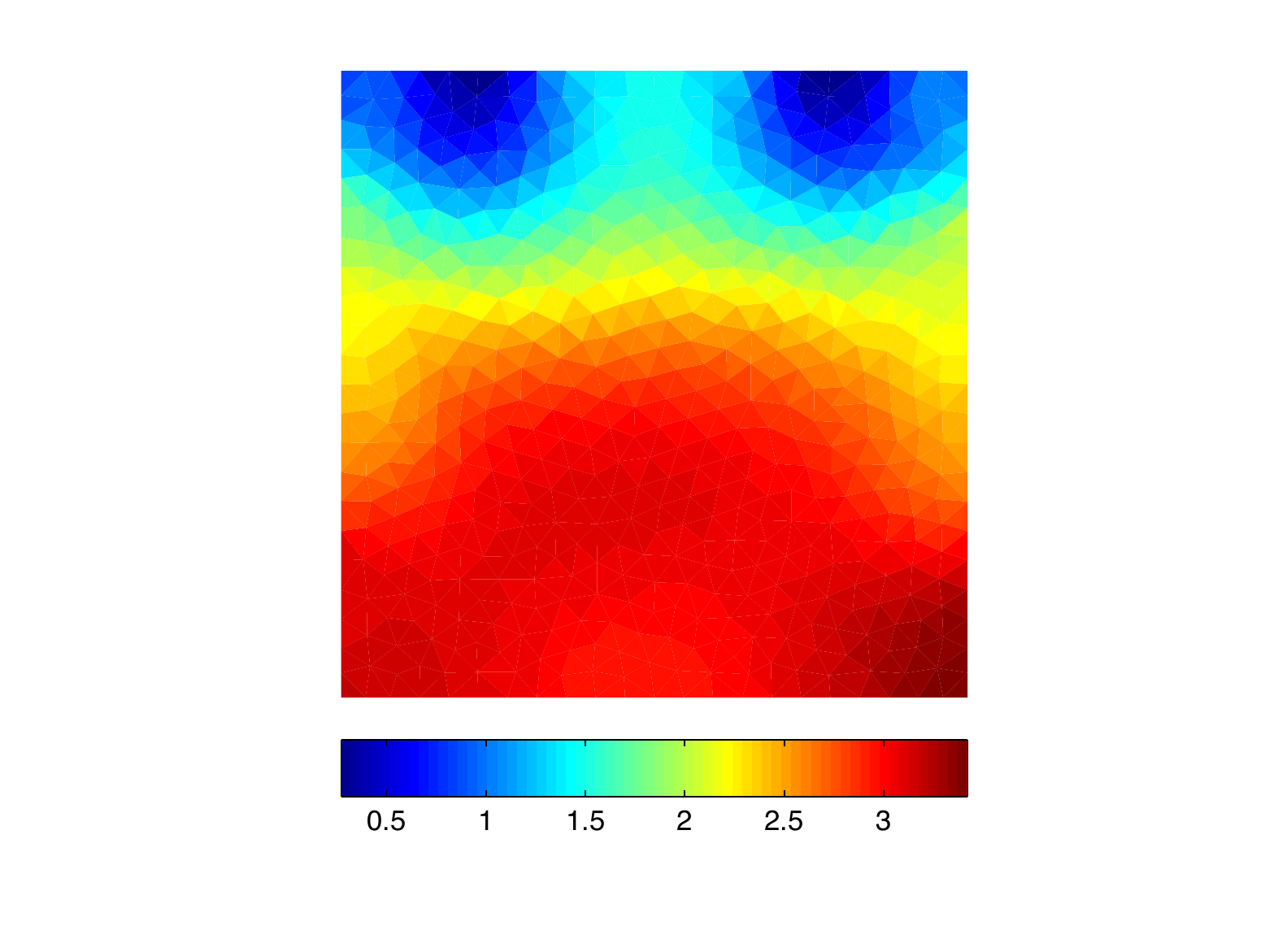,height=4cm,width=5.33cm} \hspace*{-20pt} \epsfig{file=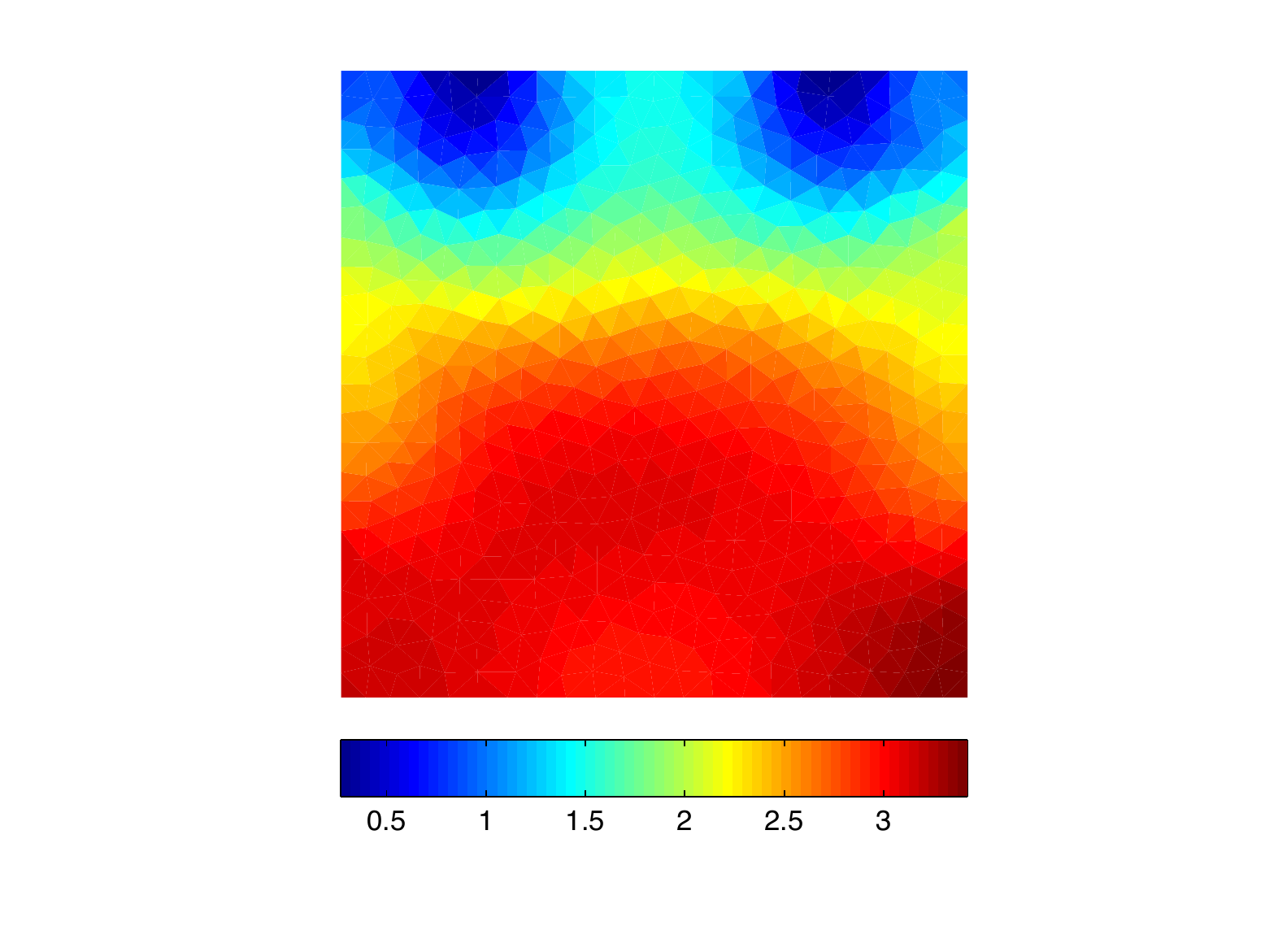,height=4cm,width=5.33cm}  \hspace*{-20pt} \epsfig{file=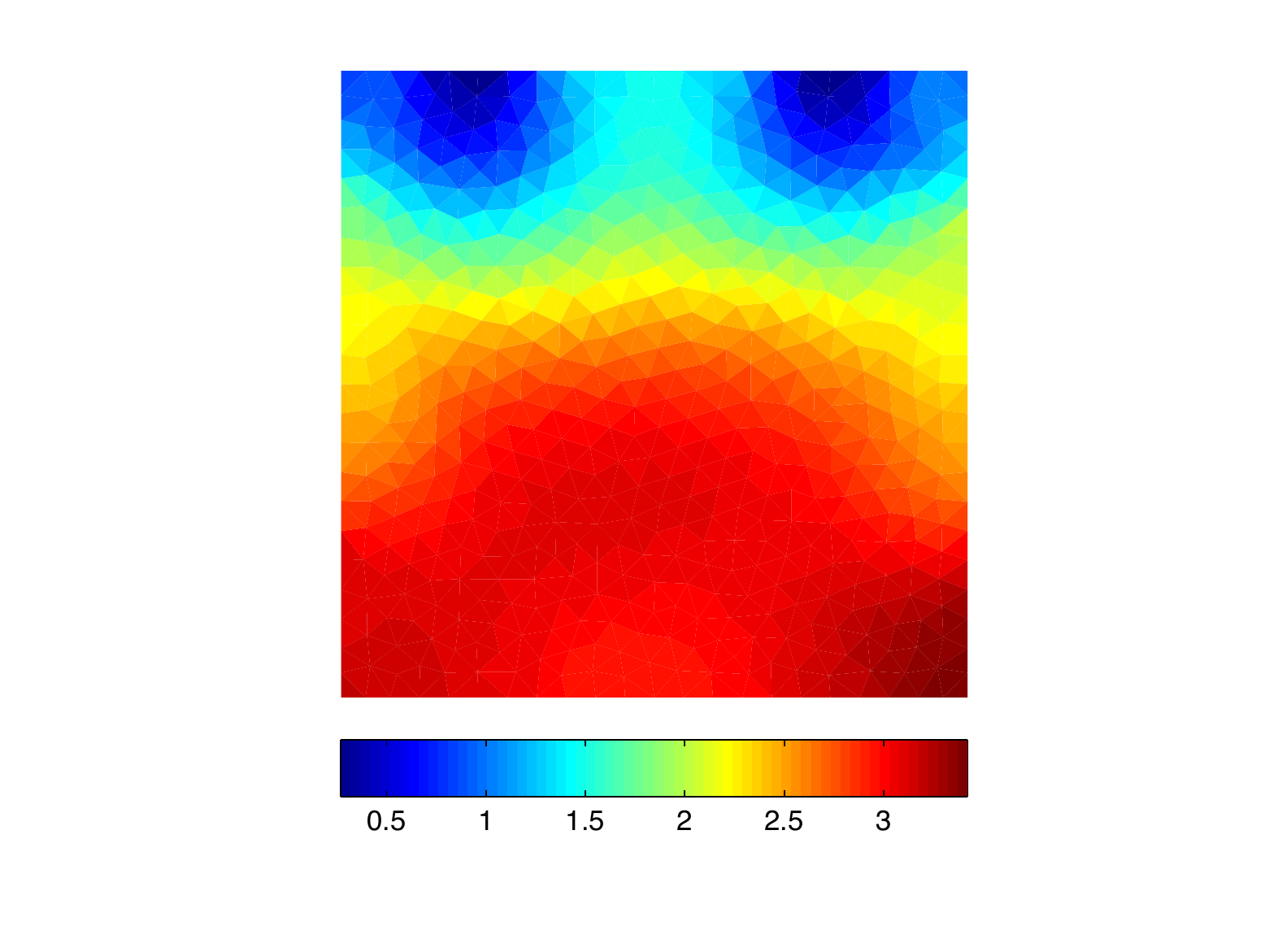,height=4cm,width=5.33cm}
\caption{At the top, the simulated target scaled permittivity $\omega \epsilon^*$ profile on $B_f$, as used in the first test example and the arrangement of the electrodes. In the second row, from left to right, the respective images resulted from first exterior iteration using $\alpha = 5 \times 10^{-6}$, namely the imaginary components of $\gamma_0 + \tau \delta \gamma_1$, $\gamma_0 + \tau \delta \gamma_2$, and $\gamma_1$ on $B_i$. Similarly at the bottom row,  the respective images from the second exterior GN iteration,  $\gamma_1 + \tau \delta \gamma_1$, $\gamma_1 + \tau \delta \gamma_2$, and $\gamma_2$, using the same value of $\alpha$.}
\label{fig5}
\end{figure}

To reconstruct the admittivity function we implement the proposed
iteration  (\ref{mainit}) using a precision matrix
$\mathbf{C}^{-1}_\gamma = \mathbf{R}'\mathbf{R}$, where
$\mathbf{R} \in \mathbb{R}^{N \times N}$ is a smoothness enforcing
operator. In  the numerical experiments we use two different
values of the regularization parameter in order to investigate the
performance of the scheme at different levels of regularization.
Using $\alpha = 5 \times 10^{-4}$ and $\alpha = 5 \times 10^{-6}$,
we execute two exterior GN iterations each one incorporating two
inner Newton iterations after which the algorithm converged to an
error value just above the noise level. The error reduction is
illustrated by the graphs of figures \ref{fig2} and \ref{fig2b},
showing a significant reduction in both the misfit error $Q(\delta
\gamma_p)$ and the image error $\|\delta \gamma_q^* - \Pi \delta
\gamma_p\|$ respectively for $p=0,1,2$ for the first and second GN
iterations, i.e. $q=1,2$. Each exterior iteration was initialized
with $\delta \gamma_0 = 0$ hence we can regard $\delta \gamma_1$
as the Tikhonov solution (\ref{tik2}) and $\delta \gamma_2$ the
quadratic regressor after two iterations (\ref{mainit}). To aid
convergence a backtracking line search algorithm was used where
the optimal step sizes for each iteration, interior and exterior.
The computational time required to assemble the Jacobians
$\mathbf{J} \in \mathbb{C}^{390 \times 1038}$ was about 0.34 s,
while each of the 390 matrices $\mathbf{K}^k \in \mathbb{C}^{1038 \times
1038}$ took about 4.75 s and then each iteration about 12
s depending on the line search. These times are based on running
Matlab \cite{matlab} on a machine with a dual core processor at
2.53 GHz. Despite the substantial computational overhead, the
method can be appealing in the cases where the inverse problem is
heavily underdetermined with only a few measurements. Moreover,
the assembling of the $\mathbf{K}$ matrices is well suited for
parallel processing. The images of the reconstructed admittivity
perturbation at each iteration are plotted in figures \ref{fig4}
(real component) and \ref{fig5} (imaginary component) below their
respective target images for comparison. As the error graphs
clearly indicate, the reconstructed images show a profound
quantitative improvement in spatial resolution, with the
regularized quadratic regression solution $\delta \gamma_2$ to
outperform the Tikhonov solution $\delta \gamma_1$ in both
Gauss-Newton iterations. Notice however, that in the exterior
iteration we scaled the increment  $\delta \gamma_p$ by $\tau_q$
in order to preserve the positivity new admittivity estimate. This
scaling, if $\tau_q < 1$ tends to increase the data misfit errors,
hence one can observe some discontinuities in the error reduction
from $q=1$ to $q=2$ in the plots of figure \ref{fig2}. Similarly
for the graphs of the image error in figure \ref{fig2b}, although
this time the correction works out to the improvement of the
errors as the target images are by definition positive. For
completeness, the step sizes used in the image reconstructions of
figures \ref{fig4} and \ref{fig5} are  $\tau_{p=1}=1$,
$\tau_{p=2}=0.3$, and $\tau_{q=1}=0.78$ for the first cycle of
iterations and $\tau_{p=1}=1$, $\tau_{p=2}=0.38$, and
$\tau_{q=1}=1$ for the second.

As a second example we consider a purely conductive case, i.e.
$\omega = 0$, aiming to reconstruct the target conductivity
function appearing at the top of figure \ref{fig8}. Once again
synthetic data are simulated, using the same current and
measurement patterns as in the previous case. After computing the
measurements $\zeta$ and introducing some zero mean Gaussian noise
using the noise covariance covariance matrix $\mathbf{C}_\eta =
10^{-5} \max |\zeta|\, \mathbf{I}$ we formulate the inverse
problem at a homogeneous background conductivity $\sigma_0$, the
best homogeneous fit of the data, regularization matrix
$\mathbf{R}$, and $\alpha$ parameters equal to $10^{-5}$ and
$10^{-7}$. To aid comparison with the convergent results for the
complex admittivity case we implement the algorithm for two
interior and two exterior iterations, for each of the regularized
problems. The graphs of the data misfit and image errors are
illustrated in figures \ref{fig6} and \ref{fig7}. The graphs show
a convergence pattern similar to that of the complex case, for
both values of the regularization parameter. Also at the initial
reference point the linear and quadratic data misfit functions
obtain values $\Lambda(\delta \sigma_0) = 0.24$ and $Q(\delta
\sigma_0)  = 0.09$,  demonstrating once again that the
contribution of the quadratic term can be significant if the
reference point is not sufficiently close to the solution. In
terms of its computational cost, implementing the algorithm for
the purely real admittivity has brought the processing time to
about a half of that consumed for the complex case. The
reconstructed images presented in figure \ref{fig8}, correspond to
the various conductivity updates as computed for two exterior and
two interior iterations, with $\alpha = 10^{-7}$. Initializing
with $\delta \sigma_0 = 0$ the second row, from left to right,
shows the conductivity updates after each interior iteration for
the first exterior GN iteration, and the bottom row the respective
images from the second exterior iteration. The step sizes used in
these results, as computed by the line search algorithms are
$\tau_{p=1}=1$, $\tau_{p=2}=0.3$, and $\tau_{q=1}=0.62$ for the
first cycle of iterations and $\tau_{p=1}=0.62$,
$\tau_{p=2}=0.24$, and $\tau_{q=1}=1$ for the second. Moreover at
the beginning of the second GN iteration the misfit functions have
been computed at $\Lambda(\delta \sigma_1) = 0.034$ and $Q(\delta
\sigma_1)  = 0.007$.

When the noise level in the data is approximately known, solving
the nonlinear EIT problem  one typically performs a number of GN
iterations until convergence is reached in the sense of the
discrepancy principle \cite{Lechleiter}, \cite{vauhkonen}. In our
results we implement only two exterior GN iterations, i.e.
$q=1,2$, each one encompassing two interior iterations, in order
to demonstrate the observed reduction in the image and data misfit
errors. Consequently, by virtue of the convergence properties of
the Newton algorithm, it is straightforward to state that the
quadratic regression solution will sustain a smaller error for any
number of GN iterations \cite{Bakushinsky}, and will thus converge
to the solution faster. On the other hand, a serious bottleneck of
the second, respectively higher-order, formulation is the
computational demand to compute the $\mathbf{K}$ matrices. In this
sense the method is more suited to the cases where high
performance computing is available, or when the number of data $m$
is fairly small.

\begin{figure}
\begin{center}
\epsfig{file=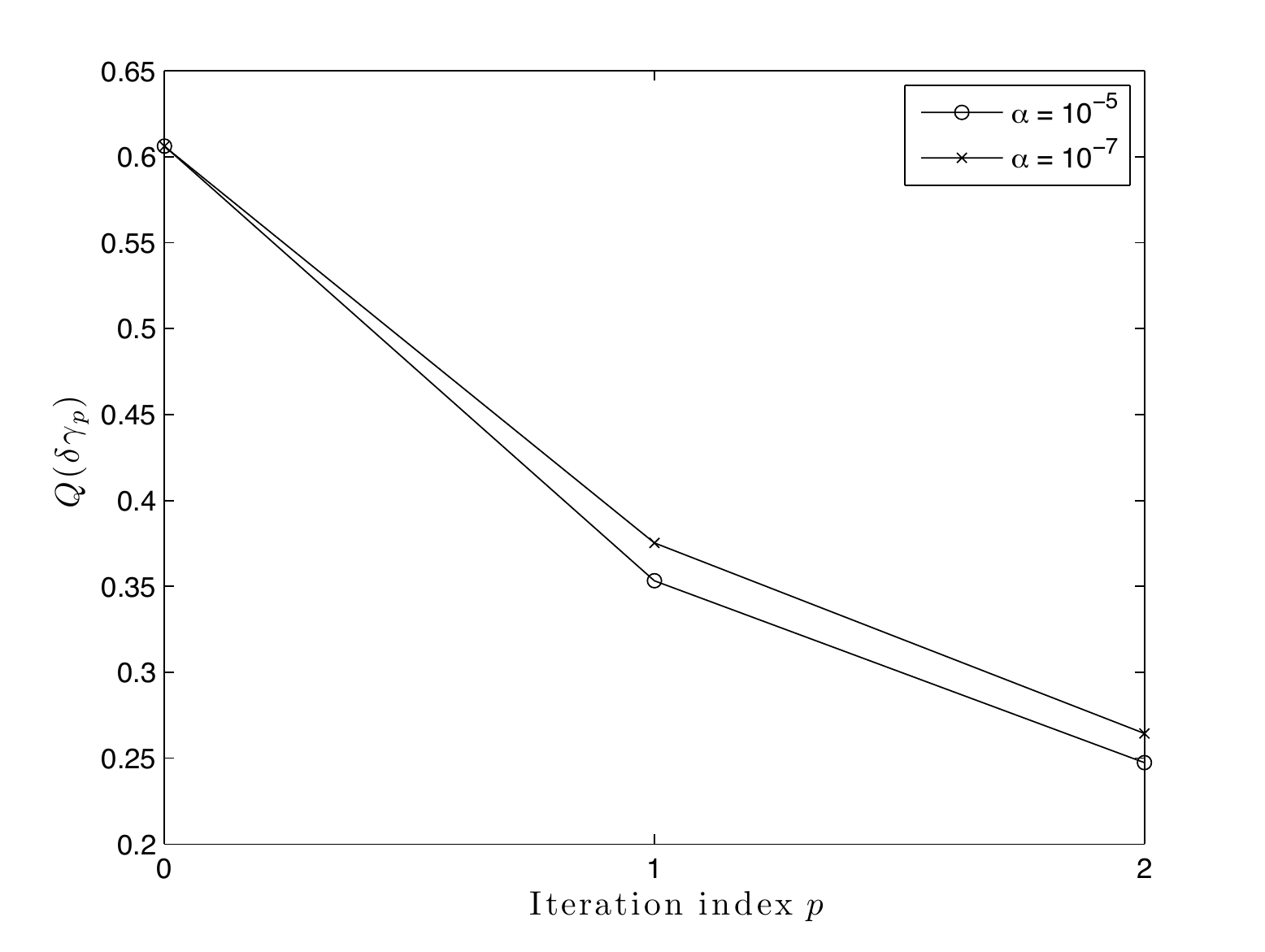,height=5cm,width=7cm} \epsfig{file=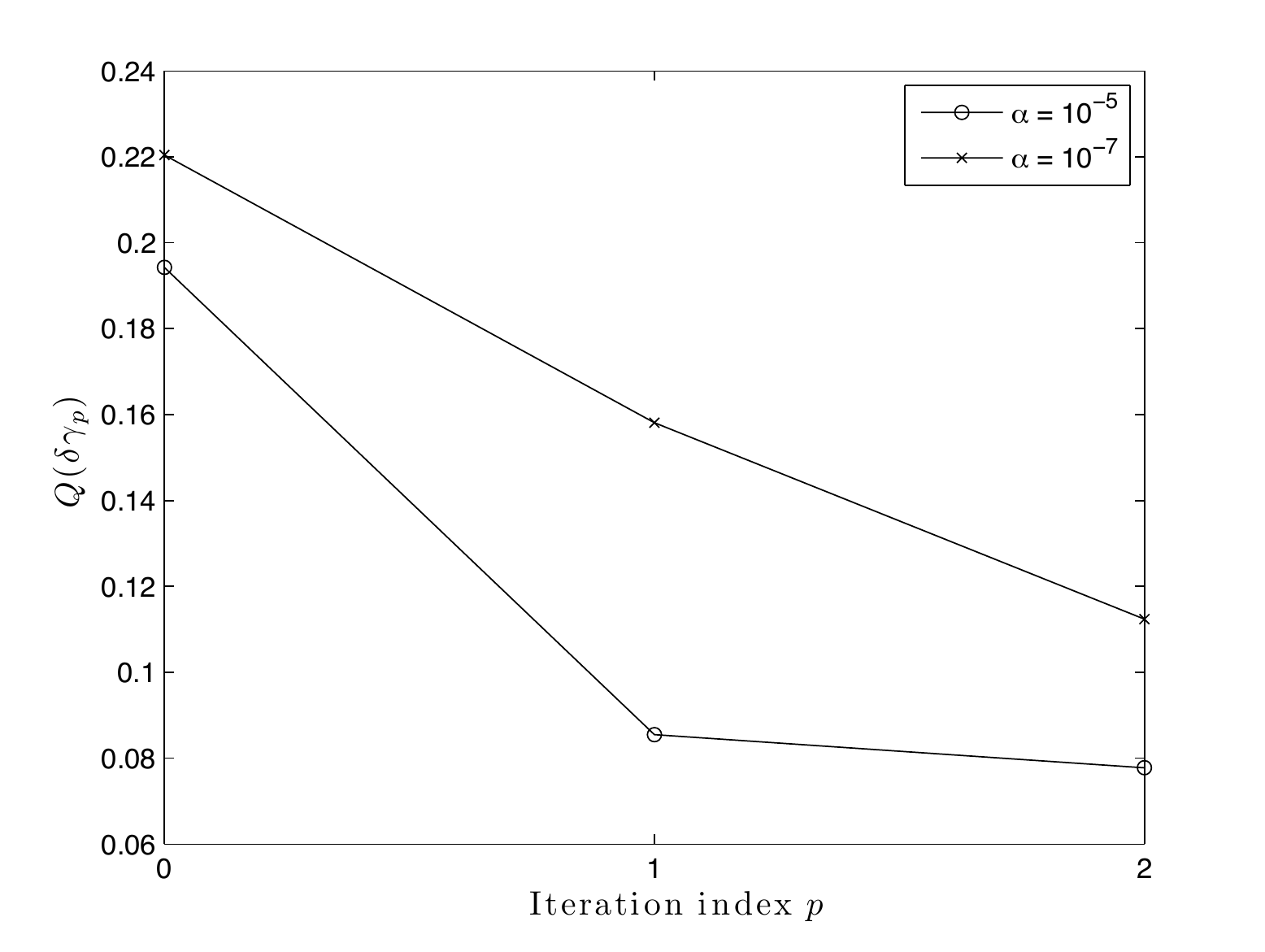,height=5cm,width=7cm}
\caption{Indicative convergence of the proposed method, in terms of minimizing the quadratic misfit error $Q(\delta \gamma_p)$  for two different values of the regularization parameter $\alpha$. Results are from the second test case with simulations at dc conditions $\omega=0$, hence the admittivity is purely real, i.e. $\gamma=\sigma$. Left the results during the first exterior iteration $q=1$, and right the corresponding values for $q=2$. In these results, $\delta \gamma_0=0$, $\delta \gamma_1$ coincides with the Tikhonov solution, and $\delta \gamma_2$ is the regularized quadratic regression solution. Notice that the quadratic regression solution has lower data misfit errors in both GN iterations. Between the first and second exterior iteration the admittivity increment was scaled to preserve positivity, hence the apparent discontinuity in the error reduction.}
\label{fig6}
\end{center}
\end{figure}

\begin{figure}
\begin{center}
\epsfig{file=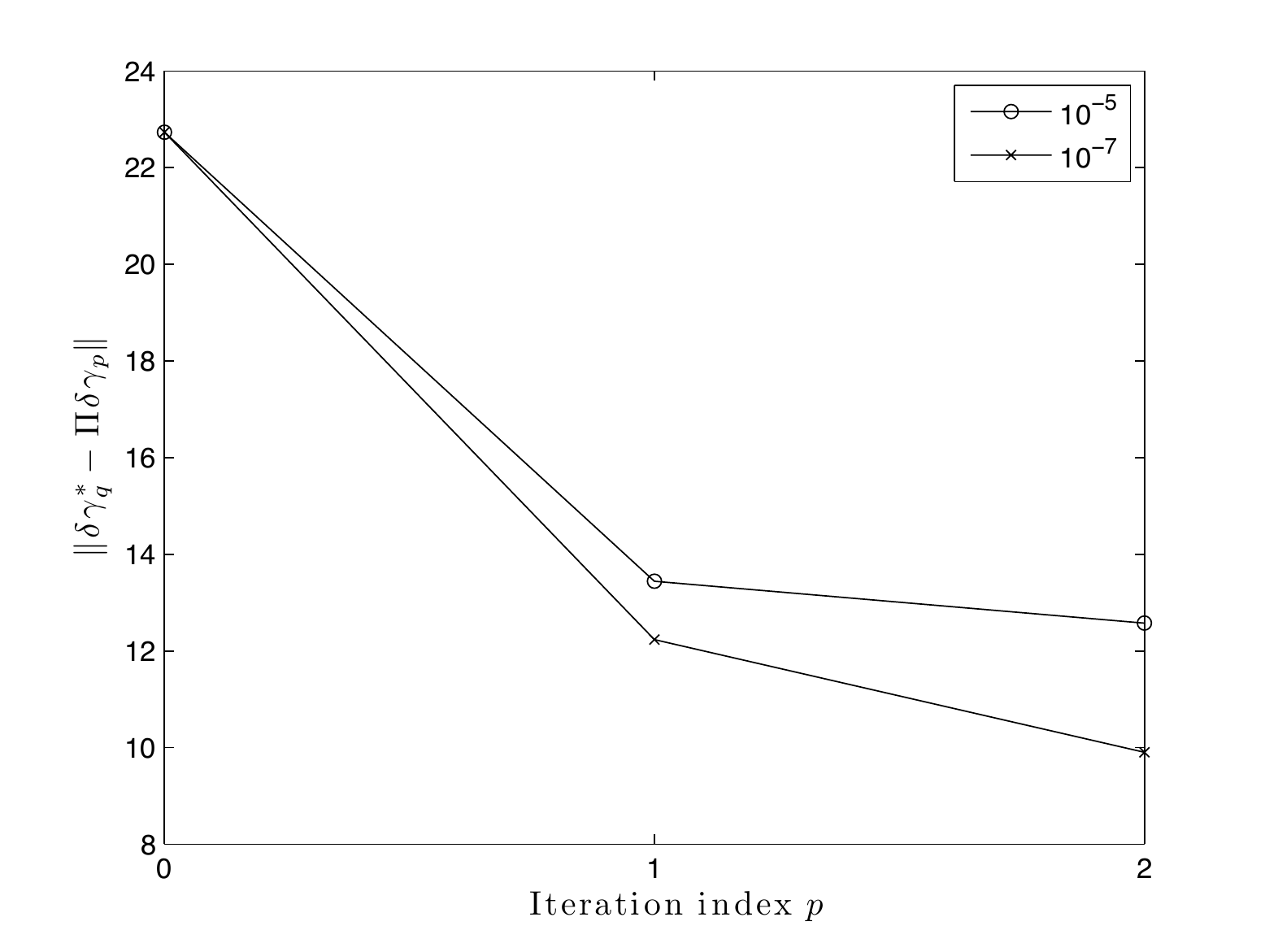,height=5cm,width=7cm} \epsfig{file=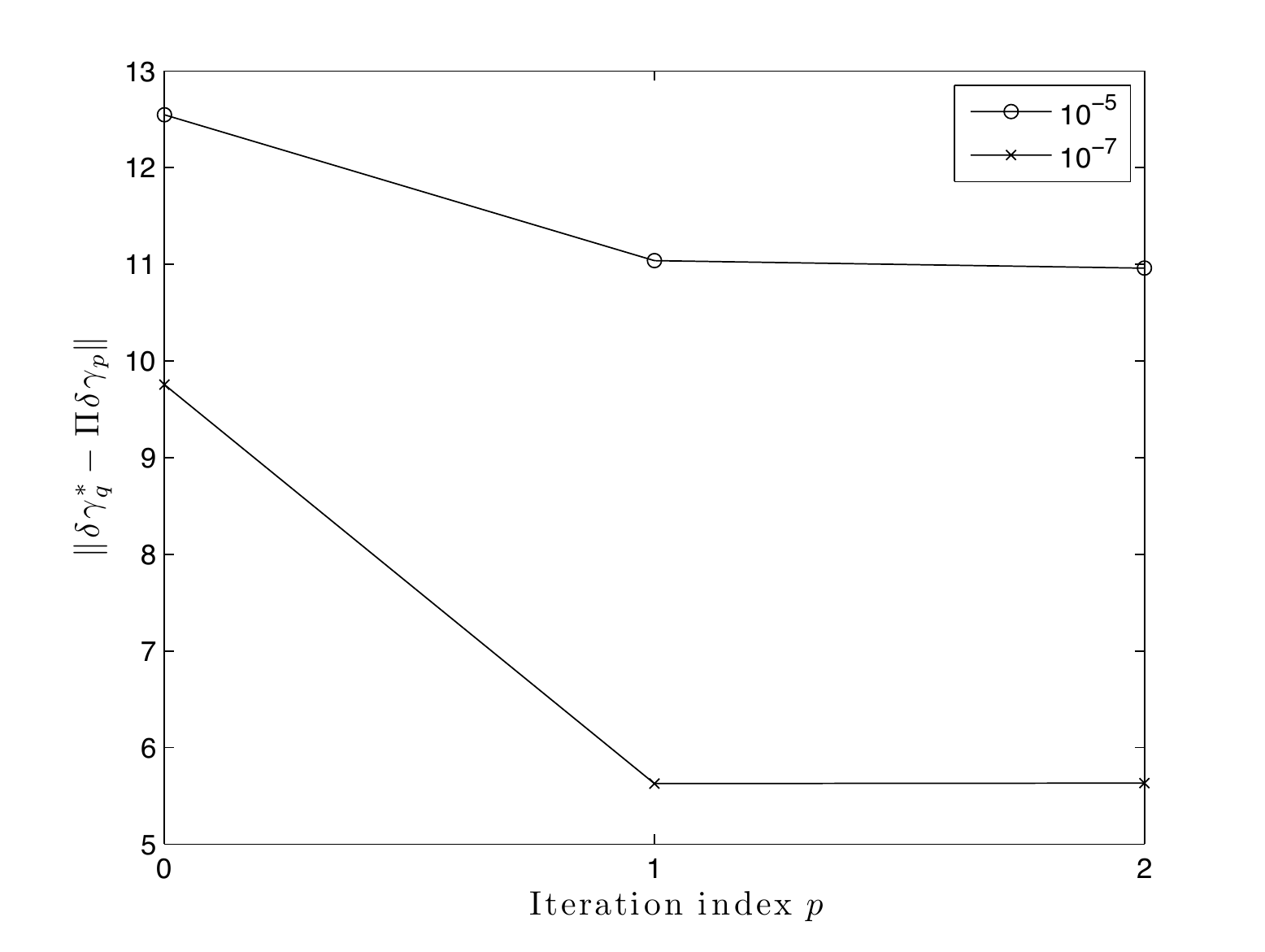,height=5cm,width=7cm}
\caption{Indicative convergence of the proposed method, in terms of minimizing the the image error $\|\delta \gamma_q^* - \Pi \delta \gamma_p\|$  at each interior iteration for two different values of the regularization parameter $\alpha$. Results are from the second test case with simulations at dc conditions $\omega=0$, hence the admittivity is purely real, i.e. $\gamma=\sigma$. Left the results during the first exterior iteration $q=1$, and right the corresponding values for $q=2$. In the figures $\delta \gamma_0$ is the initial homogeneous guess, $\delta \gamma_1$ coincides with the Tikhonov solution, and $\delta \gamma_2$ is the regularized quadratic regression solution. Between the first and second exterior iteration the admittivity increment was scaled to preserve positivity, hence the apparent discontinuity in the error reduction.}
\label{fig7}
\end{center}
\end{figure}

\begin{figure}
\centering \epsfig{file=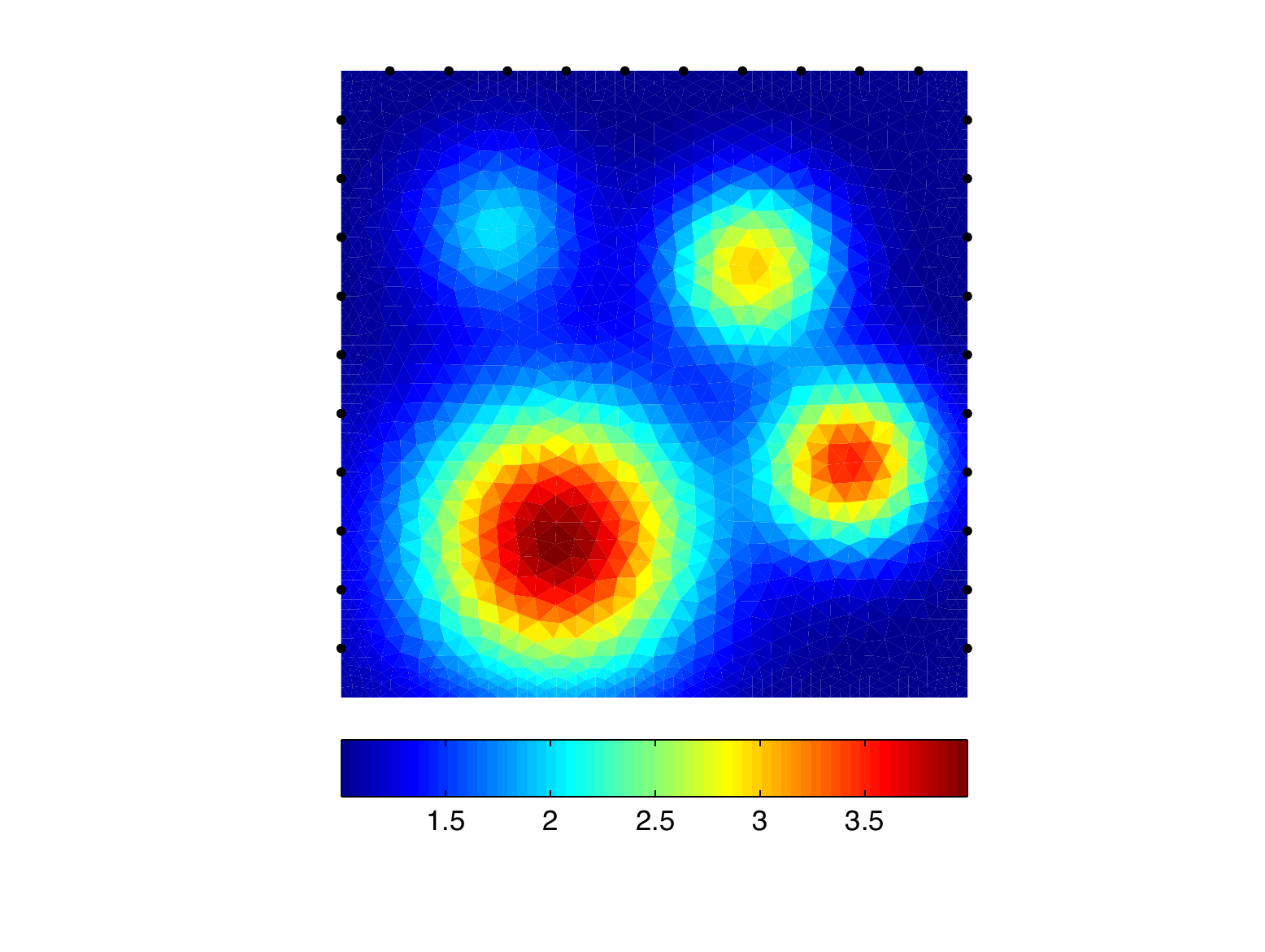,height=4cm,width=5.33cm}\\
\epsfig{file=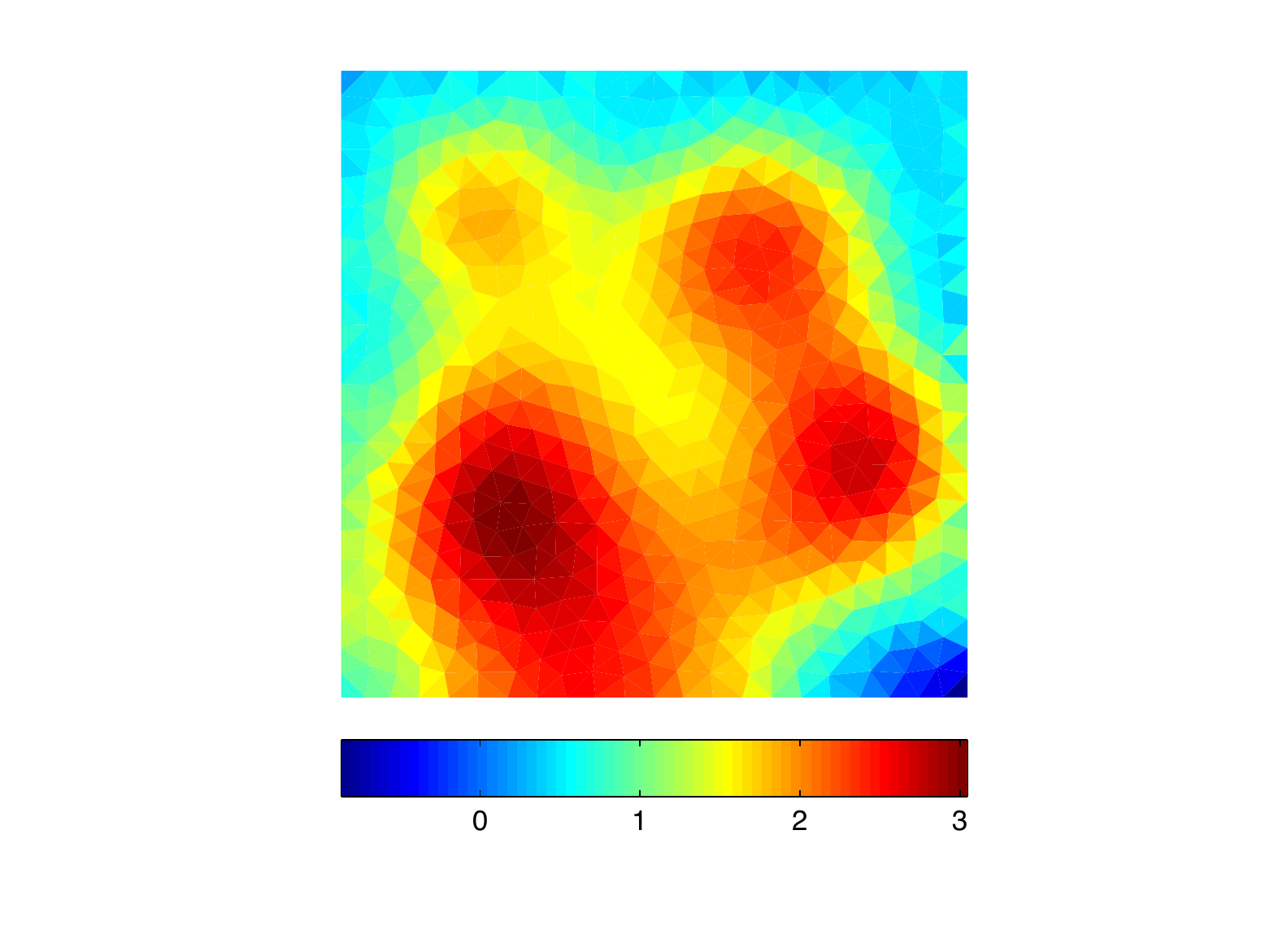,height=4cm,width=5.33cm} \hspace*{-20pt}\epsfig{file=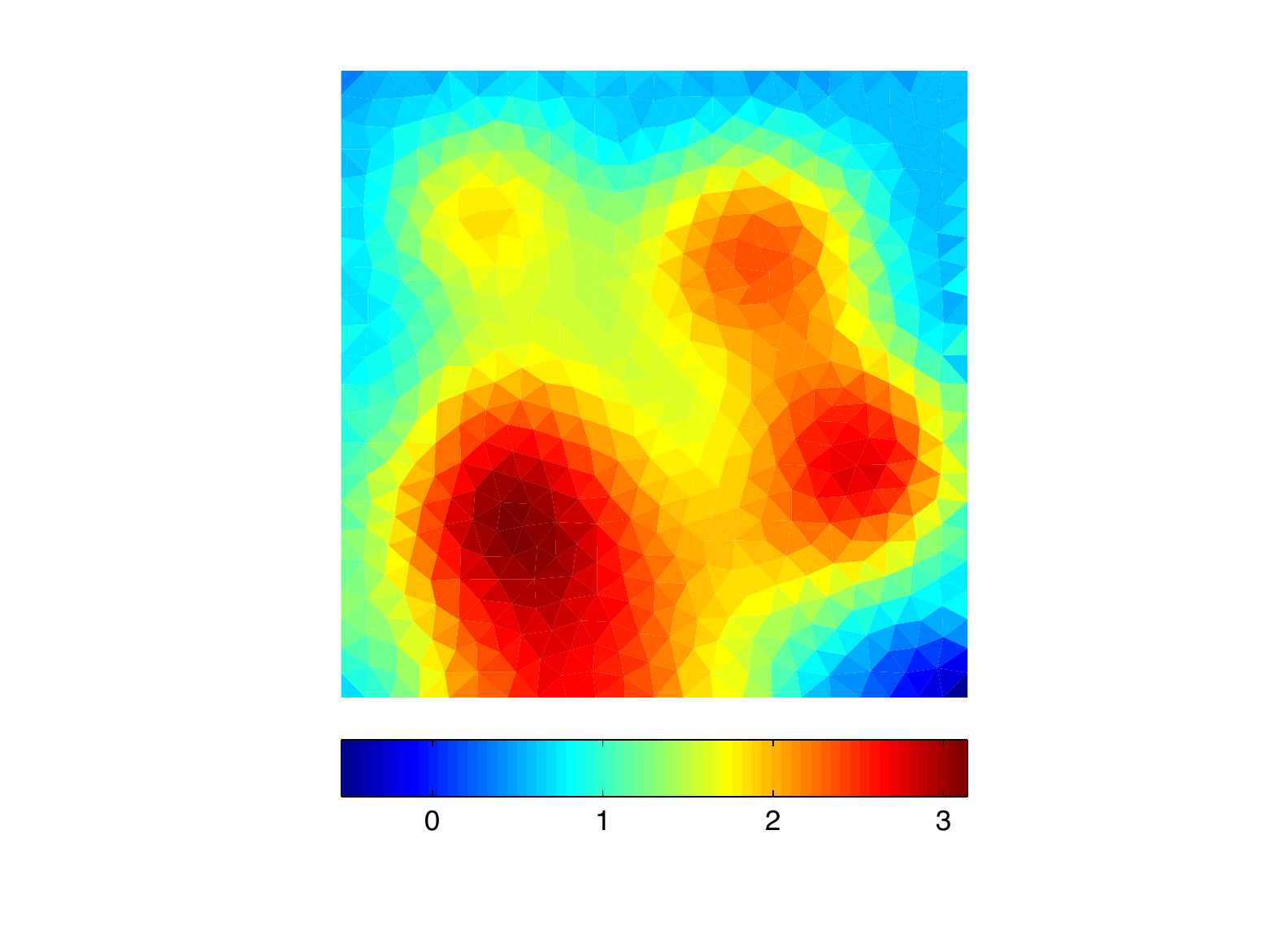,height=4cm,width=5.33cm} \hspace*{-20pt} \epsfig{file=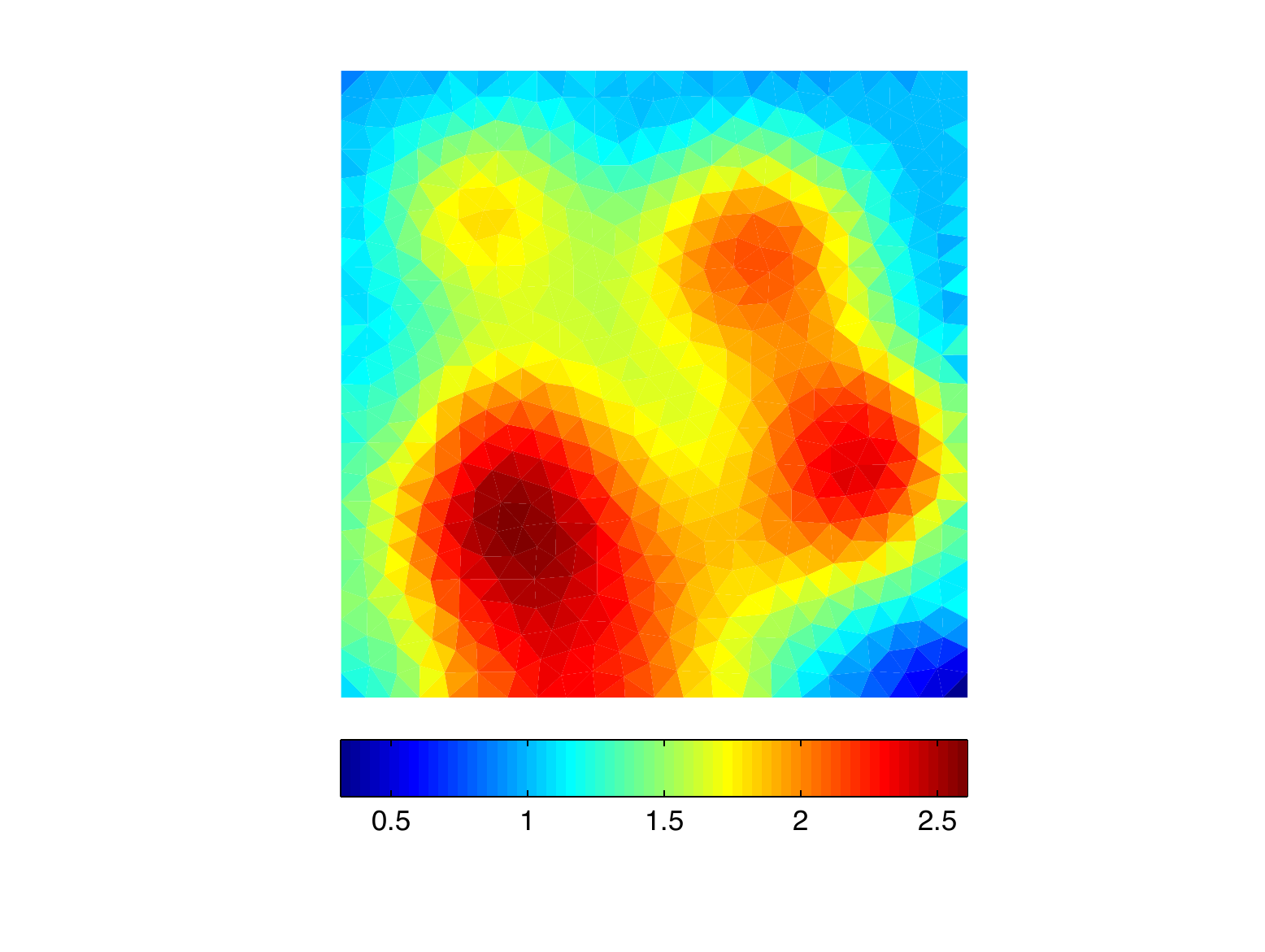,height=4cm,width=5.33cm}\\
\epsfig{file=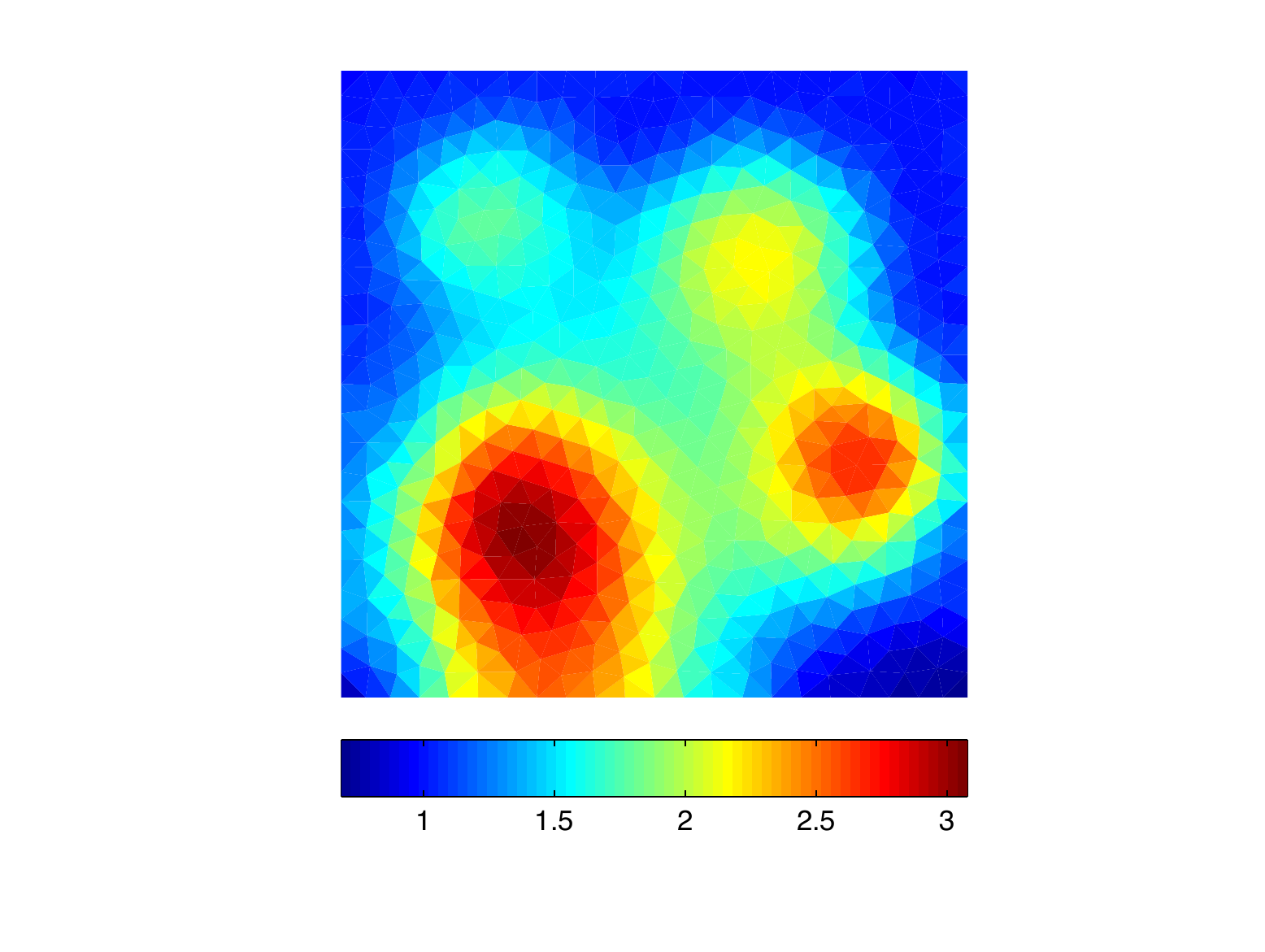,height=4cm,width=5.33cm} \hspace*{-20pt} \epsfig{file=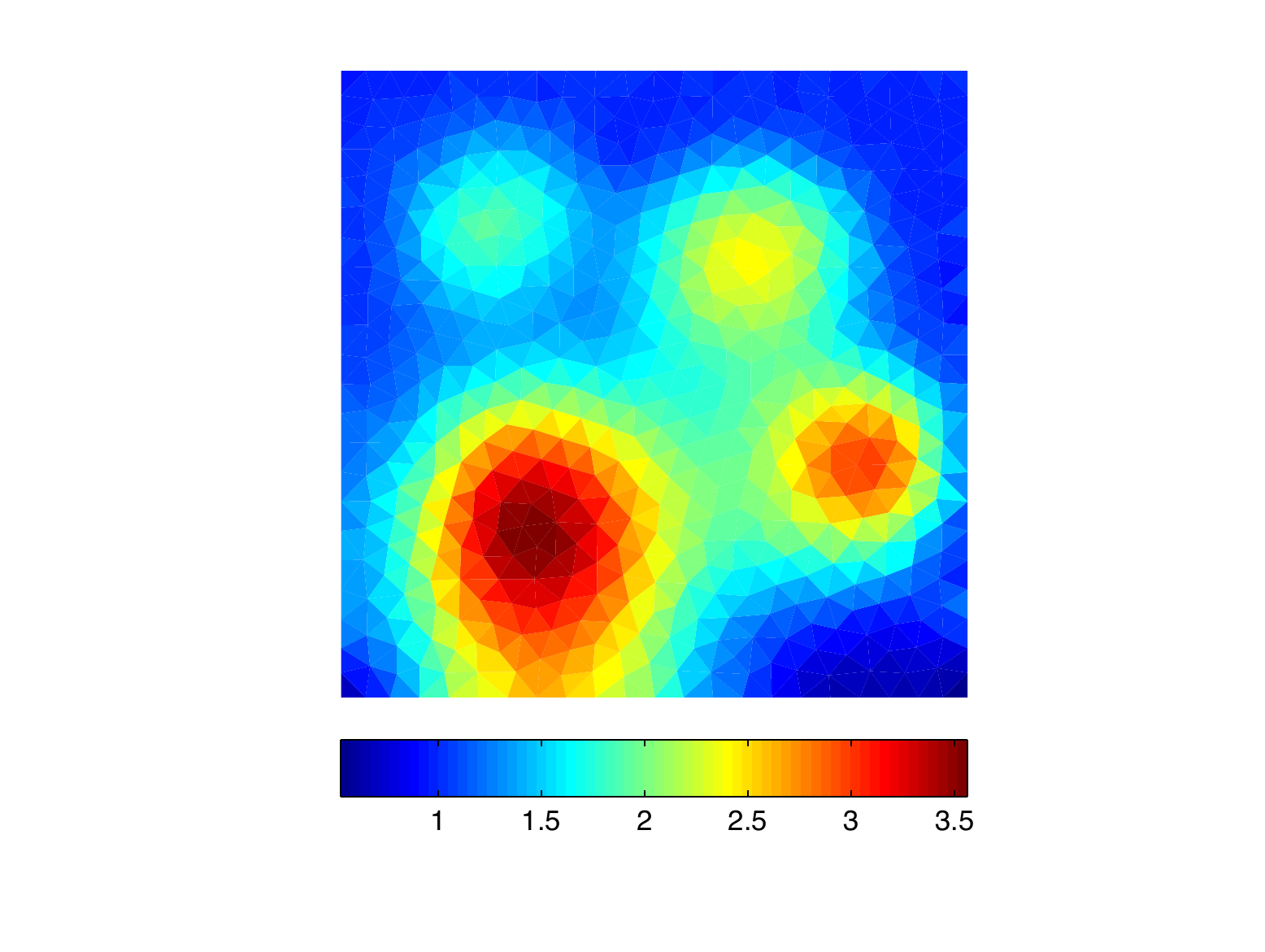,height=4cm,width=5.33cm}  \hspace*{-20pt} \epsfig{file=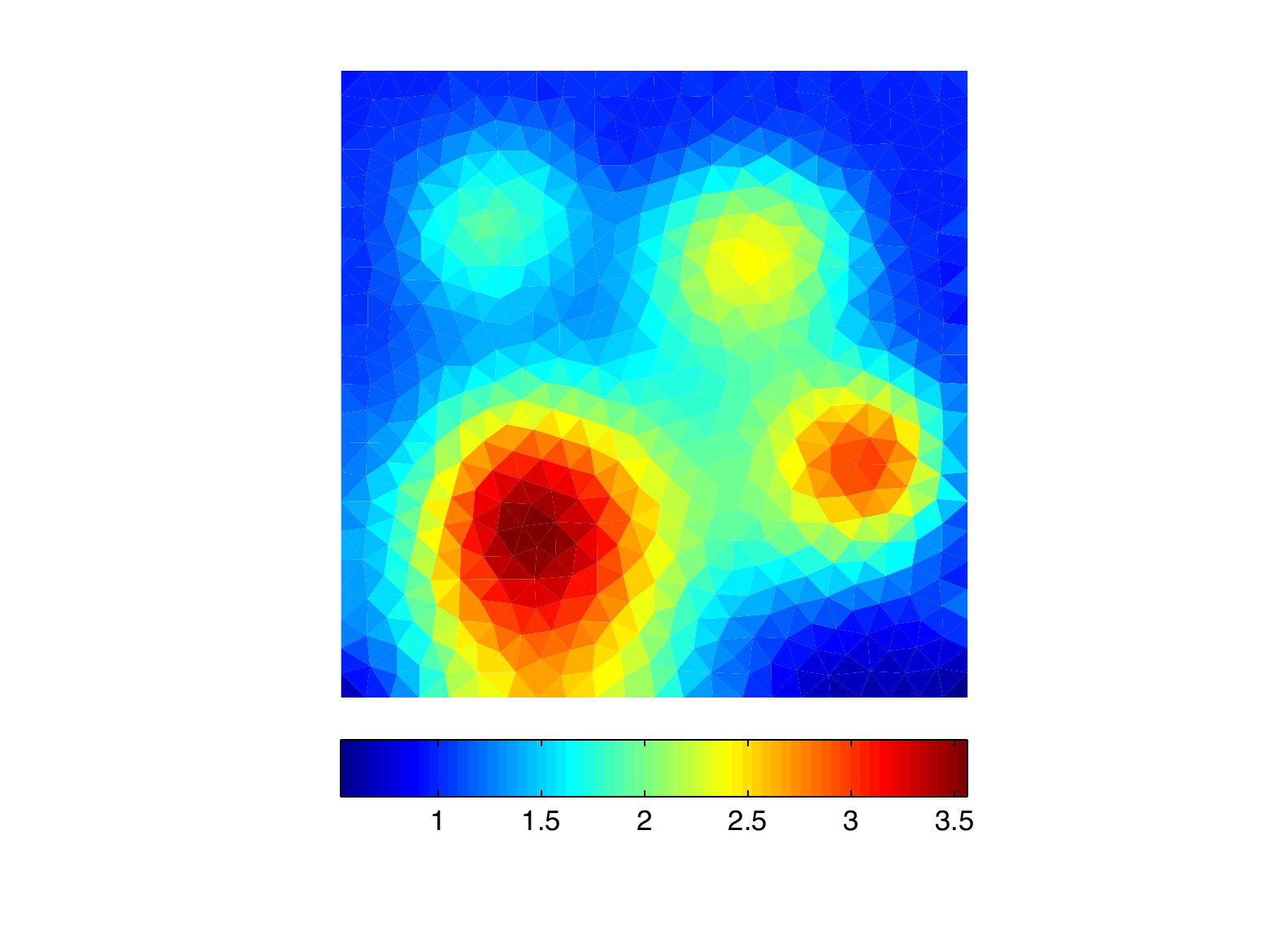,height=4cm,width=5.33cm}
\caption{Simulated and reconstructed admittivity functions, for the second test case at direct current conditions. Top row, the target simulated conductivity $\sigma^*$ discretized in $B_f$. Below from first exterior iteration using $\alpha = 10^{-7}$, namely $\sigma_0 + \tau \delta \sigma_1$ , $\sigma_0 + \tau \delta \sigma_2$, and $\sigma_1$ on $B_i$. Similarly at the bottom row, the respective images from the second exterior GN iteration, $\sigma_1 + \tau \delta \sigma_1$, $\sigma_1 + \tau \delta \sigma_2$, and $\sigma_2$, using the same value of $\alpha$.}
\label{fig8}
\end{figure}

\section{Conclusions}\label{conclusions}
This paper proposes a new approach for the inverse impedance
tomography problem. Based on a a power perturbation approach we
derive a nonlinear integral transform relating changes in
electrical admittivity to those observed in the respective
boundary measurements. This transform was then modified by
assuming that the electric potential in the interior of a domain
with unknown electrical properties can be approximated by a
first-order Taylor expansion centered at an a priori admittivity
estimate. This framework yields a quadratic regression problem
which we then regularized in the usual Tikhonov fashion.
Implementing Gauss-Newton's iterative algorithm we demonstrate
that the method quickly converges to results that outperform those
typically computed by applying the algorithm on the linearized
inverse problem. An important shortcoming of this approach is the
computational cost of computing the second and higher derivatives,
as they require the assembly of large dense matrices of dimension
equal to that of the parameter space. A possible remedy to this
can be found in model reduction methods \cite{lwg}. Another
interesting extension is to consider a reformulation of the
inverse problem in terms of some surrogate parameter functions,
e.g. the logarithm of the admittivity, in a way that preserves the
necessary positivity on the electrical parameters.

\section*{Appendix}

Here we present an alternative approach to the derivation of (\ref{nlt}) suggested to us by one of the anonymous reviewers.  Rather than relying on the conservation laws as was the case for our approach, the one presented below is based more on the use of variational methods applied to both the forward and adjoint problems. From
the weak form (\ref{weakv}) for $\psi = \bar v(\gamma_p,I^m)$ and
$\Psi = \overline{V_\ell}(\gamma_p,I^m)$ assuming $u(\gamma,I^d)$,
$U_\ell(\gamma,I^d)$ then
\begin{align*}
& \int_B \mathrm{d}x \, \gamma \nabla u (\gamma,I^d) \cdot \nabla
\bar v(\gamma_p,I^m) \\& + \sum_{\ell=1}^L z_\ell^{-1}
\int_{\Gamma_{e_\ell}}\mathrm{d}s \bigl ( u(\gamma,I^d) -
U_\ell(\gamma, I^d) \bigr ) \bigl ( \bar v(\gamma_p,I^m) -
\overline{V_\ell}(\gamma_p,I^m)\bigr ) = \sum_{\ell=1}^L I^d_\ell
\overline{V_\ell}.
\end{align*}
Repeating for a model with  $u(\gamma_p,I^d)$,
$U_\ell(\gamma_p,I^d)$ yields
\begin{align*}
&\int_B \mathrm{d}x \, \gamma_p \nabla u (\gamma_p,I^d) \cdot
\nabla \bar v(\gamma_p,I^m) \\& + \sum_{\ell=1}^L z_\ell^{-1}
\int_{\Gamma_{e_\ell}}\mathrm{d}s \bigl ( u(\gamma_p,I^d) -
U_\ell(\gamma_p, I^d) \bigr ) \bigl ( \bar v(\gamma_p,I^m) -
\overline{V_\ell}(\gamma_p,I^m)\bigr ) = \sum_{\ell=1}^L I^d_\ell
\overline{V_\ell},
\end{align*}
and thus by subtracting and inserting $\pm \int_B \mathrm{d}x \,
\gamma_p \nabla u(\gamma, I^d) \cdot \nabla \bar v(\gamma_p,I^m)$
one arrives at
\begin{align*}
0 & = \int_B \mathrm{d}x (\gamma-\gamma_p) \nabla u(\gamma,I^d) \cdot \nabla \bar v(\gamma_p,I^m) + \int_B \mathrm{d}x \gamma_p \nabla \bigl (u(\gamma,I^d) - u(\gamma_p,I^d) \bigr ) \cdot \bar v (\gamma_p,I^m) \\
& +  \sum_{\ell=1}^L z_\ell^{-1} \int_{\Gamma_{e_\ell}}\mathrm{d}s
\bigl ( u(\gamma_p,I^d) - U_\ell(\gamma_p, I^d) \bigr ) \bigl (
\bar v(\gamma_p,I^m) - \overline{V_\ell}(\gamma_p,I^m)\bigr ).
\end{align*}
Similarly from the weak form of the adjoint problem assuming $\bar
v(\gamma_p,\overline{I^m})$ and
$\overline{V_\ell}(\gamma_p,\overline{I^m})$ for $\psi =
u(\gamma,I^d)$ and $\Psi_\ell = U_\ell(\gamma,I^d)$ and  $\psi =
u(\gamma_p,I^d)$ and $\Psi_\ell = U_\ell(\gamma_p,I^d)$ we get
\begin{eqnarray*}
\sum_{\ell=1}^L \bigl ( U(\gamma,I^d) & - & U(\gamma_p,I^d) \bigr )\overline{I_\ell^m} = \int_B \mathrm{d}x \gamma_p \nabla \bigl ( u(\gamma,I^d) - u(\gamma_p,I^d)\bigr ) \cdot \nabla \bar v(\gamma_p,\overline{I^m}) \\
& + & \sum_{\ell=1}^L z_\ell^{-1} \int_{\Gamma_{e_\ell}} \bigl (
u(\gamma,I^d) - u(\gamma_p,I^d) + U_\ell(\gamma,I^d) -
U_\ell(\gamma_p,I^d) \bigr )\bigl (\bar v(\gamma_p,\overline{I^m})
- \overline{V_\ell}(\gamma_p,\overline{I^m}) \bigr ),
\end{eqnarray*}
thus combining the last two relations yields the result
(\ref{nlt}).

\section*{Acknowledgment} The authors would like to thank the reviewers for their helpful comments and suggestions during the review process. NP acknowledges helpful discussions with Irene Moulitsas on the quadratic regression problem and the help of Bill Lionheart and Kyriakos Paridis who commented on earlier drafts. NP is grateful to the Cyprus Program at MIT and the Cyprus Research Promotion Foundation for the financial support of this work.

\end{document}